\documentclass{amsart}
\usepackage{mathpazo,amsmath, amssymb, amsthm, commath}
\usepackage[msc-links, initials]{amsrefs}
\usepackage[margin= 0.75 in, top = 1 in, bottom = 1 in]{geometry}

\usepackage{xcolor,appendix}

\usepackage{tabularx, graphicx, wrapfig, tikz, subcaption,pgfplots}
\pgfplotsset{compat=1.18}
\usepackage{hyperref}
\hypersetup{
    colorlinks = true,linkcolor=blue,citecolor=brown,
    linkbordercolor = {white}}
\usepackage[shortlabels]{enumitem}

\makeatletter
\def\@tocline#1#2#3#4#5#6#7{\relax
  \ifnum #1>\c@tocdepth 
  \else
    \par \addpenalty\@secpenalty\addvspace{#2}%
    \begingroup \hyphenpenalty\@M
    \@ifempty{#4}{%
      \@tempdima\csname r@tocindent\number#1\endcsname\relax
    }{%
      \@tempdima#4\relax
    }%
    \parindent\z@ \leftskip#3\relax \advance\leftskip\@tempdima\relax
    \rightskip\@pnumwidth plus4em \parfillskip-\@pnumwidth
    #5\leavevmode\hskip-\@tempdima
      \ifcase #1
       \or\or \hskip 1em \or \hskip 2em \else \hskip 3em \fi%
      #6\nobreak\relax
    \hfill\hbox to\@pnumwidth{\@tocpagenum{#7}}\par
    \nobreak
    \endgroup
  \fi}
\makeatother

\newcommand*\ticoord{\ensuremath{\vcenter{\hbox{\includegraphics[width=3em]{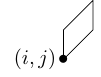}}}}}

\newcommand*\tiicoord{\ensuremath{\vcenter{\hbox{\includegraphics[width=3em]{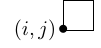}}}}}

\newcommand*\swap{\ensuremath{\vcenter{\hbox{\includegraphics[width=0.5em, angle = 90]{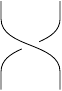}}}}}
\newcommand{\R}{{\mathbb R}}
\newcommand{\N}{{\mathbb N}}
\newcommand{\C}{{\mathbb C}}
\newcommand{\D}{{\mathbb D}}

\newcommand{\T}{{\mathbb T}}
\newcommand{\Z}{{\mathbb Z}}
\newcommand{\I}{{\mathbb I}}

\renewcommand{\P}{{\mathbb P}}
\newcommand{\Oo}{\mathcal{O}}

\newcommand{\mb}[1]{\mathbf{#1}}
\newcommand{\mc}[1]{\mathcal{#1}}

\newcommand{\re}{{\text{Re}}}
\newcommand{\im}{{\text{Im}}}

\DeclareMathOperator*{\res}{Res}

\newcommand{\Li}[1][2]{\mathrm{Li}_{#1}}

\newcommand{\ii}{{\mathrm i}}
\newcommand{\dd}{{\mathrm d}}
\newcommand{\ee}{{\mathrm e}}

\newmuskip\pFqskip
\mathchardef\pFcomma=\mathcode`, 

\newcommand*\pfq[6]{%
  \begingroup
  \begingroup\lccode`~=`,
    \lowercase{\endgroup\def~}{\pFcomma\mkern\pFqskip}%
  \mathcode`,=\string"8000
  {}_{#1}\phi_{#2}\biggl[\genfrac..{0pt}{}{#3}{#4};#5, #6\biggr]%
  \endgroup
}

\renewcommand{\deg}{\text{deg}}
\renewcommand{\arg}{\text{arg}}
\newcommand{\qbinom}{\genfrac{[}{]}{0pt}{}}

\renewcommand{\epsilon}{\varepsilon}
\renewcommand{\subset}{\subseteq}

\newcommand{\qandq}{\quad \text{and} \quad}
\newcommand{\qasq}{\quad \text{as} \quad }
\newcommand{\qforq}{\quad \text{for} \quad }

\newtheorem{lemma}{Lemma}[section]
\newtheorem{theorem}[lemma]{Theorem}
\newtheorem{rhp}[]{Riemann-Hilbert Problem}

\newtheorem{proposition}[lemma]{Proposition}

\theoremstyle{definition}
\newtheorem{remark}[lemma]{Remark}

\newtheorem{example}{Example}

\numberwithin{equation}{section}

\AtBeginEnvironment{remark}{%
  \pushQED{\qed}%
}
\AtEndEnvironment{remark}{\popQED\endexample}

\title{\bf The $q^{\mathrm{Volume}}$ lozenge tiling model via non-Hermitian orthogonal polynomials }
\author{A. Barhoumi}
\address[A.B.]{Department of Mathematics, Royal Institute of Technology (KTH), Stockholm, Sweden.  Email:  \texttt{ahmadba@kth.se}}
\author{M. Duits}
\address[M.D.]{Department of Mathematics, Royal Institute of Technology (KTH), Stockholm, Sweden.  Email:  \texttt{duits@kth.se}}
\curraddr{Sapienza University di Roma, Dipartimento di Matematica, Piazzala Aldo Moro no. 5, 0018 Roma, Italy.}
\thanks{The authors were supported by the European Research Council (ERC), Grant Agreement No. 101002013. This material is partially based upon work supported by the Swedish Research Council under grant no. 2021-06594 while the authors were in residence at Institut Mittag-Leffler in Djursholm, Sweden during the Fall of 2024}
\date{\today}



\begin{document}
\begin{abstract}
    We consider the $q^\text{Volume}$ lozenge tiling model on a large, finite hexagon. It is well-known that random lozenge tilings of the hexagon correspond to a two-dimensional determinantal point process via a bijection with ensembles of non-intersecting paths. The starting point of our analysis is a formula for the correlation kernel due to Duits and Kuijlaars which involves the Christoffel-Darboux kernel of a particular family of non-Hermitian orthogonal polynomials. Our main results are split into two parts: the first part concerns the family of orthogonal polynomials, and the second concerns the behavior of the so-called arctic curve. In the first half, we identify the orthogonal polynomials as a non-standard instance of little $q$-Jacobi polynomials and compute their large degree asymptotics in the $q \to 1$ regime. A consequence of this analysis is a proof that the zeros of the orthogonal polynomials accumulate on an arc of a circle and an asymptotic formula for the Christoffel-Darboux kernel. In the second half we use these formulas to show that, away from the boundary of the hexagon, the boundary of the liquid region converges to the Airy process in the sense of finite dimensional distributions. At inflection points of the arctic curve, we show that we do not need to subtract/add a parabola to the Airy line ensemble, and this effect persists at a distance $o(N^{-2/9})$ in the tangent direction. 
\end{abstract}
\maketitle 


\section{Introduction}

In this work, we will consider lozenge tilings of a regular hexagon as shown in Figure \ref{fig:tiling-a}. A simple shearing transformation produces tilings of the hexagon with corners as shown in Figure \ref{fig:tiling-b}. With a slight abuse of language, we will refer to both of these hexagons as $N \times N \times N$ hexagons. In these sheared coordinates the ``lozenges" are of the three types shown below, and we take the lozenges to have vertical and horizontal lengths 1:
\begin{figure}[ht!]
    \begin{subfigure}[b]{0.3\textwidth}
            \centering
            \begin{tikzpicture}[scale = 0.5]
            \draw[fill = red] (0,0) -- (0, 1) -- (1,2) -- (1,1) -- cycle;
        \end{tikzpicture}
        \caption{Type I tile}
    \end{subfigure}
    \begin{subfigure}[b]{0.3\textwidth}
        \centering
        \begin{tikzpicture}[scale = 0.5]
        \draw[fill = cyan] (0,0) -- (0, 1) -- (1,1) -- (1,0) -- cycle;
    \end{tikzpicture}
    \caption{Type II tile}
    \end{subfigure}
    \begin{subfigure}[b]{0.3\textwidth}
        \centering
        \begin{tikzpicture}[scale = 0.5]
        \draw[fill = yellow] (0,0) -- (1, 1) -- (2,1) -- (1,0) -- cycle;
    \end{tikzpicture}
    \caption{Type III tile}
    \end{subfigure}
    \label{fig:tiles}
\end{figure}

\noindent It is well-known (and pictorially evident) that lozenge tilings of the hexagon are in bijection with boxed plane partitions, i.e. arrangements of unit cubes in a cubic room of side length $N$. In view of this, one can assign to a tiling a volume which is given by the volume of the cubes in the corresponding boxed plane partition. We will be interested in studying random tilings of the hexagon where the probability of a tiling $\mc T$ is proportional to $q^{\text{Volume}}$, where $q \in \R_+$ is arbitrary. 

\begin{figure}[t]
    \begin{subfigure}[b]{0.49\textwidth}
        \centering
        \includegraphics[scale = 0.5]{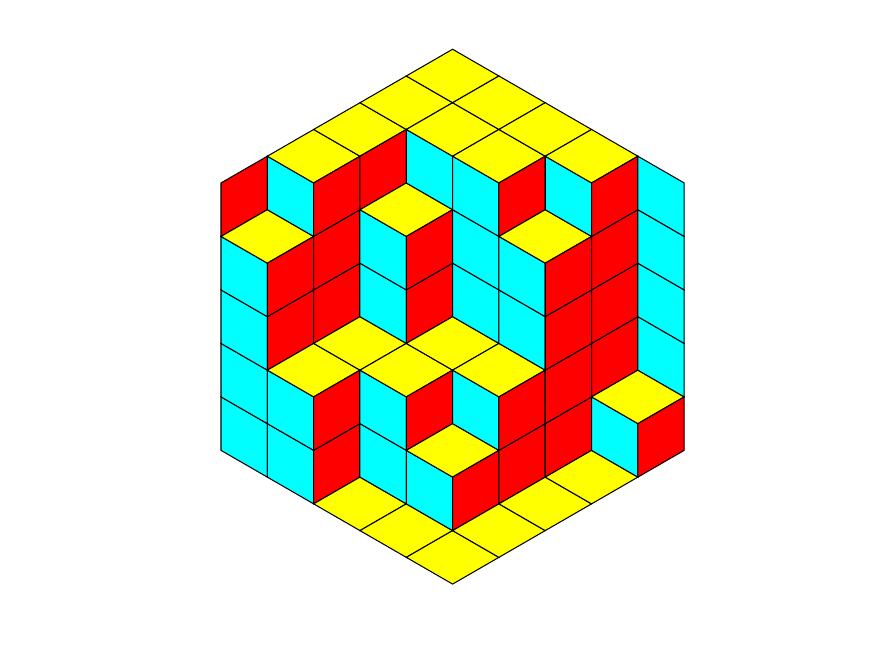}
    \caption{Symmetric hexagon}
    \label{fig:tiling-a}
    \end{subfigure}
    \begin{subfigure}[b]{0.49\textwidth}
        \centering
        \includegraphics[scale = 0.5]{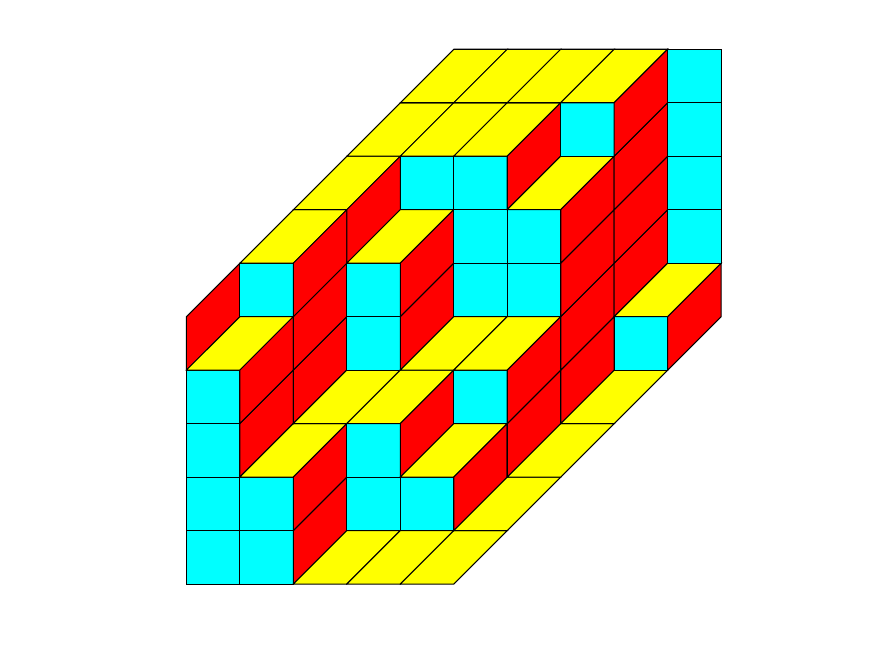}
        \put(-110,5){$(N, 0)$}
        \put(-180,5){$(0, 0)$}
        \put(-190, 90){$(0, N)$}
        \put(-125, 155){$(N, 2N)$}
        \put(-50, 155){$(2N, 2N)$}
        \put(-40, 65){$(2N, N)$}
    \caption{Sheared hexagon}
    \label{fig:tiling-b}
    \end{subfigure}
    \caption{\centering Sample tiling of the hexagon with $N = 5$. All tiling images were generated using code kindly provided by Christophe Charlier. }
    \label{fig:sample-tilings}
\end{figure}

At this point, there is a vast literature on random lozenge tilings, and we refer to \cites{MR4299268, MR2581882, MR3526828} as general introductions to the area. To place our work in context, we start from what is arguably the simplest tiling model of the hexagon: the uniform tiling model (where all tilings are equally likely). One of the early works on this is \cite{MR1641839}, where the authors described the so-called \emph{arctic circle}, a boundary which separates a \emph{liquid region} from a \emph{frozen region}, see Figure \ref{fig:arctic-a} below. This came on the heels of similar studies for domino tilings of the Aztec diamond \cites{JSP, MR1412441, MR1815214}. 

 The above works were succeeded by many studies of finer statistical properties of uniform lozenge tilings. One which is close in spirit to this work is \cite{MR2283089}, where the authors used a connection with Hahn polynomials, first discovered in \cites{MR1900323, MR2454474}, to obtain asymptotics of the one-point correlation function both in the interior of the liquid region and at the arctic circle. A more general line of research is uniform tilings of planar domains, for which there is a large literature and we mention only \cites{MR3278913, MR3298467, MR3861299,MR3413988, MR4105942,MR4660134,Aggarwal2021EdgeSF}. It is well-known that there is a bijection between lozenge tilings of the hexagon and dimer covers of certain bipartite graphs. This connection is also well-studied; we refer to \cite{MR2523460} for a general introduction. Uniform dimer models on general graphs and their limit shapes were studied in e.g. \cites{kuchumov, MR2215138}.  

Many generalizations of the uniform model have been proposed and analyzed. The $q^{\text{Volume}}$ model is one such generalization which reduces to the uniform measure when $q = 1$, and itself is a special case of $q$-deformations of the uniform measure introduced in \cite{BGR} and further generalized in \cite{MR3784910}. In \cite{BGR}, the authors propose a perfect sampling algorithm for generating these $q$-deformed tilings, provide an explicit equation for the arctic curve, and study the local behavior of the measures in the liquid region. A key object in the study of lozenge tilings is the \emph{height function}, which was first introduced by Thurston \cite{MR1072815}. In various tiling models, the height function is expected to concentrate near a deterministic function known as the \emph{limit shape}. This convergence was shown for the aforementioned $q$-deformed models using so-called loop equations in \cite{MR3944289} and using the connection with orthogonal polynomials in \cite{MR4808695}. For the particular case of $q^\text{Volume}$, this result had already been obtained in \cite{MR2358053}. In \cite{MR3944289}, it was shown that the fluctuations of the height function along vertical slices are Gaussian and conjectured that the fluctuation field converges to the so-called \emph{Gaussian Free Field} (GFF). This was later proved in \cite{MR4791420} using a two-dimensional version of the loop equations and in \cite{MR4808695} using methods first developed in \cites{breuer_central_2017,MR3785589}. The reductions of these results to the case of the uniform tiling model were obtained in \cites{MR3298467,MR3861715,MR3785589}.

The current writing is inspired by new developments on a generalization that has been particularly popular in recent years: that of random tilings of planar domains where the probability measure is defined by doubly periodic weights. These models exhibit very rich statistical properties, including the appearance of the so-called smooth disordered phase.  Various instances of these models were studied in the context of the Aztec diamond in, e.g., \cites{MR3846843, MR3479561, MR4228276, MR4260472}, and the recent pre-prints \cites{BB23, BB24}. There are fewer works on the doubly periodic tilings of the hexagon \cites{MR4206375, MR4124992}. One particular approach for studying such doubly periodic weighted tilings of the hexagon (or the Aztec diamond) is based on matrix valued orthogonal polynomials \cite{MR4228276}. Interestingly, this approach also sheds new light on less complicated models. Indeed, for uniform weights it provides an interesting connection to Jacobi polynomials with one negative parameter, which can be exploited for asymptotic studies as the size of the hexagon tends to infinity. In \cite{MR4124992}, this approach was successfully applied to a weight on lozenge tilings that is two periodic in one direction.  The approach is also applicable to the $q^{\text{Volume}}$ model (despite not being doubly periodic) and this is the starting point of the current writing.  We show that the $q^{\text{Volume}}$ model is related to little $q$-Jacobi polynomials (with nonstandard parameters). We will compute the asymptotic behavior of these polynomials as their degree tends to infinity, and further show how these asymptotics can be used to study random lozenge tilings for large hexagons. In principle, this will allow us to compute limiting correlation functions everywhere in the hexagon, but we will focus on the boundary of the liquid region. As far as we know, this boundary has not been analyzed in the literature before for $q\neq 1$. We will pay some special attention to the inflection points of the boundary, i.e. points where the curvature vanishes. 

\subsection{General set-up} 
A standard way of defining a probability  measure on all possible lozenge tilings of the hexagon goes as follows: we assign a weight $w(\mathcal T)$ to a given tiling $\mathcal T$ by 
\[
w(\mc T) := \prod_{\tiicoord \in \mc T} w\left( \tiicoord \right) \prod_{\ticoord \in \mc T} w\left( \ticoord \right),
\]
where each tile comes with a weight corresponding to the path weights, namely 
\begin{equation}\label{eq:generalweight}
w\left( \tiicoord \right) = a_{ij}, \qandq  w\left( \ticoord \right) = b_{ij}.
\end{equation}
With this, we can define the probability measure 
\begin{equation}
\P(\mc T) = \dfrac{w(\mc T)}{\sum_{\widetilde{\mc T}} w(\widetilde{\mc T})}.
\label{eq:tiling-weight}
\end{equation}

\begin{figure}[t]
    \begin{subfigure}{0.3 \textwidth}
     \centering
            \begin{tikzpicture}[scale = 0.5]
            \draw[fill = red] (0,0) -- (0, 1) -- (1,2) -- (1,1) -- cycle;
            \draw[line width = 0.5mm] (0,0.5) -- (1,1.5);
        \end{tikzpicture}
        \vspace{20 pt}
        
        \begin{tikzpicture}[scale = 0.5]
        \draw[fill = cyan] (0,0) -- (0, 1) -- (1,1) -- (1,0) -- cycle;
        \draw[line width = 0.5mm] (0,0.5) -- (1,0.5);
    \end{tikzpicture}
    \vspace{20 pt}
    
        \centering
        \begin{tikzpicture}[scale = 0.5]
        \draw[fill = yellow] (0,0) -- (1, 1) -- (2,1) -- (1,0) -- cycle;
    \end{tikzpicture}
    \end{subfigure}
    \begin{subfigure}{0.59\textwidth}
       \centering \includegraphics[width=0.75\linewidth]{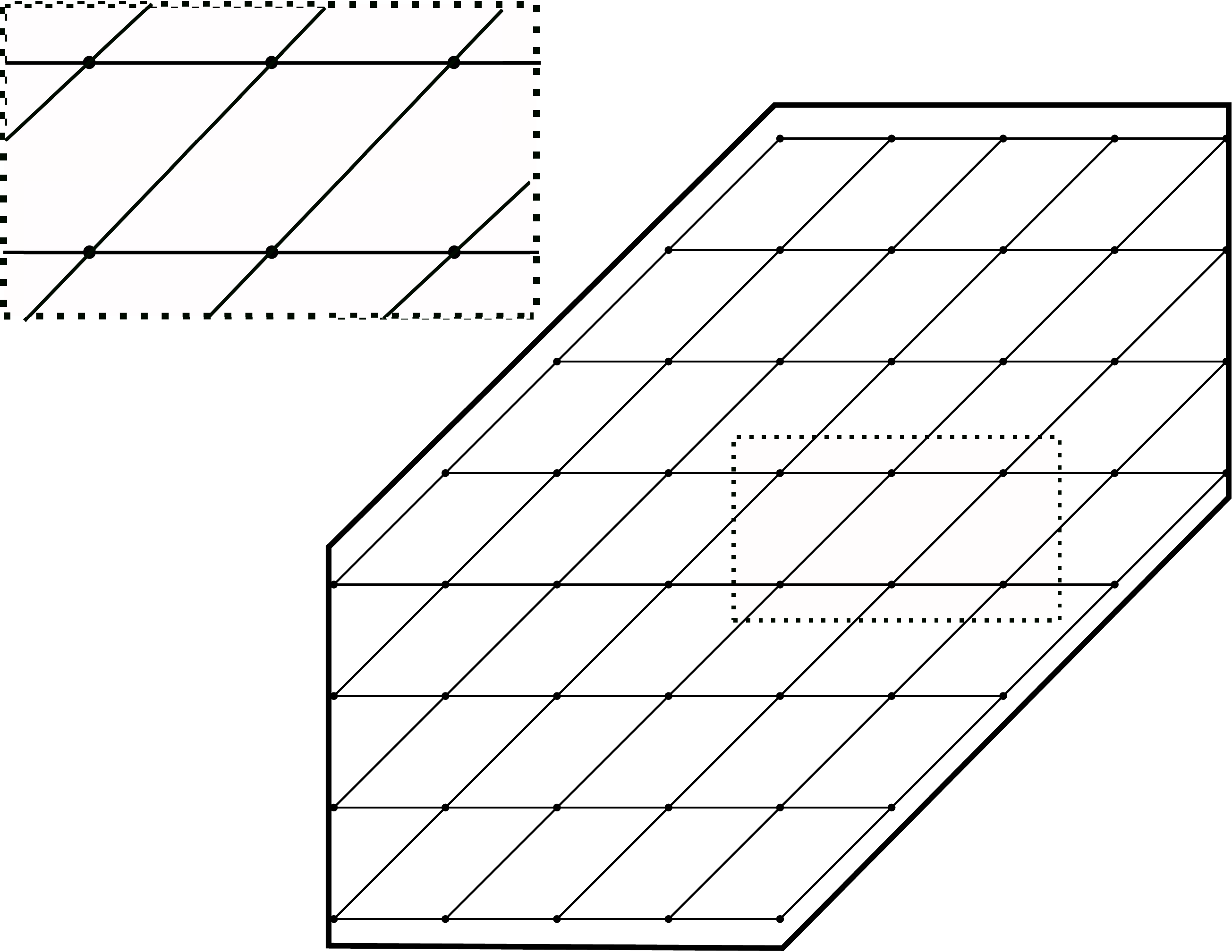}
       \put(-253, 135){$\left(i, \frac{1}{2}+j\right)$}
       \put(-193, 5){$\left(0, \frac{1}{2}\right)$}
       \put(-199, 120){$a_{ij}$}
       \put(-205, 145){$b_{ij}$}
    \end{subfigure}
    \caption{\centering Bijection between lozenge tilings and non-intersecting paths and the corresponding directed, acyclic graph of a $4 \times 4 \times 4$ hexagon.}
    \label{fig:bijection}
\end{figure}

\begin{figure}[t]
    \begin{subfigure}[b]{0.49\textwidth}
        \centering
        \includegraphics[scale = 0.5]{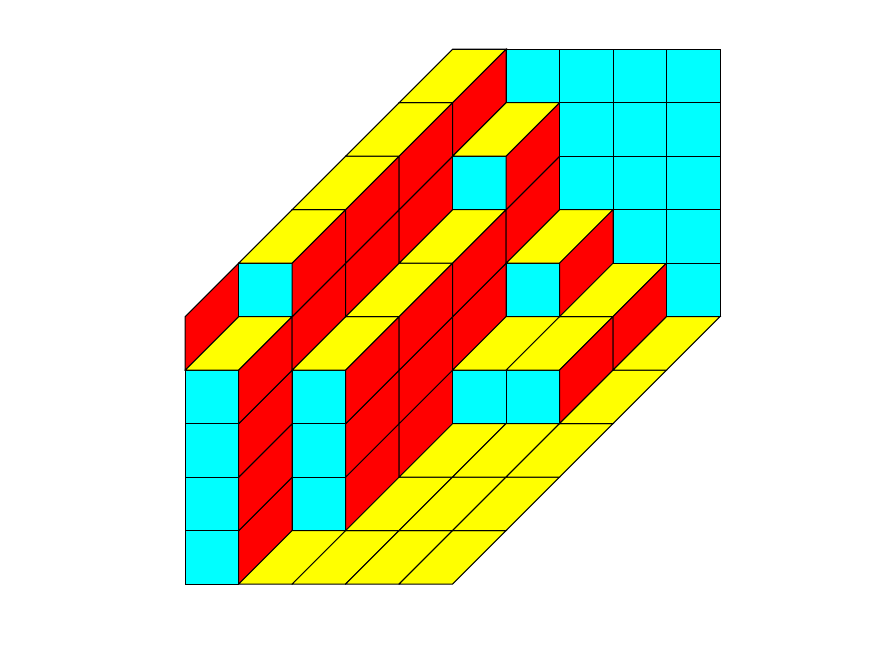}
        \caption{\centering Sample tiling}
    \end{subfigure}
    \begin{subfigure}[b]{0.49\textwidth}
        \centering
        \includegraphics[scale = 0.5]{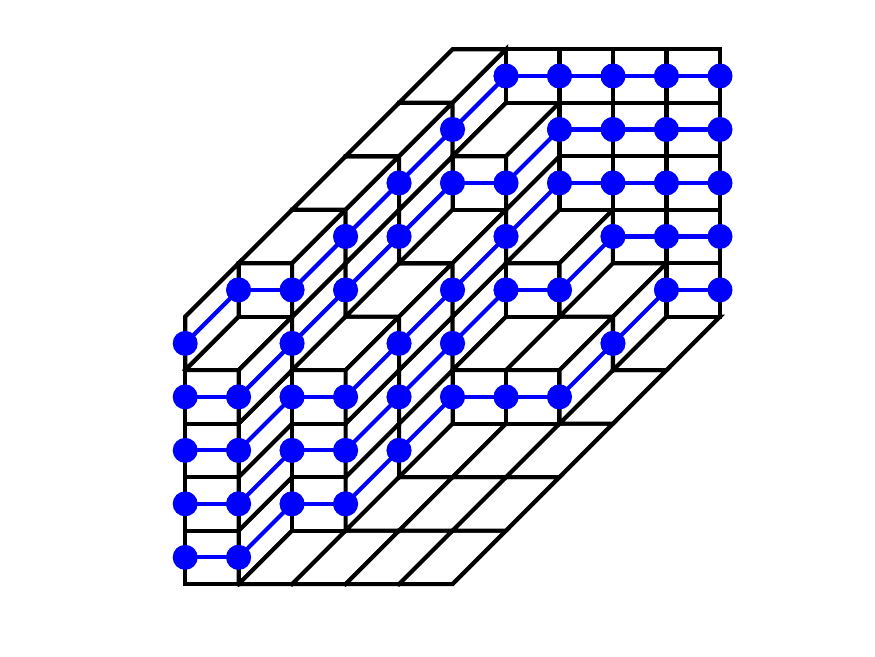}
        \caption{\centering Corresponding non-intersecting paths}
        \label{fig:paths-b}
    \end{subfigure}
    \caption{\centering Sample tiling of the hexagon with $N = 5$ and the corresponding ensemble of non-intersecting paths. }
    \label{fig:paths}
\end{figure}

One of the key features which makes almost all of the results above possible is the fact that lozenge tilings of the hexagon form a \emph{determinantal point process}. The point process is obtained through the well-known bijection of tilings of the hexagon with side length $N$ to $N$ non-intersecting, up-right paths on a directed, acyclic graph. Indeed, by drawing line segments through the lozenges as shown in Figure \ref{fig:bijection}, one can produce paths as in Figure \ref{fig:paths}. 

The result is a collection $\mc P$ of non-intersecting paths $\pi_j: \{ 0, ..., 2N \} \to \Z + \frac12$, $j = 0, 1, ..., N-1$ on the directed acyclic graph of Figure \ref{fig:bijection}. Apart from the fact that they are non-intersecting, we also note that the paths leave from initial points $\pi_j(0) = j + \frac 12$ and  end at the final points $\pi_j(2N) = N + j + \frac12$.  The general weight \eqref{eq:generalweight} for a given tiling then translates naturally to a weight on the paths:

\[
w(\mc P) := \prod_{\tiicoord \in \mc P} a_{ij} \prod_{\ticoord \in \mc P} b_{ij}.
\]

For any $j = 0, 1, ..., N-1$, let $(x_j^m)_{m = 0}^{2N} \in (\Z + \frac12)^{2N+1}$ be the intersection of the $j$th path with the vertical line passing through $(m, 0)$; the union of all $x_j^m$ is shown in Figure \ref{fig:paths-b}. A standard application of the Lindstr\"om-Gessel-Viennot lemma allows us to compute the weight of a system of non-intersecting paths and the probability measure corresponding to the point configuration $(x_j^m)_{j= 0, m = 0}^{N-1, 2N}$  can be written as a product of determinants. In the particular case of \eqref{eq:generalweight}, we obtain
\begin{equation}
\P \left( (x_j^m)_{j = 0, m = 0}^{N-1, 2N} \right)= \dfrac{1}{Z_N} \prod_{m = 0}^{2N-1} \det \left[ T_m \left(x_j^m - \frac12, x_k^{m+1} - \frac12 \right) \right]_{j, k = 0}^{N-1},  \quad x_j^0 = j + \frac12, \quad x_j^{2N} = N+j + \frac12,
    \label{eq:det-point-process-measure}
\end{equation}
and $T_m(x, y)$ is a $\Z \times \Z$ matrix given by 

\begin{equation}
    T_m(x,y) = \left\{ \begin{array}{ll} b_{mx}, & y = x + 1, \\ a_{mx}, & y = x, \\ 0, &\text{otherwise. } \end{array}\right.
\end{equation}

The measure \eqref{eq:det-point-process-measure} defines a determinantal point process whose correlation kernel is given by the Eynard-Mehta Theorem \cite{MR1628667}. On the level of the non-intersecting paths, this means that there exists a function (henceforth referred to as the correlation kernel or simply kernel) $K_N(x, y)$ such that for any 
$k \in \N$, and any set of integers $x_1, ..., x_k, y_1, ..., y_k$ satisfying $i \neq j \implies (x_i, y_i) \neq (x_j, y_j)$, 
\begin{equation}
\P \left( \substack{\text{Paths go through each of } \\ (x_1, y_1 + \frac 12), ..., (x_k, y_k + \frac12)} \right) = \det \left[ K_{N}(x_i, y_i; x_j, y_j)\right]_{i, j = 1}^k.
\label{eq:kernel-path-prob}
\end{equation}
An explicit formula for the kernel is provided by the Eynard-Mehta Theorem, but this requires inverting a large matrix. In general, this is a difficult task and the known approaches in the literature only work under special conditions on the weights.  We follow the approach of \cite{MR4228276} where, in the case where $T_m(x, y)$ are (block) Toeplitz matrices, the kernel can be written in terms of (matrix-valued) orthogonal polynomials. 


\subsection{Specialization to \texorpdfstring{$q^{\text{Volume}}$}{q-volume}}
We return to the $q^\text{Volume}$ model, which corresponds to the choice
\begin{equation} \label{eq:qweight}
    w\left( \tiicoord \right) = 1 , \qandq  w\left( \ticoord \right) = q^{-(i+ 1)}.
\end{equation}
 With this choice of weight, the rotation of tiles shown in Figure \ref{fig:rotation} changes the probability \eqref{eq:tiling-weight} by a factor $q^{-1}$. Since every tiling can be produced by a sequence of these rotations, we deduce that the resulting measure is exactly the $q^{-\text{Volume}}$ model.
 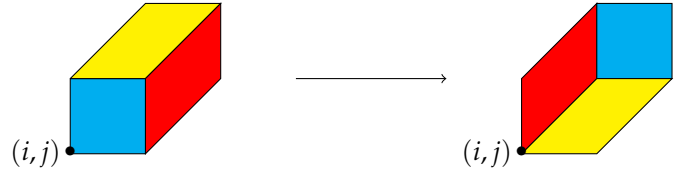
\begin{figure}[t]
        \centering
   \begin{tikzpicture}
       \draw[fill = cyan] (2,0) -- (2, 1) -- (3, 1) -- (3, 0) -- cycle;
       \draw[fill = yellow] (2, 1) -- (3, 1) -- (4, 2) -- (3, 2) -- (2, 1);
       \draw[fill = red] (3, 1) -- (3, 0) -- (4, 1) -- (4, 2) -- cycle;
       \draw[style = ->] (-1,1) -- (1, 1);
       \draw[fill = red] (-4, 0) -- (-4, 1) -- (-3, 2) -- (-3,1) -- cycle;
       \draw[fill = cyan] (-3, 1) -- (-3, 2) -- (-2, 2) -- (-2, 1) -- cycle;
       \draw[fill = yellow] (-4, 0) -- (-3, 1) -- (-2, 1) -- (-3, 0) -- cycle;
       \node[left] at (-4, 0){$(i, j)$};
       \node at (-4, 0){\textbullet};
       \node[left] at (2, 0){$(i, j)$};
       \node at (2, 0){\textbullet};
   \end{tikzpicture}
   \caption{\centering Rotation of tiles corresponding to adding a box.}
    \label{fig:rotation} 
\end{figure}
\begin{figure}[t]
    \begin{subfigure}[b]{0.49\textwidth}
        \centering
        \includegraphics[scale = 0.5]{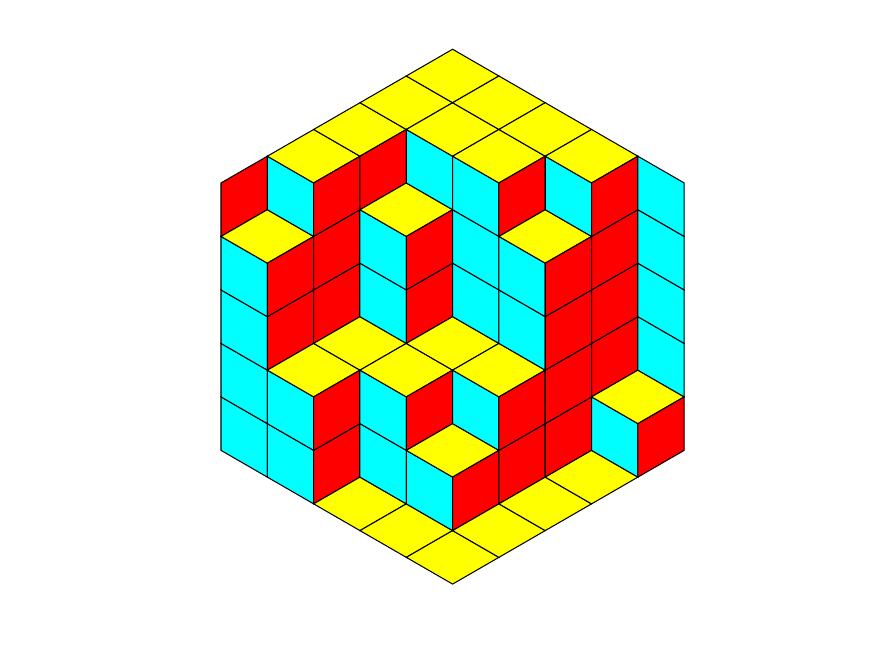}
    \caption{Sample tiling $\mc T$}
    \end{subfigure}
    \begin{subfigure}[b]{0.49\textwidth}
        \centering
        \scalebox{-1}[1]{\includegraphics[scale = 0.5, angle = 180]{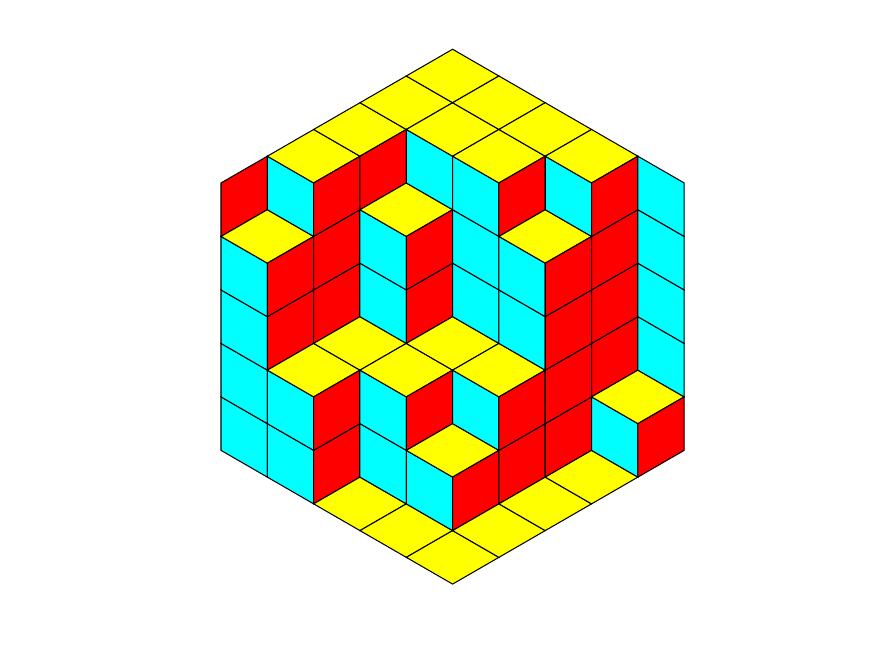}}
    \caption{Complementary tiling $\widetilde{\mc T}$}
    \end{subfigure}
    \caption{\centering Sample tiling $\mc T$ and its complementary tiling $\widetilde{\mc T}$.}
    \label{fig:complementary-tilings}
\end{figure}

For a fixed $q \in \R_+$, the models $q^{\text{Volume}}$ and $q^{-\text{Volume}}$ are intimately related. Indeed, for any given tiling $\mc T$, there exists a complementary tiling $\widetilde {\mc T}$ which is the result of rotating the hexagon $180^\circ$ and reflecting across the vertical axis, see Figure \ref{fig:complementary-tilings}. It is not hard to see that this produces a tiling with the property that 
\[
q^{\text{Vol}(\widetilde{\mc T})} = q^{N^3 - \text{Vol}(\mc T)}. 
\]
It follows from the very definition of both models that 
\[
\P_{q^{\text{Vol}}} \left( \mc T \right) = Z_{q^{\text{Vol}}}^{-1} q^{\text{Vol}(\mc T)} = \left( q^{-N^3} Z_{q^{\text{Vol}}} \right)^{-1} q^{-\text{Vol}(\widetilde{\mc T})} = Z_{q^{-\text{Vol}}}^{-1} q^{-\text{Vol}(\widetilde{\mc T})} =\P_{q^{-\text{Vol}}} ( \widetilde{\mc T} ),
\]
where the second-to-last equality follows from the definitions of $Z_{q^{\text{Vol}}}, Z_{q^{-\text{Vol}}}$:
\[
Z_{q^{\text{Vol}}} = \sum_{\mc T} q^{\text{Vol}(\mc T)} = q^{N^3} \sum_{\mc T} q^{-\text{Vol}(\widetilde{\mc T})} = q^{N^{3}} Z_{q^{-\text{Vol}}}.
\]
Hence, we make the arbitrary choice of considering $q \in (1, \infty)$ and tilings distributed according to $q^{-\text{Volume}}$.

We now identify parameters with \cite{MR4228276}*{Theorem 4.7}. To do so, first note that the choice \eqref{eq:qweight} makes $T_m(x,y)$ into a Toeplitz matrix, and $K_N(x,y)$ is a scalar function. Next, it is directly checked that Toeplitz matrices $T_m(x, y)$ have symbol 
\begin{equation}
    a_m(z) := 1 + q^{-(m+1)}z. 
\end{equation}
Viewing $T_m(x, y)$ as $\Z \times \Z$ matrices, it follows that, for any $m'> m$, the matrix
\[
T_{m, m'}(x, y) := \prod_{j = m}^{m' -1} T_{j}(x, y)  = T_{m}(x, y) \cdot T_{m + 1}(x, y) \cdots T_{m'-1}(x, y)
\]
is again a $\Z \times \Z$ Toeplitz matrix and has symbol\footnote{The corresponding statements are \emph{false} for finite and semi-infinite Toeplitz matrices.} 
\begin{equation}
    a_{m, m'}(z) = \prod_{j = m}^{m'-1} a_{ j}(z) = a_{m}(z) \cdot a_{m + 1}(z) \cdots a_{m'-1}(z) = \prod_{j = m}^{m'-1}(1 + q^{-(j+1)}z).
\end{equation}
In particular, we have 
\begin{equation}
    a_{0, 2N}(z) = \prod_{j = 0}^{2N - 1} (1 + q^{-(j + 1)}z) = \prod_{j = 1}^{2N} (1 + q^{-j}z).
\end{equation}
Finally, identifying the parameters via $(m, x, m', y) \mapsto (x_1, y_1, x_2, y_2)$, we arrive at 
\begin{multline}
    K_N(x_1, y_1; x_2, y_2) = -\dfrac{\chi_{x_1 > x_2}}{2 \pi \ii} \oint_\gamma \prod_{j = x_2 + 1}^{x_1} \left( 1+ q^{-j} z \right) \dfrac{\dd z}{z^{y_1 - y_2 + 1}}\\
     + \dfrac{1}{(2\pi \ii)^2} \oint_{\gamma} \oint_{\gamma} \left(\prod_{j = x_2+1}^{2N } (1 + q^{-j}w) \right) q^{N(2N+1)}R_N(w, z) \left(\prod_{j = 1}^{x_1}(1 + q^{-j}z) \right) \dfrac{w^{y_2}}{z^{y_1 + 1}w^{2N}} \dd z \dd w,
     \label{eq:corr-kernel}
\end{multline}
where $\gamma$ is a contour going around the origin in the positive direction and $R_N(w, z)$ is the reproducing kernel for the monic orthogonal polynomials $P_n(z) \equiv P_n(z; q, N)$ satisfying
\begin{equation}
    \dfrac{1}{2\pi \ii}\oint_\gamma z^k P_n(z) {\prod_{j = 1}^{2N}\left(1 + \frac{q^{j}}{z} \right)} \dd z = 0 \qforq k = 0, 1, ..., n-1.
    \label{eq:ortho}
\end{equation}
Since the measure of orthogonality in \eqref{eq:ortho} is not positive, it is not a priori clear that $\deg P_n = n$. Nonetheless, when it is well-defined $R_N(w,z)$ is given by the classical Christoffel-Darboux kernel (CD kernel)
\footnote{Strictly speaking, \cite{MR4228276}*{Theorem 4.7} requires $R_N(w,z)$ to be the CD kernel of the monic polynomials satisfying 
\[
    \oint_\gamma P_n(z)z^k  {\prod_{j = 1}^{2N}(1 + q^{-j} z)}\dfrac{ \dd z}{z^{2N}} = q^{-N(2N + 1)}\oint_\gamma P_n(z)z^k {\prod_{j = 1}^{2N}\left(1 + \frac{q^{j}}{z} \right)} \dd z = 0, \quad k = 0, 1, ..., n-1.
\]
This is the source of the factor $q^{N(2N+1)}$ in \eqref{eq:corr-kernel}.
}
\begin{equation}
    R_N(w, z) = \sum_{n = 0}^{N-1} \dfrac{P_n(w)P_n(z)}{\kappa_n} = \dfrac{1}{\kappa_{N-1}} \dfrac{P_{N}(z) P_{N-1}(w) - P_{N}(w) P_{N-1}(z)}{z - w},
    \label{eq:CD-ker}
\end{equation}
and 
\begin{equation}
    \kappa_n = \dfrac{1}{2\pi \ii }\oint_\gamma P^2_n(z) {\prod_{j = 1}^{2N}\left(1 + \frac{q^{j}}{z} \right)} \dd z.
    \label{eq:kappa-n}
\end{equation}

Before stating our results, it is worth noting that the authors of \cite{BGR} obtained yet another expression for the correlation kernel involving $q$-Hahn orthogonal polynomials. Their expression resembles the reproducing kernel of a projection operator, a crucial observation which allowed them to apply an argument due to Olshanski \cite{MR2492429} to deduce the distribution of tiles in the interior of the hexagon and find formulas for the arctic curve. Naturally, these results can be reproduced using the approach in this work, but we will refrain from doing so and instead study the tilings at the arctic curve. We also emphasize the fact that, despite both relying on orthogonal polynomials,  these approaches are very different. This can already be observed from the fact that the $q$-Hahn polynomials are orthogonal with respect to a discrete measure while our polynomials satisfy a non-Hermitian orthogonality relation with respect to an analytic density on a curve in the complex plane.

Returning to this work, our approach closely follows the one in \cite{MR4124992}, and is roughly split into two parts. The first is to obtain an asymptotic formula for the CD kernel, and the second is to apply a classical steepest descent type argument to conclude asymptotics of the correlation kernel \eqref{eq:corr-kernel}. This is reflected in our results, which are also split into two parts: results on the polynomials $P_n(z)$, some of which may be of independent interest, and results on the $q^{\text{Volume}}$ tiling model. 
\begin{figure}[t]
    \begin{subfigure}[b]{0.33\textwidth}
        \centering
        \includegraphics[scale = 0.6]{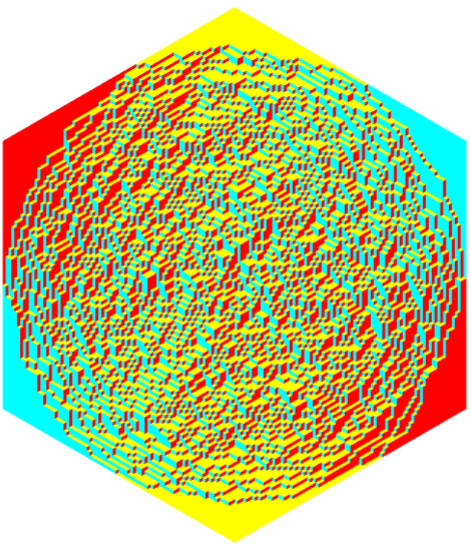}
    \caption{\centering $q = 1$ (uniform tiling)}
    \label{fig:arctic-a}
    \end{subfigure}
    \begin{subfigure}[b]{0.33\textwidth}
        \centering
        \includegraphics[scale = 0.6]{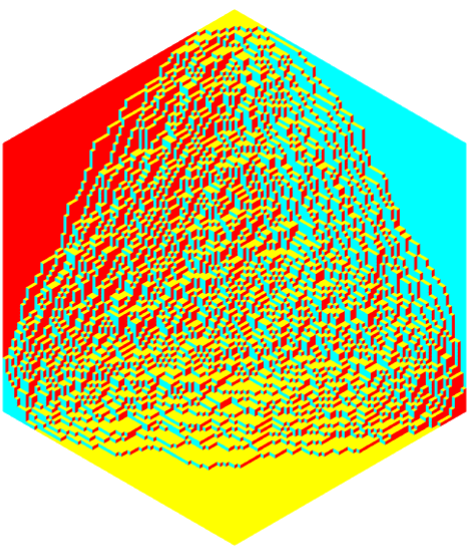}
    \caption{\centering $q = \ee^{\frac{1}{80}}$}
    \label{fig:arctic-b}
    \end{subfigure}
    \begin{subfigure}[b]{0.33\textwidth}
        \centering
        \includegraphics[scale = 0.6]{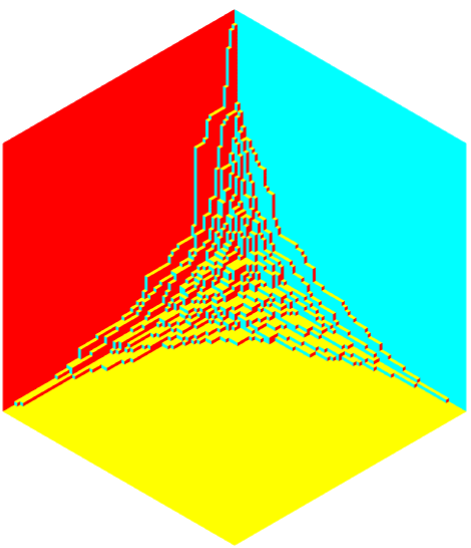}
    \caption{\centering $q = \ee^{\frac{5}{80}}$}
    \label{fig:arctic-c}
    \end{subfigure}
    \caption{\centering Sample $q^{-\text{Volume}}$ tilings of an $80 \times 80 \times 80$ hexagon.}
    \label{fig:sample-arctic-circles}
\end{figure}

\section{Statement of Results}
\label{sec:statement-of-results}

In a nutshell, our first observation is that the orthogonal polynomials defined in \eqref{eq:ortho} are, in fact, classical and expressible in terms of little $q$-Jacobi polynomials (with non-standard parameters). We will provide Plancherel-Rotach asymptotics for these polynomials, including the limiting zero distribution. Using the asymptotic behavior of the polynomials, we then describe the arctic curve and liquid region, and derive Airy asymptotics near the boundary. 

\subsection{Identification of \texorpdfstring{$P_n(z; q, N)$}{orthogonal polynomials}}

For any $a, b \in \C$, recall that the little $q$-Jacobi polynomials are given by (cf. \cite{KLS}*{Eq. (14.12.1)})
\begin{equation}
    j_n(x; a, b| q) :=  \sum_{k=0}^{n} \frac{(q^{-n}; q)_k (abq^{n+1}; q)_k}{(aq ; q)_k} \frac{(qx)^k}{(q ; q)_k} = \pfq{2}{1}{{q^{-n}},abq^{n+1}}{aq}{q}{qx}
    \label{eq:little-jacobi-def}
\end{equation}
where
\[
(\alpha ; q)_k:=\prod_{j=1}^k\left(1- \alpha q^{j-1}\right)
\]
is the $q$-Pochhammer symbol and
\[
\pfq{2}{1}{{\alpha},\beta}{\gamma}{q}{z} := \sum_{k=0}^{\infty} \frac{(\alpha; q)_k (\beta; q)_k}{(\gamma ; q)_k} \frac{z^k}{(q ; q)_k}
\]
is the $q$-hypergeometric function. Note that, in this standard normalization, $j_n(x; a, b| q)$ are not monic. Their monic counterparts are given by 
\begin{equation}
    J_n(x; a, b| q) := (-1)^{n} q^{\binom{n}{2}} \frac{(aq;q)_n}{(abq^{n+1}; q)_n} j_n(x; a, b| q)
    \label{eq:monic-jacobi}
\end{equation}
and satisfy the 3-term recurrence relation 
\begin{equation}
    x  J_n(x; a, b| q) = J_{n+1}(x; a, b| q) + (A_n + C_n) J_n(x; a, b| q) + A_{n - 1} C_n J_{n-1}(x; a, b| q),
    \label{eq:jacobi-3-term}
\end{equation}
where 
\begin{equation}
    \begin{aligned}
        A_n := q^n \dfrac{(1 - aq^{n+1})(1 - abq^{n+1})}{(1 - abq^{2n+1}) (1 - abq^{2n+2})}, \\
        C_n := aq^n \dfrac{(1 - q^n) (1 - bq^n)}{(1 - ab q^{2n})(1 - abq^{2n+1})}. 
    \end{aligned}
    \label{eq:recurrence-q-jacobi}
\end{equation}
Using calculations by Carlitz \cite{MR0227188}, we can identify polynomials $P_n(z)$ as little $q$-Jacobi polynomials.
\begin{proposition}
    For any $0 \leq n \leq 2N-1$, let $J_n(z; a, b|q)$ be as in \eqref{eq:monic-jacobi} and
    \begin{equation}
        P_n(z;q, N) :=  \left(-q^{2N+1}\right)^n\cdot  J_{n} \left( -\frac{z}{q^{2N+1}}; q^{-2N}, q^{2N} \biggl| q \right).
        \label{eq:pn-formula}
    \end{equation}
    Then, $P_n(z;q, N)$ satisfies the orthogonality relation \eqref{eq:ortho}, where $\gamma$ is a contour going around the origin in the positive direction. 
    \label{prop:jacobi}
\end{proposition}
\begin{remark}
    It follows from \eqref{eq:jacobi-3-term} and Favard's Theorem that $J_n(x; a, b| q)$ are orthogonal with respect to a positive measure supported on the real line iff $A_{n-1} C_n >0$. This is the case when, for example, one imposes the standard conditions
    \[
    0 < q < 1, \quad 0 < aq < 1, \quad bq < 1. 
    \]
    From this point of view, the choice of parameters in Proposition \ref{prop:jacobi} is non-standard in that $A_{2N-1} = 0$. Similar non-standard parameters for little $q$-Jacobi polynomials were considered in \cite{MR2832754}. 
\end{remark}
Proposition \ref{prop:jacobi} serves to, firstly, guarantee $\deg P_n = n$ for all integers $n \in [0, 2N-1]$, secondly, interpret Theorem \ref{thm:polynomial-asymptotics} in terms of a classical object in a non-classical asymptotic regime and, thirdly, allow efficient computation of zeros using expression \eqref{eq:pn-formula}. 

To arrive at asymptotic formulas for \eqref{eq:corr-kernel}, we first obtain Plancherel-Rotach asymptotics of \sloppy $P_N(z; \ee^{\frac{c}{2N}}, N)$. Since the integrand in \eqref{eq:ortho} is analytic away from $z = 0$, it is clear that one can replace contour $\gamma$ with any contour encircling the origin. It is well-understood that, in our setting, one must choose a curve $\gamma$ containing the so-called $S$-curves. These are curves which satisfy certain symmetry conditions and on which the zeros of the polynomials $P_N(z; \ee^{\frac{c}{2N}}, N)$ accumulate. Proposition \ref{prop:jacobi} makes exploring the zeros of $P_N(z; \ee^{\frac{c}{2N}}, N)$ numerically easy, and the results are shown in Figure \ref{fig:roots-various-c}. It is evident from these pictures that the roots accumulate on an arc of a circle, and this will directly follow from the asymptotics of $P_N(z; \ee^{\frac{c}{2N}}, N)$. To prove this and obtain the desired asymptotics we make the following observation:  roughly speaking, when $q = \ee^{\frac{c}{2N}}$,
\begin{equation}
    \prod_{j = 1}^{2N} \left(1 + \dfrac{q^j}{z} \right) \approx \exp \left\{ -NV_c(z) - \nu_c(z) \right\}
    \label{eq:rough-approx-weight}
\end{equation}
where 
\begin{equation}
    V_c(z) := -2\int_0^1 \log \left(1 + \dfrac{\ee^{c t}}{z} \right) \dd t = \dfrac{2}{c} \left(\mathrm{Li}_2\left( -\dfrac{\ee^c}{z} \right)  - \mathrm{Li}_2\left( -\dfrac{1}{z} \right) \right),
    \label{eq:V-def}
\end{equation}
$\mathrm{Li}_2(\diamond)$ is the classical dilogarithm (cf. \cite{DLMF}*{Section 25.12}), 
\begin{equation}
    \nu_c(z) := \frac12 \log \left(1 + \dfrac{1}{z}\right) -\frac12  \log \left(1 + \dfrac{\ee^c}{z}\right) ,
    \label{eq:nu-def}
\end{equation}
and $\log(\diamond)$ is the principal branch. Hence, following the methodology introduced in \cite{MR4714335}, one can approximate $P_N(z; \ee^{\frac{c}{2N}}, N)$ by polynomials orthogonal with respect to this limiting measure, whose asymptotics can be computed using Riemann-Hilbert methods. The presence of branch cuts means that the precise sense in which \eqref{eq:rough-approx-weight} holds will be crucial; we defer the precise statement of \eqref{eq:rough-approx-weight} to Section \ref{sec:proof-polynomial-asymptotic}. 

There is an extensive literature on the asymptotic analysis of polynomials orthogonal with respect to weights of the form of the right hand side of \eqref{eq:rough-approx-weight}. Motivated by questions in rational approximation, the role of $S$-curves in the asymptotic analysis of non-Hermitian orthogonal polynomials was developed in great detail by Stahl \cite{MR891973} and Gonchar and Rakhmanov \cite{MR922628}. Various characterizations and descriptions of $S$-curves were later obtained by Rakhmanov and Mart\'inez-Finkelshtein e.g. \cites{MR2770010, MR2964146}, and the existence of $S$-curves in the case of polynomial $V(z)$ was shown in \cite{MR3306308}. A remarkable feature of the polynomials arising in this work and in \cite{MR4124992} is that the $S$-curve is completely explicit: an arc of a circle with explicit endpoints (see \eqref{eq:angle} and \eqref{eq:endpoints} below). 

Next, we introduce the necessary notation and state our main asymptotic result for the orthogonal polynomials $P_N(z; \ee^{\frac{c}{2N}}, N)$.

\begin{figure}
    \begin{subfigure}[b]{0.3\textwidth}
            \centering
         \includegraphics[width = \textwidth]{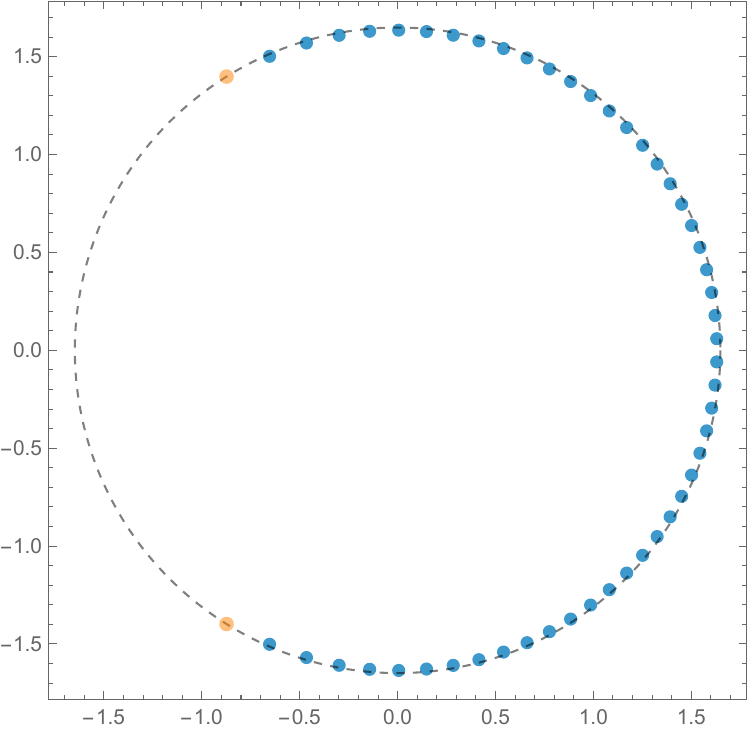}
        \caption{$c = 1$}
    \end{subfigure}
    \begin{subfigure}[b]{0.3\textwidth}
            \centering
         \includegraphics[width = 0.98\textwidth]{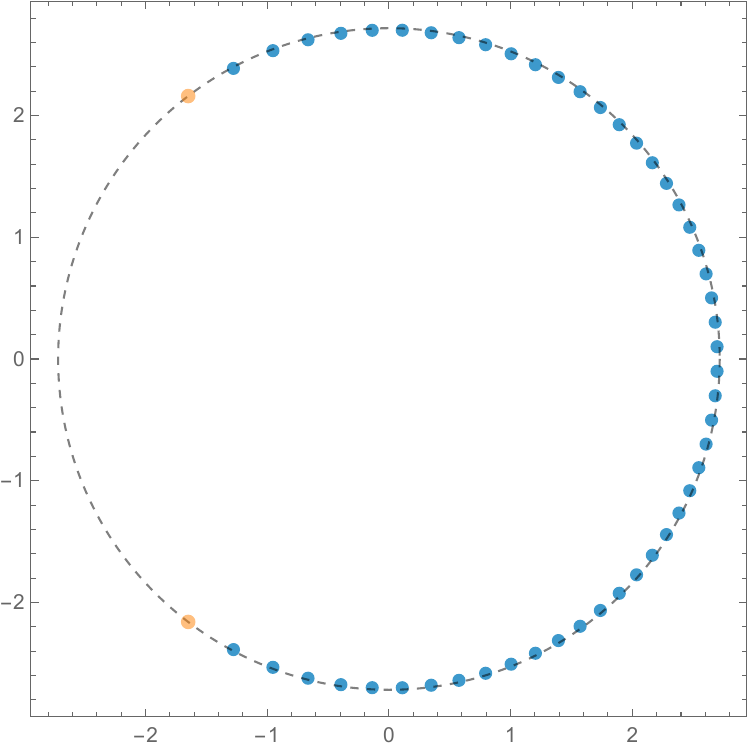}
        \caption{$c = 2$}
    \end{subfigure}
    \begin{subfigure}[b]{0.3\textwidth}
            \centering
         \includegraphics[width = \textwidth]{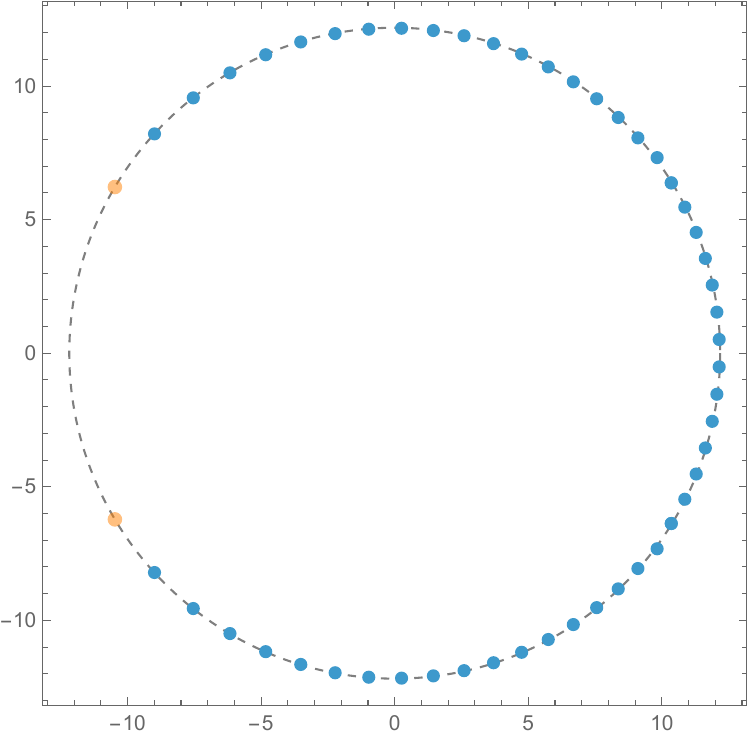}
        \caption{$c = 5$}
    \end{subfigure}

    \caption{\centering Roots (blue) of $P_{50}(z;\ee^{\frac{c}{100}}, 50)$ for various choices of $c$. The dashed circle is centered at the origin and has radius $\ee^{\frac{c}{2}}$ and the orange points are the endpoints $z_\pm$ (cf. \eqref{eq:endpoints}). }
    \label{fig:roots-various-c}
\end{figure}

\subsection{Asymptotics of \texorpdfstring{$P_n(z; q, N)$}{the orthogonal polynomials}}

Fix $c>0$ and denote 
\begin{equation}
    t := \tanh \left(\frac{c}{2} \right) \in (0, 1). 
    \label{eq:t-def}
\end{equation}
Let $\theta_c \in (\frac{2\pi}{3}, \pi)$ be the unique value satisfying
\begin{equation}
    \cos \theta_c := -\dfrac{\cosh ( \frac{c}{2})}{1 + \cosh(\frac{c}{2})} = - \dfrac{1 - \sqrt{1 - t^2}}{t^2},
    \label{eq:angle}
\end{equation}
and denote 
\begin{equation}
    z_\pm = \ee^{\frac{c}{2}} \ee^{\pm \ii \theta_c}.
    \label{eq:endpoints}
\end{equation}
These will turn out to be the endpoints of the arcs where the zeros of $P_N(z; \ee^{\frac{c}{2N}}, N)$ accumulate. To describe the density of the zeros, we will need a cache of functions which we now introduce. Let 
\begin{equation}
      \gamma_c := \ee^{\frac{c}{2}} \T \qandq  \gamma_{c, 0} := \{z \in \ee^{\frac{c}{2}} \T \ : \ \arg (z) \in (-\theta_c, \theta_c) \}
        \label{eq:curves-def}
\end{equation}
be oriented counterclockwise, and
\begin{equation}
    R_c(z) := \sqrt{(z - z_{-})(z - z_{+})}, \quad z \in \C \setminus \gamma_{c, 0}, 
    \label{eq:R-def}
\end{equation}
where the branch of the square root is analytic in the specified domain and satisfies $R_c(z) = z + \Oo(1)$ as $z \to \infty$. Finally, denote 
\begin{equation}
    Q_c(z) := t \dfrac{R_c(z)}{z + \ee^{\frac{c}{2}}}, \qquad z \in \C \setminus (\gamma_{c, 0} \cup \{-\ee^{\frac{c}{2}}\}).
    \label{eq:Q-def}
\end{equation}

The limiting density of the zeros of $P_N(z; \ee^{\frac{c}{2N}}, N)$ (shown in Figure \ref{fig:density} for various choices of $c$) is given by the boundary values of
\begin{equation}
    \psi_c(z)  := \dfrac{1}{cz} \log \left( \dfrac{1 + Q_c(z)}{1 - Q_c(z)}\right), \qquad z \in \C \setminus \left(\gamma_{c, 0} \cup[-\ee^c, -1] \cup \{0\}\right).
    \label{eq:psi-def}
\end{equation}
To see this, we start with the following proposition. 

\begin{proposition}
Let $\mu_c$ be the measure supported on $\gamma_{c, 0}$ and defined by
    \begin{equation}
        \dd \mu_c (z) := \dfrac{1}{\pi \ii } \psi_{c,-}(z) \dd z, \quad z \in \gamma_{c, 0}.
    \end{equation}
Then, $\mu_c$ is a positive probability measure.
\label{prop:prob-measure}
\end{proposition}

\begin{figure}
    \centering
    \includegraphics[width=0.5\linewidth]{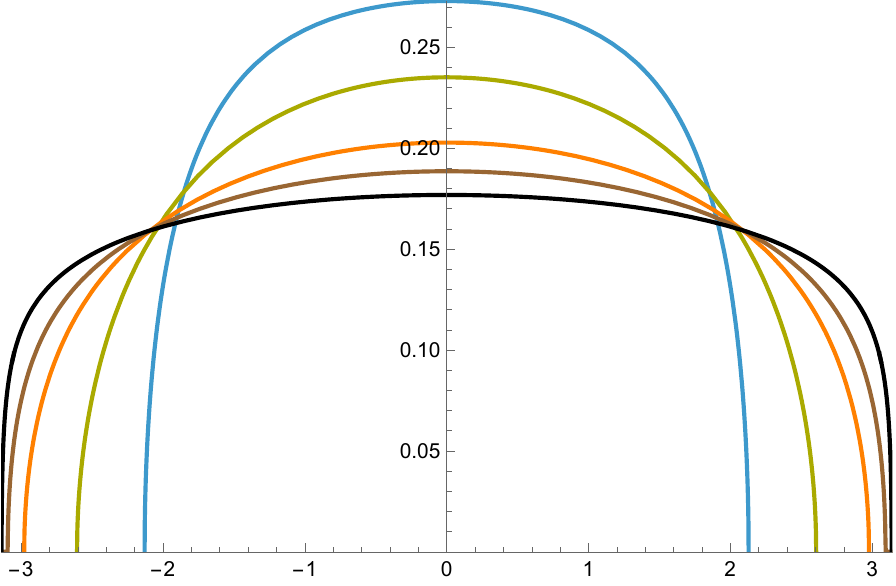}
    \caption{\centering Plot of $z\psi_-(z)$ where $z= \ee^{\frac{c}{2}}\ee^{\ii \theta}$, $\theta \in [-\theta_c, \theta_c]$, and $c =$ 1 (blue), 5 (green), 10 (orange), 15 (brown), 25 (black).}
    \label{fig:density}
\end{figure}
We can now define the functions appearing in the leading term asymptotics of $P_N(z; \ee^{\frac{c}{2N}}, N)$, starting with the so-called $g$-function. Indeed, let $\Gamma_c := \left( (-\infty, -\ee^{\frac{c}{2}}] \cup \{ \ee^{\frac{c}{2}} \ee^{\ii \theta} \ : \ \theta \in [-\pi, \theta_c] \} \right)$ be oriented towards $z_+$ and
    \begin{equation}
        g_c(z) := \int \log(z - x) \dd \mu_c(x), \quad z \in  \C \setminus \Gamma_c,
        \label{eq:g-def}
    \end{equation}
where $\log(\diamond - x)$ is the branch with a branch cut starting at $z = x$ and taken along $\Gamma_c$ towards $z = \infty$.
%
%
The main properties of the functions $g_c(z), \psi_c(z)$ and related functions will be discussed in detail in Section \ref{sec:prop-g-proofs}. For now, we continue towards the statement of our main results by defining the classical Szeg\H{o} function
\begin{equation}
    \varsigma_c(z) := \exp\left\{ \dfrac{R_c(z)}{2\pi \ii} \int_{\gamma_{c, 0}} \dfrac{\nu_c(x)}{x - z} \dfrac{\dd x}{R_{ c, -}(x)} \right\}.
    \label{eq:szego-fun}
\end{equation}
Finally, let 
\begin{equation}
    a_c(z) = \left(\dfrac{z - z_+}{z - z_-}\right)^{\frac14}, \quad z \in \C \setminus \gamma_{c, 0}
    \label{eq:a-fun}
\end{equation}
be the branch analytic outside $\gamma_{c, 0}$ and satisfying $a_c(z) \to 1$ as $z \to \infty$.

\begin{theorem}
With the notation of this section, we have that for $N$ large enough and any $z \in \C \setminus \gamma_{c, 0}$, the following holds locally uniformly:
\begin{equation}
    P_N \left(z; \ee^{\frac{c}{2N}}, N \right) = \ee^{Ng_c(z)} \left( \dfrac{1}{2}\dfrac{\varsigma_c(\infty)}{\varsigma_c(z)} \left( a_c(z) + \frac{1}{a_c(z)}\right) + \Oo(N^{-\frac{2}{3}}) \right).
\end{equation}
\label{thm:polynomial-asymptotics}
\end{theorem}
Since the leading term of the asymptotics is non-vanishing, an immediate consequence of this result is that the zeros of $P_N \left(z; \ee^{\frac{c}{2N}}, N \right)$ accumulate on the arc $\gamma_{c, 0}$ as seen in Figure \ref{fig:roots-various-c}. Similar formulas for $z \in \gamma_{c, 0}$ can be deduced from our analysis, but since we do not use these we will omit them.

\subsection{Asymptotics of \texorpdfstring{$K_N(x_1, y_1; x_2, y_2)$}{the correlation kernel}}
For any $x, y \in \Z$ such that $(x, y)$ belongs to the interior of the hexagon shown in Figure \ref{fig:tiling-b}, let $(\xi, \eta)$ be the coordinates defined by
\begin{equation}
     x =: N(1 + \xi)  \qandq  y =: N( 1 + \eta) , 
    \label{eq:rescaled-coordinates}
\end{equation}
and denote the rescaled hexagon by 
\begin{equation}
    \mc H := \left\{ (\xi, \eta) \ : \  |\xi|, |\eta| \leq 1, \ |\eta - \xi| \leq 1 \right\}.
    \label{eq:hexagon-set}
\end{equation}
A function which will play a key role in the analysis is 
\begin{multline}
    \Phi_c(z; \xi, \eta) := g_c(z) + 2\int_0^{\frac{1+\xi}{2}} \log \left(1 + z\ee^{-cu} \right) \dd u - (1 + \eta) \log z + \frac{\ell_c}{2}, \qquad \ell_c := -2 g_c(z_+) + V_c(z_+), \\
    z \in \C \setminus \left( (-\infty, 0] \cup \{ \ee^{\frac{c}{2}} \ee^{\ii \theta} \ : \ \theta \in [-\pi, \theta_c] \} \right).
    \label{eq:phase-def}
\end{multline}
Note that $\Phi_c(z)$ depends on the function $g_c(z)$, which was defined in \eqref{eq:g-def}. As is to be expected, the main contributions in the asymptotics of the correlation kernel will come from a neighborhood of (a subset of) the critical points of $\Phi_c(z; \xi, \eta)$ (in $z$). It follows from the definition of $\Phi_c(z), g_c(z)$ and \eqref{eq:conj-sym} below that 
\[
  \dod{\Phi_c}{z}(z; \xi, \eta) =   \overline{\dod{\Phi_c}{z}(\overline{z}; \xi, \eta)}. 
\]
Thus, the critical points of $\Phi_c(z; \xi, \eta)$ are either real or come in complex-conjugate pairs. The following will be crucial in describing the liquid region analytically. 
\begin{lemma}
    For any $(\xi, \eta) \in \mc H$, there is at most one critical point of $\Phi_c(z; \xi, \eta)$ in $\C_+$. 
    \label{lemma:saddle-pt}
\end{lemma}
The proof of Lemma \ref{lemma:saddle-pt}, along with other preliminaries to the saddle point analysis we will carry out to obtain Theorems \ref{thm:airy-convex}, \ref{thm:airy-inflection}, are in Section \ref{sec:saddle-point-preliminaries}. 

When it exists, denote the saddle point predicted by Lemma \ref{lemma:saddle-pt} by $s(\xi, \eta)$. 
We now define the \emph{liquid region}:
\begin{equation}
    \mc L_c := \{ (\xi, \eta) \in \mc H \ : \  \im \ s(\xi, \eta) > 0 \}.
    \label{eq:liquid-def}
\end{equation}
The boundary of the liquid region, known as the arctic curve, is given by
\[
\partial \mc L_c = \left \{ (\xi, \eta) \in \mc H \ : \ \dod{\Phi_c}{z} (z; \xi, \eta) \biggl |_{z = s(\xi, \eta )} = \dod[2]{\Phi_c}{z} (z; \xi, \eta) \biggl |_{z = s(\xi, \eta )} = 0 \right\}.
\]
This turns out to give a parametrization of $\partial \mc L_c$ with the saddle point $s$ as a parameter living on two copies of $\R$ as shown in Figure \ref{fig:arctic-circle-parametrization}. The details are contained in Section \ref{sec:saddle-point-preliminaries}. A feature of $\partial \mc L_c$ which we will pay close attention to is the appearance of inflection points. In this direction, we have the following proposition. 
\begin{proposition}
        Let $c_*$ be the unique positive number satisfying
        \[
            \cosh \left( \frac{c_*}{2}\right) = 1 + \sqrt{3}.
        \]
        Then, for $c < c_*$, $\partial \mc L_c$ is convex. When $c \geq c_*$, $\partial \mc L_c$ has finitely many inflection points. When $c > c_*$, the open segment $\mathfrak{S}_c := \overline{CD}$ (see Figure \ref{fig:arctic-circle-parametrization}) contains exactly one inflection point.
    \label{prop:inflection}
\end{proposition}

In the asymptotic analysis required to prove Theorem \ref{thm:airy-convex} and Theorem \ref{thm:airy-inflection}, we will need to identify certain contours of steepest descent which depend on $c$ and $(\xi, \eta)$. Thus, for definiteness, we state the following theorems for the (open) segment $\mathfrak{S}_c$ shown in Figure \ref{fig:arctic-circle-parametrization}. That being said, it follows directly from the definition of the model that it possesses a $2\pi/3$ rotational symmetry and a reflection symmetry across the vertical axis (in the coordinates of Figure \ref{fig:sample-arctic-circles}), and thus the result holds with minor modifications on the image of $\mathfrak{S}_c$ under these symmetries. The same proof can be applied to any other points on $\partial \mc L_c \setminus \partial \mc H$ with similar minor modifications. 

\begin{figure}
    \centering
    \begin{subfigure}[b]{0.47\textwidth}
        \centering
        \includegraphics[width=0.7\linewidth]{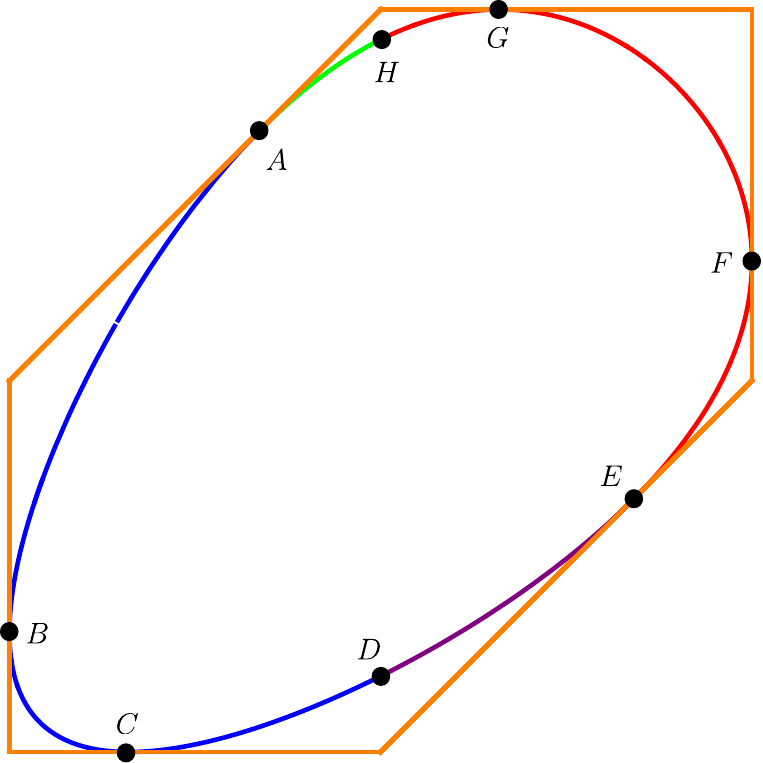}
        \put(-120, 15){$\mathfrak{S}_1$}
    \end{subfigure}
    \begin{subfigure}[b]{0.49\textwidth}
        \centering
        \includegraphics[width = 1\linewidth]{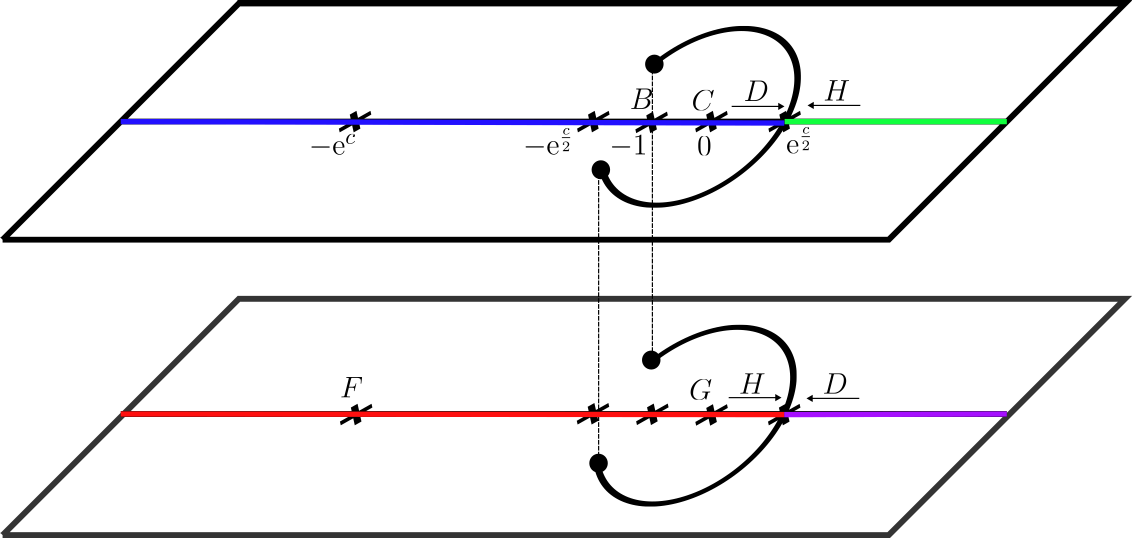}
        \put(-90,55){$R_c(\boldsymbol z) \sim {-z}$ as $\boldsymbol z \to \infty$}
        \put(-90,120){$R_c(\boldsymbol z) \sim {z}$ as $\boldsymbol z \to \infty$}
    \end{subfigure}
    \caption{\centering The arctic curve when $c = 1$ and its image under the map $(\xi, \eta) \mapsto s(\xi, \eta)$ on the Riemann surface corresponding to $R_c(z)$. Here, $\boldsymbol z$ denotes points on the Riemann surface whose projection to the plane is $z$.}
    \label{fig:arctic-circle-parametrization}
\end{figure}

\begin{theorem}
    Fix $(\xi, \eta) \in \mathfrak{S}_c \subset \partial \mc L_c  \setminus \partial \mc H$ and denote 
    \begin{equation}
        \varphi_{ijk} := \dpd[i]{}{z} \dpd[j]{}{\xi} \dpd[k]{}{\eta} \Phi_c(z; \xi, \eta) \biggl|_{\substack{z = s(\xi, \eta) }}.
        \label{eq:Phi-derivatives}
    \end{equation}
    For $j = 1, 2$, set
    \begin{equation}
        \begin{aligned}
            x_{j} &:= N(1 + \xi_{j, N}), \\
            y_{j} &:= N(1 + \eta_{j, N}),
        \end{aligned}
    \end{equation}
    and 
    \begin{equation}
        \mb n := \begin{bmatrix} \varphi_{110} \\ \varphi_{101}\end{bmatrix}, \qandq \mb n^\perp := \begin{bmatrix} -\varphi_{101} \\ \varphi_{110}  \end{bmatrix}.
    \end{equation}
    For any $\alpha_j, \beta_j \in \R$, let 
    \begin{equation}
        \begin{bmatrix}
            \xi_{j, N}\\ \eta_{j, N} 
        \end{bmatrix} := \begin{bmatrix}
            \xi\\ \eta
        \end{bmatrix} + \dfrac{\alpha_j}{N^{\frac23}} \mb n  + \frac{\beta_j}{N^\frac13} \mb n^{\perp}.
    \end{equation}
    Finally, let 
    \begin{equation}
        \widetilde{K}_N(x_1, y_1; x_2, y_2) := -s(\xi, \eta) k_6^{-1} \cdot \dfrac{ \ee^{- \left(N^{\frac{2}{3}} k_5 \beta_1 + N^{\frac{1}{3}} (k_3 \alpha_1 + k_4 \beta_1^2)\right) } }{ \ee^{- \left(N^{\frac{2}{3}} k_5 \beta_2 + N^{\frac{1}{3}} (k_3 \alpha_2 + k_4 \beta_2^2)\right) } } \cdot \dfrac{\ee^{ k_1 \beta_2^3 - k_2 \alpha_2 \beta_2 }}{\ee^{ k_1 \beta_1^3 - k_2 \alpha_1 \beta_1 } } {K}_N(x_1, y_1; x_2, y_2),
    \end{equation}
    where $k_j$'s are geometric constants given
    below (cf. \eqref{eq:k-values}). Then,
    \begin{equation}
        \lim_{N\to \infty} N^{\frac13} \widetilde{K}_N(x_1, y_1; x_2, y_2) = A(\tau(\beta_1), r(\alpha_1, \beta_1); \tau(\beta_2), r(\alpha_2, \beta_2)),
        \label{eq:kernel-limit-non-inflection}
    \end{equation}
    where 
    \begin{align}
        r(\alpha, \beta) &:=  -\left( \frac{2}{\varphi_{300}}\right)^\frac13 \left( \alpha \|\mb n \|^2 + \beta^2 \left( \frac12 \varphi_{120} \varphi_{101}^2 - \frac{1}{2} \frac{(\varphi_{201} \varphi_{110} - \varphi_{210} \varphi_{101})^2}{\varphi_{300}}\right) \right), \medskip \label{eq:r-fun}\\
        \tau(\beta) &:= -\beta \left( \frac{\varphi_{300}}{2}\right)^\frac13\frac{\varphi_{201} \varphi_{110} - \varphi_{210} \varphi_{101}}{\varphi_{300}},\label{eq:tau-fun}
    \end{align}
    and $A(\tau_1, r_1; \tau_2, r_2)$ is the extended Airy kernel given by 
    \begin{equation}
    A(\tau_1, r_1; \tau_2, r_2) := \begin{cases}
        \displaystyle \int_0^\infty \ee^{-t (\tau_1 - \tau_2)} \mathsf{Ai}(r_1 + t) \mathsf{Ai}(r_2 + t) \dd t, & \tau_1 \geq \tau_2, \medskip \\
        \displaystyle -\int_{-\infty}^0 \ee^{-t (\tau_1 - \tau_2)} \mathsf{Ai}(r_1 + t) \mathsf{Ai}(r_2 + t) \dd t, & \tau_1 < \tau_2.
    \end{cases}
\end{equation}
\label{thm:airy-convex}
\end{theorem}


\begin{remark}
    It is important to exclude points on $\partial \mc H$ since, at these points, the kernel has been shown to converge to the correlation kernel of the so-called GUE corners process \cite{MR3733659}. Naturally, the same phenomenon was observed for uniform lozenge tilings of the hexagon and domino tilings of the Aztec diamond earlier \cites{MR2118857, MR2274865, MR2268547, MR2336598}. 
\end{remark}

The last theorem fits with the universality of the Airy line ensembles at the edge of the liquid region for random tilings. Important to note is that $r(\alpha, \beta)$ is quadratic in $\beta$. This is a common feature of this type of asymptotics and stems from the curvature of the arctic curve. At the places where the arctic curve is locally convex, we will have 
$$
\frac12 \varphi_{120} \varphi_{101}^2 - \frac{1}{2} \frac{(\varphi_{201} \varphi_{110} - \varphi_{210} \varphi_{101})^2}{\varphi_{300}}>0.$$
If it is locally concave it will be negative.  This raises the natural question of what happens in the case of an inflection point, i.e. a point where the curvature vanishes, cf. Proposition \ref{prop:inflection}. A non-vanishing curvature was central in the  recent universality proof for uniform lozenge tilings for polygonal domains \cite{Aggarwal2021EdgeSF}.  We prove the following result:
\begin{theorem} \label{thm:inflection}
    Let $c>c_*$ (cf. Proposition \ref{prop:inflection}) and $(\xi_*, \eta_*)$ be the unique inflection point on $\mathfrak S_c$. For any $\alpha_j, \beta_j, \omega \in \R$ and any $\delta < \frac19$, let 
    \begin{equation}
        \begin{bmatrix}
            \xi_{j, N}\\ \eta_{j, N} 
        \end{bmatrix} := \begin{bmatrix}
            \xi_{*}\\ \eta_{*} 
        \end{bmatrix} + \dfrac{\alpha_j}{N^{\frac23}} \mb n  + \frac{{\beta}_j + \omega N^{\delta}}{N^\frac13} \mb n^{\perp}. 
    \end{equation}
   Then, with the notation of Theorem \ref{thm:airy-convex}, and the replacement $\beta_j \mapsto \beta_j + \omega N^{\delta}$ in the definition of $\widetilde K_N$, we have
    \begin{equation}
        \lim_{N\to \infty} N^{\frac13} \widetilde{K}_N(x_1, y_1; x_2, y_2) = A(\tau({\beta}_1), r(\alpha_1); \tau({\beta}_2), r(\alpha_2)),
    \end{equation}
    where 
    \(
    r(\alpha) := r(\alpha, 0). 
    \)
    \label{thm:airy-inflection}
\end{theorem}
This result shows that we still have the Airy line ensemble in the limit, but there is an interesting twist. First of all, $r(\alpha)$ does not depend on $\beta$. This means we no longer subtract/add a parabola to the Airy line ensemble. This is a consequence of the fact that the arctic curve is flat near an inflection point. In fact, it is so flat that we only observe the curvature at distance $N^{-{2}/{9}}$ in the direction of the tangent line to the curve. This is why we may shift the local parameters $\beta_j$  by $\omega N^{\delta}$  as long as $\delta<\frac19$ and still see a flat Airy line ensemble in the local limit.
\begin{remark}
     When $c = c_*$, the arctic curve $\partial \mc L_{c_*}$ has three inflection points of higher order (where two inflection points merge). In this case, the curve is even more flat and one expects to be able to adapt the arguments presented here with an increased exponent $\delta < \frac{1}{6}$; see Remark \ref{remark:exponents} below.
    \label{remark:intro-exponents}
\end{remark}
\begin{remark}
    It is not difficult to find models in the class of Schur processes whose arctic curves have inflection points. For instance, in the limit in which the vertical sides of the hexagon go to infinity first, the $q^{\text{Volume}}$ model is an example of a Schur process, and thus the correlation functions have simple double integral expressions.  When further taking the limit for which the width of the hexagon goes to infinity simultaneously as $q\to 1$, the arctic curve has an inflection point. We verified that also in that case, an analogous result to Theorem \ref{thm:inflection} holds.
\end{remark}

\subsection{Overview of the rest of the paper}

In the next section, we prove Proposition \ref{prop:jacobi}. Section \ref{sec:prop-g-proofs} contains preliminary results on the measure $\mu$ and the $g$-function, including a proof of Proposition \ref{prop:prob-measure}. These will be used in the proof of Theorem \ref{thm:polynomial-asymptotics}, which is in Section \ref{sec:proof-polynomial-asymptotic}. Section \ref{sec:saddle-point-preliminaries} contains preliminary results on the function $\Phi_c(z; \xi, \eta)$ and the liquid region $\mc L_c$, including proofs of Lemma \ref{lemma:saddle-pt} 
and Proposition \ref{prop:inflection}. Finally, we prove Theorems \ref{thm:airy-convex}, \ref{thm:airy-inflection} in Section \ref{sec:kernel-asymptotics}. In the remainder of the paper, we take $c$ to be fixed and drop the dependence on it where this does not cause confusion.

\section{Proof of Proposition \ref{prop:jacobi}}
\label{sec:prop-jacobi-proof}

We first observe that the moments of the measure of orthogonality can be written explicitly in terms of $q$-binomial coefficients. Indeed, by Gauss' $q$-binomial Theorem we have 
\begin{equation}
    \prod_{j = 1}^{2N} \left( 1 + \dfrac{q^j}{z}\right) = \sum_{j = 0}^{2N} \qbinom{2N}{j}_q q^{\frac{j(j+1)}{2}} \frac{1}{z^j},
\end{equation}
and a residue calculation implies 
\begin{equation}
    \mu_k(q, N) := \dfrac{1}{2\pi \ii }\oint_{\gamma} z^k \prod_{j = 1}^{2N} \left( 1 + \dfrac{q^j}{z}\right) \dd z =   q^{\frac{(k+1)(k+2)}{2}}\qbinom{2N}{k+1}_q .
    \label{eq:moments}
\end{equation}
We can algebraically manipulate the definition of the $q$-binomial coefficients to get
\begin{equation}
    \qbinom{x}{k+1}_q =  \frac{(q^{x - k}; q)_{k+1}}{(q; q)_{k+1}} = (-1)^{k+1} q^{-\frac{(k+1)(k+2)}{2}}\dfrac{(q^{x}; q^{-1})_{k+1}}{(q^{-1}; q^{-1})_{k+1}}.
    \label{eq:q-binom-id-1}
\end{equation}
This allows us to make contact with the work of Carlitz \cite{MR0227188}; denote
\begin{equation}
    \Delta_{n, k} (q, N) := (-1)^{n - k} \det 
    \begin{bmatrix}
    \mu_{0} & \mu_{1} & \cdots & \widehat {\mu_k} & \cdots & \mu_{n} \\
    \mu_{1} & \mu_{2} & \cdots & \widehat {\mu_{k+1}} & \cdots & \mu_{n+1} \\
    \vdots & \vdots & \ddots & \vdots \\
    \mu_{n-1} & \mu_{n} & \cdots & \widehat {\mu_{n + k - 1} }&\cdots & \mu_{2n-1} 
    \end{bmatrix}, \qforq k = 0, 1, ..., n,
\end{equation}
and the hats denote a column that is removed. Furthermore, for any integers $0 \leq k_0 < k_1 < ... < k_{m}$ let
\begin{equation}
    E_q(x; k_0, k_1, ..., k_{m}) := \det \left[\dfrac{(x; q)_{k_r + s}}{(q; q)_{k_r + s}} \right]_{r, s = 0}^{m}. 
\end{equation}
Then, it follows from \eqref{eq:moments}, \eqref{eq:q-binom-id-1} that 
\begin{equation}
    \Delta_{n, k}(q, N) = (- 1 )^n E_{q^{-1}} \left(q^{2N}; 1, 2, ..., \widehat{k+1}, ..., n+1 \right).
\end{equation}
By \cite{MR0227188}*{Eqs. (1.13), (4.3)}, we have 
\begin{equation}
    \Delta_{n, k}(q, N) =  (- 1 )^{\frac{n(n+1)}{2}} q^{-\frac{n(n - 1)(n - 2)}{6}} V(q^{-k_0}, q^{-k_1}, ..., q^{-k_{n - 1}}) \prod_{j = 0}^{n-1} \dfrac{(q^{2N}; q^{-1})_{k_j}}{(q^{-1}; q^{-1})_{k_j + n-1}} \prod_{j =1}^{n - 1} (q^{2N} - q^{-j})^{n-j}
    \label{eq:minors-1}
\end{equation}
where 
\begin{equation}
    k_j := \begin{cases}
        j + 1, & j = 0, 1, ..., k - 1, \\
        j+2, & j = k, k+1, ..., n-1,
    \end{cases}
\end{equation}
and 
\begin{equation}
    V(x_0, x_1, ..., x_{n-1}) = \prod_{0\leq j < k \leq n-1} (x_k - x_j)
\end{equation}
is the standard Vandermonde determinant.
Observe that for $n < 2N$ and $q \in (1, \infty)$, $\Delta_{n, k}(q, N)$ is non-vanishing. In particular, $\Delta_{n, n}(q, N) \neq 0$ and thus, the orthogonal polynomials $P_n(z)$ are given by the standard formula 
\begin{equation}
    P_n(z; q, N) = \dfrac{1}{\Delta_{n,n}(q, N)} \det \begin{bmatrix}
    \mu_{0} & \mu_{1} & \cdots & \mu_{n} \\
    \mu_{1} & \mu_{2} & \cdots & \mu_{n+1} \\
    \vdots & \vdots & \ddots & \vdots \\
    \mu_{n-1} & \mu_{n} & \cdots & \mu_{2n-1} \\
    1 & z & \cdots & z^{n}
    \end{bmatrix} = z^n + \sum_{k = 0}^{n-1} \dfrac{\Delta_{n, k}(q, N)}{\Delta_{n,n}(q, N)} z^k.
    \label{eq:classical-det}
\end{equation}
To compute the coefficients of $P_n(z; q, N)$ explicitly, we first observe the following identity, whose proof is an algebraic manipulation:
\begin{multline}
    V(q^{-1}, q^{-2}, ..., \widehat{q^{-(k+1)}}, ..., q^{-(n+1)})  = q^{- \binom{n-k}{2}} \dfrac{(q^{-1}; q^{-1})_{n}}{(q^{-1}; q^{-1})_{n -k} (q^{-1}; q^{-1})_{k}}  V(q^{-1}, q^{-2}, ..., q^{-n}) \\
    =   q^{- \binom{n-k}{2}}\qbinom{n}{k}_{q^{-1}}  V(q^{-1}, q^{-2}, ..., q^{-n}),
     \label{eq:vandermonde-identity}
\end{multline}
where the hat indicates an omitted argument. Using this and \eqref{eq:minors-1} we find 
\begin{equation}
    \dfrac{\Delta_{n, k}(q, N)}{\Delta_{n,n}(q, N)} =\underbrace{q^{-\binom{n-k}{2}}\qbinom{n}{k}_{q^{-1}}}_{\text{(A)}}\cdot\underbrace{\frac{(q^{2N};q^{-1})_{n+1}}{(q^{2N};q^{-1})_{k+1}}}_{\text{(B)}}\cdot\underbrace{\frac{(q^{-1};q^{-1})_{n+k}}{(q^{-1};q^{-1})_{2n}}}_{\text{(C)}}, \quad k = 0, 1, ..., n - 1,
    \label{eq:Delta-ratio}
\end{equation}
where we have labelled the terms for the convenience of the reader. Finally, using the identities 
\begin{align*}
    \qbinom{n}{k}_{q^{-1}} &= q^{-k(n-k)}\qbinom{n}{k}_{q} = q^{-k(n-k)}\,\frac{(q;q)_n}{(q;q)_{n-k}\,(q;q)_k}, \\
    (q;q)_{n-k} & =(-1)^{n-k}q^{\binom{n-k+1}{2}}\frac{(q^{-n};q)_n}{(q^{-n};q)_k}, \\
    (q^{2N}; q^{-1})_{k+1} &= (-1)^{k+1} (q^{2N})^k (q^{2N} - 1) q^{-\binom{k+1}{2}} (q^{-2N+1}; q)_k  ,\\
    (q^{-1}; q^{-1})_{n+k} &= (-1)^{n+k} q^{-\binom{n+k+1}{2}} (q;q)_n (q^{n+1};q)_k,
\end{align*}
we can rewrite \eqref{eq:Delta-ratio} as 
\begin{equation}
    \dfrac{\Delta_{n, k}(q, N)}{\Delta_{n,n}(q, N)} =(-q^{2N})^{n-k}\;\underbrace{\frac{(q^{-n};q)_k\,(q;q)_n}{(q;q)_k\,(q^{-n};q)_n}}_{\text{(A)}}\;\underbrace{\frac{(q^{-2N+1};q)_n}{(q^{-2N+1};q)_k}}_{\text{(B)}}\;\underbrace{\frac{(q^{n+1};q)_k}{(q^{n+1};q)_n}}_{\text{(C)}}.
    \label{eq:Delta-ratio-2}
\end{equation}
It follows directly from the definition of the $q$-hypergeometric function that 
\begin{equation}
    P_n(z; q, N) = (-q^{2N})^n\dfrac{(q^{-2N+1}; q)_n (q; q)_{n}}{(q^{n+1}; q)_n(q^{-n}; q)_n} \cdot  \pfq{2}{1}{{q^{-n}},q^{n+1}}{q^{-2N+1}}{q}{-\dfrac{z}{q^{2N}}}.
\end{equation}
Finally, using the identity 
\[
\dfrac{(q^{-n}; q)_n}{(q;q)_n} = (-1)^nq^{-\frac{n(n+1)}{2}}
\]
and the definitions \eqref{eq:little-jacobi-def}, \eqref{eq:monic-jacobi} (note the argument of the hypergeometric function in \eqref{eq:little-jacobi-def}) yields the desired formula.

\section{Construction of the equilibrium measure and the \texorpdfstring{$g$}{g}-function}
\label{sec:prop-g-proofs}
In this section, we collect basic facts about the functions $\psi(z)$ and $g(z)$, see \eqref{eq:psi-def} and \eqref{eq:g-def}, respectively, and prove Proposition \ref{prop:prob-measure}. These, in turn, will be used to prove Theorem \ref{thm:polynomial-asymptotics} in Section \ref{sec:proof-polynomial-asymptotic}.

\subsection{The density \texorpdfstring{$\psi(z)$}{of the equilibrium measure}} 
\label{subsec:proof-pos-measure}

First, we verify that $\psi(z)$ as defined in \eqref{eq:psi-def} is analytic on its domain. For this, we record the following lemma which will be useful in the upcoming analysis. 
\begin{lemma}
    Under the mapping $z \mapsto cz\psi(z)$, 
    \begin{enumerate}[(a)]
    \item the pre-image of $\R$ is a subset of $\R \cup \gamma_0$,
    \item the pre-image of $\ii \R$ is a subset of $\R \cup (\gamma \setminus \gamma_0)$,
    \item the image of $\C_+ \setminus  (\ee^{\frac{c}{2}} \overline{\D})$ is a subset of the first quadrant,
    \item the image of $\C_- \setminus  (\ee^{\frac{c}{2}} \overline{\D})$ is a subset of the fourth quadrant.
    \end{enumerate}
    \label{lemma:pre-image-psi}
\end{lemma}
\begin{proof}
\begin{enumerate}[(a)]
    \item Fix $r \in \R$ and suppose $cz\psi(z) = r$. Direct algebraic manipulation implies
    \[
        Q(z) = \dfrac{\ee^r - 1}{\ee^r + 1} \in \R   \implies \dfrac{R^2(z)}{(z+\ee^{\frac{c}{2}})^2} \in \R_+,
    \]
    and thus  
    \[
        \dfrac{R^2(z)}{(z+\ee^{\frac{c}{2}})^2} = 1 - \dfrac{2\ee^{\frac{c}{2}}(1+\cos(\theta_c))}{z + \dfrac{\ee^c}{z} + 2\ee^{\frac{c}{2}}} \implies z + \dfrac{\ee^c}{z} \in \R
    \]
    and since $\im\left(z + \dfrac{\ee^c}{z} \right) = \im(z) \left( 1 - \dfrac{\ee^c}{|z|^2} \right)$, we have $z\in \R$ or $|z| = \ee^\frac{c}{2}$. For $z = \ee^{\frac{c}{2}} \ee^{\ii \theta}$, 
    \begin{equation}
        \label{eq:Q-squared}
        Q^2(z) = t^2 \dfrac{\sin \left( \dfrac{\theta_c - \theta}{2} \right)\sin \left( \dfrac{\theta_c + \theta}{2} \right)}{\cos^2(\theta/2)} ,
    \end{equation}
    and thus, for $\theta \in [-\theta_c, \theta_c]$, we find $Q^2(z) \in [0, 1)$, which implies $Q_{\pm}(z) \in \R$. Consequently, it follows from the choice of branch of $R(z)$ that $Q_-(z) \in [0, 1)$ and $Q_{+}(z) \in (-1, 0]$.
    
    \item Similar to (a), we suppose $cz\psi(z) = \ii r$, $r \in \R$, and find 
    \[
        Q(z) \in \ii \R \implies \dfrac{R^2(z)}{(z+\ee^{\frac{c}{2}})^2} \in \R_- \implies z + \dfrac{\ee^c}{z} \in \R \implies z \in \R \cup \gamma.
    \]
    Finally, it follows from \eqref{eq:Q-squared} that, for $z = \ee^\frac{c}{2} \ee^{\ii \theta}$, we have $\theta \in [-\pi, \pi] \setminus [-\theta_c, \theta_c].$
    
    \item Suppose $z\in \C_+ \setminus (\ee^{\frac{c}{2}} \overline{\D})$. Then, the identity $\im\left(z + \dfrac{\ee^c}{z} \right) = \im(z) \left( 1 - \dfrac{\ee^c}{|z|^2} \right) > 0$ and so, 
    \[
        \im \left( \dfrac{R^2(z)}{(z+\ee^{\frac{c}{2}})^2} \right) > 0 
    \]
    Thus, it follows from (a), the asymptotic formula $R(z) \sim z$ as $z \to \infty$, and the analyticity of $R(z)$ for $z\in \C_+ \setminus (\ee^{\frac{c}{2}} \overline{\D})$ that 
    \[
        \sqrt{\dfrac{R^2(z)}{(z+\ee^{\frac{c}{2}})^2}} = \dfrac{R(z)}{z+\ee^{\frac{c}{2}}}, \qquad z\in \C_+ \setminus (\ee^{\frac{c}{2}} \overline{\D})
    \]
    where $\sqrt{\diamond}$ is the principal square root and the result follows. 

    \item It follows from the choice of branch that for $z \in \C \setminus \gamma_0$, $R(\overline{z}) = \overline{R(z)}$ and thus $Q(\overline{z}) = \overline{Q(z)}$. Since the logarithm in \eqref{eq:psi-def} is principal, we have 
    \begin{equation}
        \overline{cz\psi(z)} = c\overline{z}\psi(\overline{z}), \qquad z \in \C \setminus (\gamma_0 \cup [-\ee^c, -1] \cup \{0\})
        \label{eq:conj-sym}
    \end{equation}
    and the result now follows from (c).
    \end{enumerate}
\end{proof}
\begin{lemma}
The function $\psi(z)$ is the unique function satisfying the following scalar Riemann-Hilbert Problem (RHP):
\begin{rhp}
    Seek a scalar function $\psi(z)$ satisfying:
    \begin{enumerate}[(a)]
        \item $\psi(z)$ is analytic in $\C \setminus (\gamma_0 \cup [-\ee^c, -1] \cup \{0\})$
        \item $\psi(z)$ has continuous boundary values on $\gamma_0 \cup (-\ee^c, -1)$ satisfying\footnote{Here and throughout the text, given an arc $\gamma$, we understand jump conditions stated for $\gamma$ to hold on $\gamma\setminus \{\text{endpoints}\}$.}
        \begin{equation}
            \psi_+(z) = -\psi_-(z), \qquad z \in \gamma_0 
            \label{eq:psi-jump-arc}
        \end{equation}
        and
        \begin{equation}
        \psi_+(z) - \psi_-(z) = \dfrac{2\pi \ii }{cz}, \qquad z \in (-\ee^c, -1),
        \label{eq:psi-jump-interval}
        \end{equation}
        where the interval is oriented from left to right.
        \item $\psi(z)$ has a simple pole at $z = 0$ satisfying 
        \begin{equation}
            \res_{z = 0} \psi(z) = -1. 
        \label{eq:psi-zero-residue}
        \end{equation}
        \item $\psi(z)$ satisfies the following asymptotic conditions: 
        \begin{align}
            \psi(z) &= \dfrac{1}{z} + \Oo(z^{-2}) \qasq z \to \infty
            \label{eq:psi-infnity-residue}, \\
            \psi(z) &= \Oo \left(\sqrt{z - z_{\pm}} \right) \qasq z \to z_{\pm}, \label{eq:psi-endpoint-decay}\\
            cz\psi(z) &= -\log(z + \ee^{c}) + \Oo(1) \qasq z \to -\ee^c, \label{eq:psi-interval-decay-1} \\
            cz\psi(z) &= \log(z + 1) + \Oo(1) \qasq z \to -1. \label{eq:psi-interval-decay-2}
        \end{align}
    \end{enumerate}
    \label{rhp:psi}
\end{rhp}
\end{lemma}
\begin{proof}
    \emph{ Uniqueness.} Suppose $\psi_1(z), \psi_2(z)$ are two functions which satisfy Riemann-Hilbert Problem \ref{rhp:psi} and consider $$f_j(z) := \dfrac{z}{R(z)}\psi_j(z), \qquad j = 1, 2,$$ where $R(z)$ is given by \eqref{eq:R-def}. Then, $f_j(z)$ are continuous across $\gamma_0$ and thus, by an application of Morera's Theorem, are analytic in $\C  \setminus ( [-\ee^c, -1] \cup \{0, z_{\pm}\})$. It follows from \eqref{eq:psi-endpoint-decay}, RHP \ref{rhp:psi}(c) that $z = z_{\pm}$ and $z = 0$, respectively, are removable singularities, and so $f_j(z)$ are analytic in $\C  \setminus ( [-\ee^c, -1])$. Finally, for $z \in (-\ee^c, -1)$, $f_j(z)$ both satisfy the jump condition 
    \begin{equation}
        f_+(z) - f_-(z) = \dfrac{2\pi \ii}{cR(z)}. \label{eq:f-jump}
    \end{equation}
    Consider the difference $d(z) := f_1(z) - f_2(z)$. Then, it follows from \eqref{eq:f-jump} that $d(z)$ is analytic in $\C\setminus \{-\ee^c, -1\}$ and, from \eqref{eq:psi-interval-decay-1}, \eqref{eq:psi-interval-decay-2} that the singularities at $z = -1, -\ee^c$ are removable. Thus, $d(z)$ is entire. Furthermore, using \eqref{eq:psi-infnity-residue}, we find $\lim_{z\to \infty}d(z) = 0$. By Liouville's Theorem, $d(z) \equiv 0$ and thus $f_1(z) = f_2(z)$.\\

    \emph{ Existence.} Solutions to scalar Riemann-Hilbert problems are classical, see e.g. \cite{gakhov}, and one can verify that 
    \begin{equation}
        \psi(z) := \dfrac{R(z)}{cz} \int_{-\ee^c}^{-1} \dfrac{1}{R(t)} \dfrac{\dd t }{t - z}, \qquad z \in \C \setminus (\gamma_0 \cup [-\ee^c, -1] \cup \{0\}),
        \label{eq:psi-integral}
    \end{equation}
    satisfies RHP \ref{rhp:psi}(a, b) and \eqref{eq:psi-endpoint-decay} -- \eqref{eq:psi-interval-decay-2}. To verify RHP \ref{rhp:psi}(c) and \eqref{eq:psi-infnity-residue} requires the computation of the integrals 
    \[
        \int_{-\ee^c}^{-1} \dfrac{\dd t}{R(t)} = -c \qandq \int^{-1}_{-\ee^c} \dfrac{\dd t}{tR(t)} = c\ee^{-\frac{c}{2}},
    \]
    both of which are elementary and can be carried out explicitly. One can analogously verify that the right hand side of \eqref{eq:psi-def} solves RHP \ref{rhp:psi}. Uniqueness of the solution of the RHP then implies the equivalence of \eqref{eq:psi-def} and \eqref{eq:psi-integral}.
\end{proof}

\begin{proof}[Proof of Proposition \ref{prop:prob-measure}]
For $z = \ee^{\frac{c}{2}} \ee^{\ii \theta}$, we have 
\[
\dd \mu (z) = \frac{1}{\pi} \ee^{\frac{c}{2}} \ee^{\ii \theta} \psi_-(\ee^{\frac{c}{2}} \ee^{\ii \theta}) \dd \theta.
\]
Thus, in view of \eqref{eq:psi-def}, showing positivity of $\mu$ is equivalent to proving that 
\begin{equation}
    \log \left( \dfrac{1 + Q_-(z)}{1 - Q_-(z)}\right) > 0 \qforq z \in \gamma_0\setminus \{z_+, z_-\}.
    \label{eq:zpsi-positivity-support}
\end{equation}
We begin by noting that for $z = \ee^{\frac{c}{2}} \ee^{\ii \theta}$ with $\theta \in [-\theta_c, \theta_c]$, 
\[
    Q^2(z) = t^2 \dfrac{\sin \left( \dfrac{\theta_c - \theta}{2} \right)\sin \left( \dfrac{\theta_c + \theta}{2} \right)}{\cos^2(\theta/2)} \in [0, 1),
\]
which implies $Q_{\pm}(z) \in \R$. Consequently, it follows from the choice of branch of $R(z)$ that $Q_-(z) \in [0, 1)$ and \eqref{eq:zpsi-positivity-support} immediately follows.
Thus, $\mu$ is a positive measure. 
It remains to verify that $\mu$ is a probability measure. This follows from \eqref{eq:psi-zero-residue} and \eqref{eq:psi-infnity-residue}. Indeed, let $C$ be a contour encircling $\gamma_0$ and not $z = 0$ nor $[-\ee^c, -1]$ and $\widetilde{C}$ be a contour encircling $[-\ee^c, -1]$ and not $z = 0$ or $\gamma_0$. Then, it follows from \eqref{eq:psi-jump-interval} that 
\begin{equation}
       \dfrac{1}{2\pi \ii } \int_{\widetilde{C}} \psi(z) \dd z = \dfrac{1}{2\pi \ii } \int_{-\ee^c}^{-1} (\psi_- - \psi_+)(x) \dd x =  -\int_{-\ee^c}^{-1} \frac{1}{cx} \dd x = 1
       \label{eq:psi-integral-interval}
\end{equation}
and thus, by a residue calculation, 
\begin{equation}
    \int_{\gamma_0} \dd \mu(z) = \dfrac{1}{2\pi \ii } \int_{C} \psi(z) \dd z = -\dfrac{1}{2\pi \ii } \int_{\widetilde{C}} \psi(z) \dd z - \res_{z = 0} \psi(z) -\res_{z = \infty} \psi(z) = 1.
\end{equation}

\end{proof}


\subsection{Properties of the \texorpdfstring{$g$}{g}-function}
With the measure $\mu(z)$ in hand, we now derive properties of $g(z)$, defined in \eqref{eq:g-def}, which will be crucial for the asymptotic analysis in Section \ref{sec:proof-polynomial-asymptotic}, see Proposition \ref{prop:g-fun} below. First, observe that 
\begin{equation}
    g'(z) = \int \dfrac{\dd \mu(t)}{z - t}, \quad z \in \C \setminus \gamma_0.
    \label{eq:g-prime}
\end{equation}
Then, since $\mu(z)$ is a probability measure, we have 
\begin{equation}
   g'(z) = \dfrac{1}{z} + \Oo(z^{-2}).
   \label{eq:g-prime-infty}
\end{equation}
Furthermore, by the Plemelj-Sokhotski formula (cf. \cite{gakhov}*{Chapter 1}), we have 
\begin{equation}
    g'_+(z) - g'_-(z) = -2\psi_-(z), \quad z \in \gamma_0. 
    \label{eq:g-prime-jump}
\end{equation}
To compute $g'(z)$ in terms of $V'(z), \psi(z)$, we record some basic facts about $V(z)$. Recall that we have defined $V(z)$ in \eqref{eq:V-def}; this expression is well-defined when $z \in \gamma \setminus \{ -\ee^{\frac{c}{2}} \}$, whereas when $z = -\ee^{\frac{c}{2}}$ it must be interpreted as 
\begin{equation}
    V(-\ee^{\frac{c}{2}}) = -2 \int_0^{\frac12} \log \left( 1 - {\ee^{c \left(t - \frac12 \right)}}\right) \dd t -2 \int_{\frac12}^1\log_+ \left( 1 - {\ee^{c \left(t - \frac12 \right)}}\right) \dd t. 
    \label{eq:V-discontinuity}
\end{equation}
In spite of the appearance of the discontinuity of the logarithm in \eqref{eq:V-discontinuity}, it can be directly checked that $\exp\{ N V(z) \}$ is continuous. The second expression of $V(z)$ in \eqref{eq:V-def} follows from the integral representation \cite{DLMF}*{Eq. 25.12.2} of the dilogarithm. Using the integral representation of $V(z)$ or the known jumps of the dilogarithm, 
\[
\Li[2, +](x) - \Li[2,-](x) = 2\pi \ii \log|x|, \quad x \in (1, \infty),
\]
where the $\pm$ signs correspond to orienting $(1, \infty)$ from left to right, one can check that 
\begin{equation}
    V_+(x) - V_-(x) = \begin{cases}
        4\pi \ii, & x \in (-1, 0), \medskip \\
        4\pi \ii \left( 1 - \dfrac1c \log|x|  \right), & x \in (-\ee^c, -1).
    \end{cases}
    \label{eq:V-jump}
\end{equation}
Furthermore, it follows from the series definition that 
\begin{equation}
    \dod{}{z} \Li(z) = -\frac1z \log(1 - z).
    \label{eq:Li-derivative}
\end{equation}
Using this, one obtains 
\begin{equation}
    V'(z) = \frac2c \left( \frac1z \log \left( 1 + \frac{\ee^c}{z} \right)-\frac1z \log\left( 1 + \frac1z \right)  \right) = \dfrac{2}{cz} \log \left( \dfrac {z + \ee^c}{z + 1} \right),
    \label{eq:V-prime}
\end{equation}
where the second equality follows from the choice of branch of the logarithm. It follows that $V'(z)$ is analytic in $\C \setminus \left( [-\ee^c, -1] \cup \{0\}\right)$ with a simple pole at $z = 0$ and satisfies 
\begin{equation}
    V'_+(x) - V'_-(x) = -\frac{4\pi \ii}{cx}, \quad x \in (-\ee^c, -1),
    \label{eq:V-prime-jump}
\end{equation}
where the $\pm$ signs correspond to orienting $[-\ee^c, -1]$ from left to right. 

\begin{lemma} 
    With $g(z), \psi(z)$, and $V(z)$ as in \eqref{eq:g-def}, \eqref{eq:psi-def}, and \eqref{eq:V-def}, respectively, we have the identity
    \[
    g'(z) = \frac12 V'(z) + \psi(z), \quad z \in \C \setminus \gamma_0.
    \]
    \label{lemma:g-prime-identity}
\end{lemma}
\begin{proof}
    Let 
    \[
    f(z) = g'(z) - \frac12 V'(z) - \psi(z), \quad z \in \C \setminus (\gamma_0 \cup [-\ee^c, -1]).
    \]
    Then, $f(z)$ is analytic in the specified domain except at $z = 0$. At $z = 0$, by \eqref{eq:V-prime} and \eqref{eq:psi-zero-residue}, $\res_{z = 0} \frac12 V'(z) = 1 = -\res_{z = 0} \psi(z)$. Thus, we conclude that $f(z)$ is analytic on the stated domain. It follows from \eqref{eq:psi-jump-arc}, \eqref{eq:g-prime-jump} that $f(z)$ is continuous across $\gamma_0$ and, by Morera's Theorem, must be analytic across $\gamma_0 \setminus \{z_\pm\}$. It is clear from the definition of $f(z)$ that it is bounded at $z = z_{\pm}$ and thus those singularities are removable. Similarly, it follows from \eqref{eq:psi-jump-interval} and \eqref{eq:V-prime-jump} that $f(z)$ is analytic across $(-\ee^c, -1)$. Since $f(z)$ has at most logarithmic singularities at $z = -\ee^c$ and $z = -1$, these singularities are removable as well. Thus, $f(z)$ is analytic and bounded in $\C$. Finally, it follows from the definition of $g'(z)$, $V'(z)$ and \eqref{eq:psi-infnity-residue} that $\lim_{z \to \infty} f(z) = 0$. Thus, by Liouville's Theorem, $f(z) \equiv 0$ and the result follows. 
\end{proof}
We are now ready to prove the main proposition of this subsection: 
\begin{proposition}
    Let $g(z)$ be as in \eqref{eq:g-def}, set $\ell \equiv \ell_c := -2 g(z_+) + V(z_+)$, and define  \begin{equation}
    \phi(z) \equiv \phi_c(z) := \int_{z_+}^z \psi(s) \dd s, \quad z \in \C \setminus \left( (-\infty, 0] \cup \{ \ee^{\frac{c}{2}} \ee^{\ii \theta} \ : \ \theta \in [-\pi, \theta_c] \} \right),
    \label{eq:phi-def}
\end{equation}
where the contour of integration is chosen to be contained in the indicated set. Then,
    \begin{equation}
        g(z) = \dfrac{1}{2} V(z) - \frac{\ell}{2} + \phi(z),
        \label{eq:g-V-phi}
    \end{equation}
   for all $z$ in the domain of analyticity of $\phi(z)$. Furthermore, with $\Gamma$ as in \eqref{eq:g-def}, 
   \begin{align}
      \label{eq:g-phi-jump-1} g_+(z) + g_-(z) - V(z) + \ell &= 0, \quad z \in \gamma_0, \\
      \label{eq:g-phi-jump-2} g_+(z) - g_-(z) - 2\phi_+(z) &= 0, \quad z \in \gamma_0,\\
      \label{eq:g-phi-jump-3}g_+(z) - g_-(z) &= 2\pi \ii, \quad z \in \Gamma \setminus \gamma_0.
   \end{align}
   Finally, we have 
   \begin{equation}
       \re \left(  g_+(z) + g_-(z) - V(z) + \ell \right) <  0, \quad z \in \gamma \setminus (\gamma_0 \cup \{z_{\pm}\}). 
       \label{eq:g-phi-ineq-off-arc}
   \end{equation}
   \label{prop:g-fun}
\end{proposition}
\begin{proof}
Integrating the identity in Lemma \ref{lemma:g-prime-identity} with a path of integration contained in the domain specified in \eqref{eq:phi-def}, we arrive at \eqref{eq:g-V-phi}. The jump conditions \eqref{eq:g-phi-jump-1}, \eqref{eq:g-phi-jump-2} follow from \eqref{eq:g-V-phi}. Jump condition \eqref{eq:g-phi-jump-3} follows from definition \eqref{eq:g-def}, the choice of branch of the logarithm, and the fact that $\mu(z)$ is a probability measure. 

Finally, to see \eqref{eq:g-phi-ineq-off-arc}, we note that 
\[
    \re(g_+(z) + g_-(z) - V(z) + \ell ) = 2\re(\phi(z)), \quad z \in (\gamma \setminus \gamma_0) \cap \C_+. 
\]
Writing $z = \ee^{\frac{c}{2}} \ee^{\ii \theta}$, $\theta \in (\theta_c, \pi)$, we have
\[
    \phi(z) = \ii \int_{\theta_c}^\theta \ee^{\frac{c}{2}} \ee^{\ii \tau}\psi(\ee^{\frac{c}{2}} \ee^{\ii \tau})  \dd \tau,
\]
and it follows from \eqref{eq:Q-squared} that 
\[
    \re \left( \ee^{\frac{c}{2}} \ee^{\ii \tau}\psi(\ee^{\frac{c}{2}} \ee^{\ii \tau}) \right) = 0  , \qquad \tau \in (\theta_c, \pi),
\]
and, by Lemma \ref{lemma:pre-image-psi}(c), it follows that 
\begin{equation}
    \im \left( \ee^{\frac{c}{2}} \ee^{\ii \tau}\psi(\ee^{\frac{c}{2}} \ee^{\ii \tau})\right)  > 0, \qquad \tau \in (\theta_c, \pi).
    \label{eq:re-ineq-intermediate}
\end{equation}

Thus, $\re(\phi(z)) < 0$ for $z \in (\gamma \setminus \gamma_0) \cap \C_+$. 

Similarly, for $z \in (\gamma \setminus \gamma_0) \cap \overline{\C_-}$, we have 
\[
    \re(g_+(z) + g_-(z) - V(z) + \ell ) = \re(\phi_+(z) + \phi_-(z)). 
\]
Denoting $z = \ee^{\frac{c}{2}} \ee^{-\ii \theta}, \ \theta \in (\theta_c, \pi]$, and recalling the choice of contour of integration, we have 
\begin{equation}
    \phi_\pm(z) = \ii \int_{-\theta_c}^{-\theta}  \ee^{\frac{c}{2}} \ee^{\ii \tau}\psi(\ee^{\frac{c}{2}} \ee^{\ii \tau})  \dd \tau  \mp \int_{z_+}^{z_-} \psi_-(x) \dd x,
    \label{eq:jump-calc-1}
\end{equation}
and so 
\begin{equation}
    \phi_+(z) + \phi_-(z) = 2\ii \int_{-\theta_c}^{-\theta}  \ee^{\frac{c}{2}} \ee^{\ii \tau}\psi(\ee^{\frac{c}{2}} \ee^{\ii \tau})  \dd \tau. 
    \label{eq:jump-calc}
\end{equation}
By \eqref{eq:Q-squared} and Lemma \ref{lemma:pre-image-psi}(d), we have 
\[
\re\left(\ee^{\frac{c}{2}} \ee^{\ii \theta}\psi(\ee^{\frac{c}{2}} \ee^{\ii \theta})\right) = 0 \qandq \im \left(\ee^{\frac{c}{2}} \ee^{\ii \theta}\psi(\ee^{\frac{c}{2}} \ee^{\ii \theta}) \right) < 0, \qquad \theta \in (-\pi, -\theta_c),
\]
and the result follows.
\end{proof}

The set $\{\re(\phi(z)) = 0\}$ will play an important role in what follows, so we will describe its structure here.  It follows from the square-root vanishing of $\psi(z)$ at $z = z_\pm$ that, locally near $z = z_\pm$ the set $\{ z : \re(\phi(z)) = 0 \}$ consists of three arcs at equal angles. It follows from \eqref{eq:psi-infnity-residue}, \eqref{eq:psi-zero-residue} and the definition of $\phi(z)$ that
\begin{align*}
    \phi(z) &= \log z + \Oo (1) \qasq z \to \infty, \\
    \phi(z) &= -\log z + \Oo (1) \qasq z \to 0.
\end{align*}
Thus, the arcs making up $\{ z : \re(\phi(z)) = 0 \}$ are compact and bounded away from infinity. Since $\phi(z)$ is purely imaginary on $\gamma_0$, and $\re(\phi(z))$ is negative on $\gamma \setminus \gamma_0$, and $\re(\phi(z)) > 0$ in a small neighborhood on either side of $\gamma_0$, the two remaining arcs must intersect the real line, once inside $\gamma$ and once outside. Computational examples of the level sets of $\{ z : \re(\phi(z)) = 0 \}$ are shown in Figure \ref{fig:phi-levels}. Denote the arc inside $\gamma$ by $\gamma_{in}$ and the one outside $\gamma$ as $\gamma_{out}$, as shown in Figure \ref{fig:phi-levels}. In the sequel, we will need contours of integration which encircle the origin and are contained in the domain bounded by $\gamma_0 \cup \gamma_{out}$; that such contours exist relies on understanding the location and number of intersections of $\gamma_{out}$ with the real axis, so we record this information below.
\begin{figure}[t]
    \begin{subfigure}[b]{0.3\textwidth}
        \centering
    \begin{tikzpicture}
        \node[anchor=south west,inner sep=0] at (0,0) {\includegraphics[width = \textwidth]{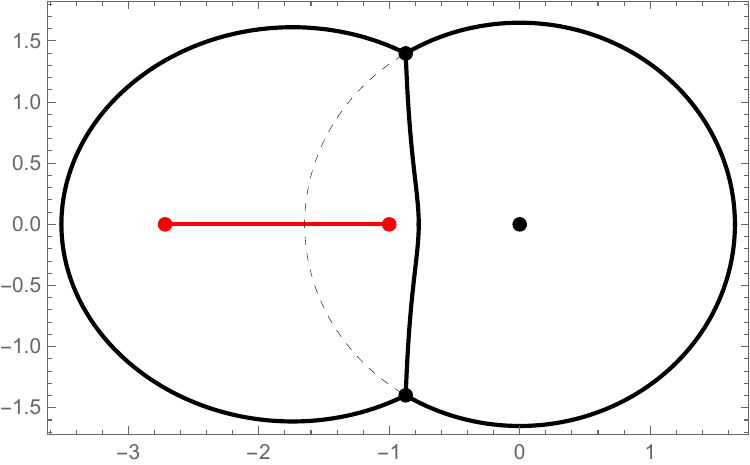}};
        \put(20,20){$+$}
        \put(40,30){$-$}
        \put(110,35){$+$}
        \put(85, 50){\tiny $\gamma_{in}$}
        \put(25, 80){\tiny $\gamma_{out}$}
        \put(130, 10){\tiny $\gamma_+$}
        \put(100, 65){\tiny $\gamma_-$}
    \end{tikzpicture}
        \caption{$c = 1$}
    \end{subfigure}
    \begin{subfigure}[b]{0.3\textwidth}
        \centering
    \begin{tikzpicture}
        \node[anchor=south west,inner sep=0] at (0,0) {\includegraphics[width = \textwidth]{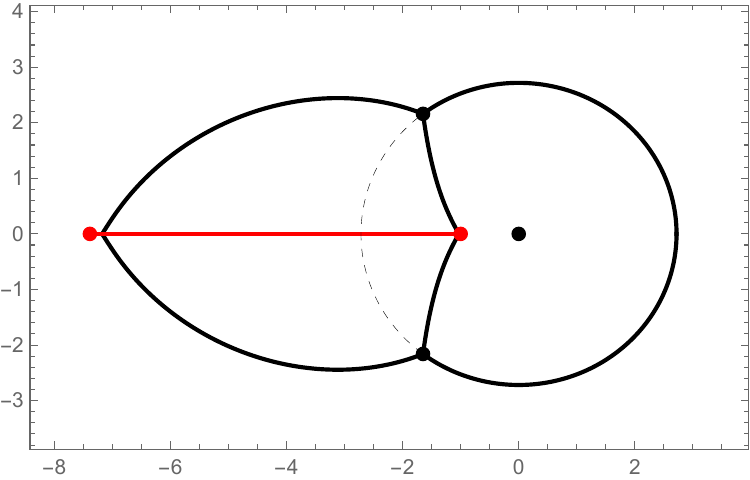}};
        \put(20,20){$+$}
        \put(40,30){$-$}
        \put(110,35){$+$}
        \put(90, 40){\tiny $\gamma_{in}$}
        \put(25, 80){\tiny $\gamma_{out}$}
        \put(130, 10){\tiny $\gamma_+$}
        \put(100, 65){\tiny $\gamma_-$}
    \end{tikzpicture}
        \caption{$c = 2$}
    \end{subfigure}
    \begin{subfigure}[b]{0.3\textwidth}
        \centering
    \begin{tikzpicture}
        \node[anchor=south west,inner sep=0] at (0,0) {\includegraphics[width = \textwidth]{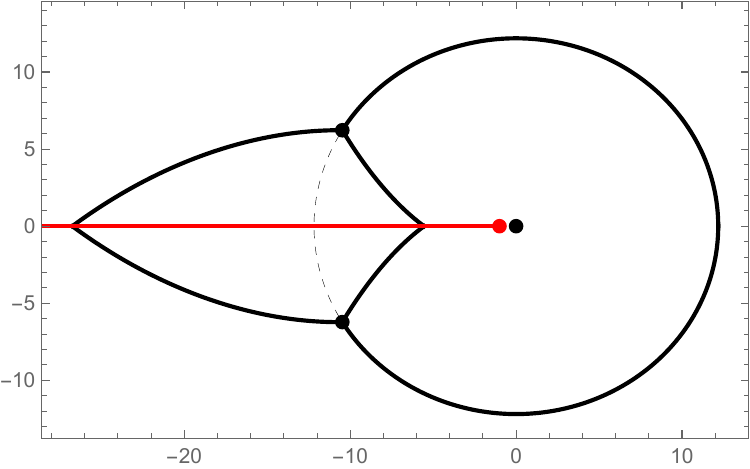}};
        \put(20,20){$+$}
        \put(50,40){$-$}
        \put(110,35){$+$}
        \put(80, 40){\tiny $\gamma_{in}$}
        \put(25, 65){\tiny $\gamma_{out}$}
        \put(130, 10){\tiny $\gamma_+$}
        \put(100, 63){\tiny $\gamma_-$}
    \end{tikzpicture}
        \caption{$c = 5$}
    \end{subfigure}
    \caption{\centering The curves $\re\left(\phi(z) \right) = 0$ (black) for various choices of $c$. The interval in red is $[-\ee^c, -1]$ and the dashed curve is the circle $\ee^{\frac{c}{2}}\T$. The signs indicate the sign of $\re(\phi(z))$ in the corresponding domain.}
    \label{fig:phi-levels}
\end{figure}
\begin{lemma}
    Let $\gamma_{in}, \gamma_{out}$ be as defined above and $\phi(z)$ as in \eqref{eq:phi-def}. Then, for $z \in \C \setminus  \{0\}$ (where $\re(\phi(z))$ is extended to $\C\setminus\{0\}$ by continuity),
    \begin{enumerate}[(a)]
        \item $\re(\phi(\overline{z})) = \re({\phi(z)}).$ In particular, $\gamma_{in}, \gamma_{out}$ are symmetric across the real axis. 
        \item $\re(\phi(z)) = \re(\phi(\ee^c/\overline{z})).$ In particular, $\gamma_{in}$ is the reflection of $\gamma_{out}$ across the circle $\gamma$. 
        \item $\gamma_{out}$ intersects the real axis exactly once at $x_{out} < -\ee^{c/2}$. Consequently, $\gamma_{in}$ intersects the real axis exactly once at $x_{in} \in (-\ee^{c/2}, 0)$. 
    \end{enumerate}
    \label{lemma:gamma-in-out}
\end{lemma}
\begin{proof}
    \begin{enumerate}[(a)]
        \item Starting from \eqref{eq:conj-sym} and integrating both sides, we conclude 
        \[
            \phi(\overline{z}) = \overline{\phi(z)} + \phi_{\pm}(z_-),
        \]
        where the positive boundary value is chosen if $|z| < \ee^{\frac c2}$ and the negative one otherwise. Taking the real part of both sides and recalling that $\re(\phi_{\pm}(z_-)) = \re(\pm \ii \pi) = 0$ (cf. Proposition \ref{prop:prob-measure}) proves the identity.
        
        \item One can directly verify, using \eqref{eq:psi-def} and some algebraic manipulations, that
        \begin{equation}
            c \frac{\ee^c}{z} \psi\left( \dfrac{\ee^c}{z} \right) = -c z\psi(z).
            \label{eq:psi-symmetry}
        \end{equation}
        Integrating \eqref{eq:psi-symmetry} yields $$\phi\left(\frac{\ee^c}{z} \right)=\phi(z)+C_\pm, $$ for some constants $C_\pm$ (one for $|z| < \ee^{\frac {c}{2}}$ and one outside the disc). Note that $z_- = \ee^{c}/{z_+}$ and $\phi(z_+) = 0$, so $C_{\pm} = \phi_{\mp}(z_-) = \mp \pi \ii$. Since $C_\pm \in \ii \R$, using the conjugation symmetry from (a) proves the desired identity.

        \item First, we already noted that the arc of $\{\re(\phi(z)) = 0\}$ outside of the circle $\ee^{c/2}\T$ remains outside the circle and intersects the real axis at least once. For $x > \ee^{c/2}$, $R(x) > 0$ and so, using \eqref{eq:psi-def}, we have $cx\psi(x) > 0$ and so
        \[
        \re(\phi(x)) =\int_{\ee^{c/2}}^x\psi(t) \dd t >0.
        \]
        In particular, $\gamma_{out}$ must intersect the real line on $(-\infty, -\ee^{c/2})$ and, by part (b), $\gamma_{in}$ intersects the real axis on $(-\ee^{c/2}, 0)$. 

        To prove uniqueness, we study $\dod{}{x} \re(\phi(x)) = \re(\psi(x)) = \frac{1}{cx}\re(cx\psi(x))$. Since $Q(x) \in \R$ when $x \in (-\infty, -\ee^{c/2})$, it follows from \eqref{eq:psi-def} that $$\text{sign}(\re(cx\psi(x))) = \text{sign}(Q(x)), \quad x \in (-\infty, -\ee^{c/2})$$ (for both boundary values on $(-\ee^c, -\ee^{c/2})$, cf. \eqref{eq:psi-jump-interval}). Since $R(x) < 0$ and $x + \ee^{c/2} < 0$ on this interval, \eqref{eq:Q-def} implies $Q(x) > 0$ and, as $cx < 0$, $\re(\phi(x))$ is strictly decreasing on $(-\infty, -\ee^{c/2})$. Thus, there can be at most one zero on this interval. Combined with the topological argument above this lemma, uniqueness of $x_{out}$ follows. Part (b) then implies the uniqueness and location of $x_{in}$.
    \end{enumerate}
\end{proof}

\section{Proof of Theorem \ref{thm:polynomial-asymptotics}}
\label{sec:proof-polynomial-asymptotic}

To obtain asymptotic formulas for polynomials $P_n(z)$ and CD kernel $R_n(w, z)$, we first approximate them by a related family of polynomials. The following is a direct consequence of the Euler-Maclaurin formula (see also \cite{MR1703273}).

\begin{proposition}
    Recall the definitions of $V(z), \nu(z), \gamma$ in  \eqref{eq:V-def}, \eqref{eq:nu-def}, \eqref{eq:curves-def}, respectively. Then, for $z \in \gamma \setminus \{-\ee^\frac{c}{2} \}$ and $q = \ee^{\frac{c}{2N}}$,
    \begin{equation}
        \exp \left \{ NV(z) + \nu(z) \right\} \prod_{j = 1}^{2N} \left( 1 + \dfrac{q^j}{z} \right) = 1 + \Oo\left( N^{-1} \right) \qasq N \to \infty,
        \label{eq:weight-approx}
    \end{equation}
    locally uniformly in $z$. Furthermore, for all $z \in \gamma$, the left hand side of \eqref{eq:weight-approx} is $\Oo(N)$ uniformly. 
    \label{prop:weight-approx}
\end{proposition}
\begin{proof}
    Fix a compact $K \subset \gamma \setminus \{-\ee^{\frac{c}{2}}\}$ and let $\gamma^\circ$ be any open sub-arc of $\gamma \setminus \{ - \ee^{\frac{c}{2}}\}$ containing $K$. Next, rewrite the left hand side:
    \begin{equation}
        \exp \left \{ NV(z) + \nu(z) \right\} \prod_{j = 1}^{2N} \left( 1 + \dfrac{q^j}{z} \right) = \exp \left\{  \sum_{j = 1}^{2N} \log \left( 1 + \frac{\ee^{c\cdot \frac{j}{2N}}}{z} \right) + N V(z)  + \nu(z) \right\} , \quad z \in \gamma \setminus \{ -\ee^{\frac c2}\},
        \label{eq}
    \end{equation}
    where $\log (\diamond)$ is the principal branch. Using the Euler-Maclaurin formula (see e.g. \cite{MR435697}*{Section 8.1}) with $f(x; z) := \log \left(1 + z^{-1} \ee^{c x} \right)$ (which is analytic on $\gamma^\circ$) we immediately see that the exponent is $\Oo(N^{-1})$, and we have \eqref{eq:weight-approx} locally uniformly for $z \in \gamma^\circ$. 

    To prove the second claim of the proposition, we first write $\gamma = \gamma_{\leq -1} \cup \gamma_{K}$, where 
    $$\gamma_{\le-1}:=\gamma\cap\{\re (z)\le-1\},\qquad\gamma_K:=\gamma\cap\{\re (z)\ge-1\}.$$
    $\gamma_K$ is compact, avoids $-\ee^{c/2}$, and meets $\gamma_{\le-1}$ in one point in each half-plane. On $\gamma_K$, Euler-Maclaurin yields \eqref{eq:weight-approx} uniformly, and so we focus on $\gamma_{\leq -1}$. Recalling \eqref{eq:V-def}, \eqref{eq:nu-def}, we have
    \begin{multline}
        \log \left| \exp \left \{ NV(z) + \nu(z) \right\} \prod_{j = 1}^{2N} \left( 1 + \dfrac{q^j}{z} \right)\right| = \sum_{j = 1}^{2N} h \left(x_j; z \right) - 2N\int_0^{1} h(x; z) \ \dd x + \re(\nu(z)), \\
       \text{ where } \qquad  x_j = \dfrac{j}{2N} \qandq  h(x; z) = \log \left|1 + \frac{\ee^{cx}}{z} \right|.
       \label{eq:main-expression}
    \end{multline}
    Note the appearance of a logarithmic singularity (and vanishing of left hand side of \eqref{eq}) exactly when $z = -\ee^{\frac c2}$. For $z \in \gamma_{\leq -1}$, $h(x; z)$ is unimodal for $x \in [0, 1]$ as can be verified from 
    \begin{equation}
        \log \left|1 + \frac{\ee^{cx}}{z} \right| = - \frac{c}{2} + \log|z + \ee^{cx}|,
         \label{eq:log-inter}
    \end{equation}
    and its unique minimum occurs at $x_0 = \frac1c \log (- \re(z) ) \leq \frac12$. Let 
    $$J=\max\left\{j\in\{0,\dots,2N\}:\ \tfrac{j}{2N}\le x_0\right\}=\big\lfloor 2Nx_0\big\rfloor\le N<2N .$$
    Then, 
    \begin{itemize}
        \item for $1\le j\le J$, $h(x;z)$ is non-increasing on $[x_{j-1},x_j]$, so $h(x_j; z)\le 2N\int_{x_{j-1}}^{x_j}h(x; z) \dd x$, and
        \item for $J<j<2N$, $h(x;z)$ is non-decreasing on $[x_j,x_{j+1}]$, so $h(x_j; z)\le 2N\int_{x_j}^{x_{j+1}} h(x; z) \dd x$.
    \end{itemize}
    Summing these estimates implies that for $J < 2N$
    \begin{equation}
        \sum_{j = 1}^{2N} h(x_j; z) \leq 2N \int_0^1 h(x;z) \dd x - 2N \int_{x_J}^{x_{J+1} } h(x; z) \ \dd x + h(1; z).
        \label{eq:h-int}
    \end{equation}
    In fact, since $J \leq N$, \eqref{eq:h-int} holds as written. We now seek to estimate the right hand side of \eqref{eq:h-int}. It follows from \eqref{eq:nu-def} and $ |z| = \ee^{\frac c2}$ that 
    \begin{equation}
       h(1;z)+\re (\nu(z))=-\frac c2+\frac12\log\left(|z+1|\,|z+\ee^c|\right)\le-\frac c4+\log\left(1+\ee^{\frac c2}\right).
        \label{eq:nu-est}
    \end{equation}
    For the integral on $[x_J, x_{J+1}]$, by definition of $x_0$, 
    \begin{equation}
        |\ee^{cx} + z| \geq |\ee^{cx} +\re(z)| = |\ee^{cx} - \ee^{cx_0}| \geq c|x - x_0|
        \label{eq:inter-est}
    \end{equation}
    where the last inequality follows directly from the Mean Value Theorem. An elementary calculation yields
    \[
        \int^{x_{J+1}}_{x_J} \log|x - x_0| \dd x \geq -\frac{1}{2N} \left(\log 4N + 1\right) 
    \]
    with equality when $x_0$ is exactly the midpoint of $[x_J, x_{J+1}]$. Using this estimate along with \eqref{eq:inter-est} and the definition of $h(x;z)$, we find
    \begin{equation}
        -2N \int^{x_{J+1}}_{x_J} h(x; z) \dd x \leq \log(4N) + 1 + \frac{c}{2} - \log(c). 
        \label{eq:I-int-est}
    \end{equation}
    Plugging \eqref{eq:nu-est}, \eqref{eq:I-int-est} into \eqref{eq:main-expression} and exponentiating both sides, we find 
        $$\left|\exp\{NV(z)+\nu(z)\}\prod_{j=1}^{2N}\Big(1+\frac{q^j}{z}\Big)\right|\le\frac{4\,\ee^{1+\frac c4}\,\big(1+\ee^{\frac c2}\big)}{c}\;N .$$
    Combining this estimate with the uniform boundedness on $\gamma_K$ yields the result.

\end{proof}
For definiteness, though it will not matter in the end, we fix $\nu(-\ee^{\frac{c}{2}}) :=  \nu_+(-\ee^{\frac{c}{2}})$ and analyze the monic polynomials $\widehat{P}_n(z) \equiv \widehat{P}_n(z; c, N)$ satisfying (recall that $\gamma = \ee^{\frac{c}{2}} \T$)
\begin{equation}
    \int_{\gamma} z^k \widehat{P}_n(z; c, N)  \ee^{-NV(z) - \nu(z)} \dd z = 0, \qforq k = 0, 1, ..., n-1.
    \label{eq:approximate-ortho}
\end{equation}
To obtain asymptotic formulas for $\widehat P_n(z)$, we exploit their Riemann-Hilbert problem characterization first obtained by Fokas, Its, and Kitaev in \cite{FIK}. 

\begin{rhp}[Approximate Riemann-Hilbert Problem]
    Fix $c >0$ and seek a $2 \times 2$ matrix-valued function $\widehat{\mb Y}_n(z;c, N) \equiv \widehat{\mb Y}(z;c, N)$ satisfying the following conditions:
    \begin{enumerate}[(a)]
        \item $\widehat{\mb Y}(z;c, N)$ is analytic in $\C \setminus \gamma$, 
        \item $\widehat{\mb Y}(z;c, N)$ has continuous boundary values on $\gamma\setminus \{-\ee^{c/2}\}$ satisfying 
        \[
          \widehat{\mb Y}_+(z;c, N) = \widehat{\mb Y}_-(z;c, N) \begin{bmatrix}
                1 & \ee^{-NV(z) - \nu(z) } \\ 0 & 1
            \end{bmatrix}, 
        \]
        \item  As $z \to -\ee^{c/2}$, 
        \begin{equation}
        \widehat{\mb Y}(z;c, N) = \begin{bmatrix}\Oo(1) & \Oo(\log|z + \ee^{c/2}|) \\ \Oo(1) & \Oo(\log|z + \ee^{c/2}|) \end{bmatrix}.
        \label{eq:Y-log}
        \end{equation}
        Furthermore, as $z \to \infty$, we have 
        \begin{equation}
            \widehat{\mb Y}(z;c, N) = \left( \I + \dfrac{\widehat{{\mb Y}}^{(1)}}{z} + \mc O \left( z^{-2} \right) \right)z^{n \sigma_3}. 
        \end{equation}
    \end{enumerate}
    \label{rhp:initial}
\end{rhp}
A standard argument using Liouville's theorem implies that if a solution to Riemann-Hilbert problem \ref{rhp:initial} exists, it must be unique and must satisfy $\det \widehat {\mb Y}(z; c, N) \equiv 1$. Denote the Cauchy transform of a function $f(z)$ by
\[
\mc C[f](z) := \dfrac{1}{2\pi \ii}\int_{\gamma} \dfrac{f(t)}{t - z} \dd t.
\]
Since the weight of orthogonality in \eqref{eq:approximate-ortho} is not positive, it is not a priori clear whether the polynomials $\widehat{P}_n(z)$ are of degree $n$. However, under the assumptions 
\begin{equation}
    \deg \widehat{P}_n(z) = n \qandq \mc C[\widehat{P}_{n -1}(z) \ee^{-NV(z) - \nu(z)}](z) \sim z^{-n} \qasq z \to \infty,
    \label{eq:rhp-initial-assumptions}
\end{equation}
then, the unique solution of Riemann-Hilbert problem \ref{rhp:initial} is given by
\begin{equation}
   \widehat{\mb Y}(z;c, N) = \begin{bmatrix}
        \widehat{P}_n(z; c, N) & \mc C[\widehat{P}_n  \ee^{-NV(z) - \nu(z)}](z; c, N) \medskip \\ -2\pi \ii\widehat{\kappa}^{-1}_{n- 1} \widehat{P}_{n - 1}(z; c,N) & -2\pi \ii \widehat{\kappa}^{-1}_{n- 1} \mc C [\widehat{P}_{n - 1} \ee^{-NV(z) - \nu(z)}](z; c, N)
    \end{bmatrix},
    \label{eq:Y-hat-formula}
\end{equation}
where 
\begin{equation}
    \widehat{\kappa}_{n-1}(c, N) := \int_{\gamma} \widehat{P}_{n-1}^2(z; c, N) \ee^{-NV(z) - \nu(z)} \dd z.
    \label{eq:kappa-hat-n}
\end{equation}
In the following sections, we will consider the specialization of Riemann-Hilbert problem \ref{rhp:initial} where $n = N$ and drop the dependence on $N, c$ for brevity.

\subsection{Normalized Riemann-Hilbert problem}
We now carry out the Deift-Zhou non-linear steepest descent method \cite{MR1207209} to obtain an asymptotic formula for $\widehat P_n(z)$. This involves a sequence of invertible transformations, the first of which is to transform $\widehat{\mb Y}(z)$ to a matrix which behaves like the identity at infinity. To this end, let 
\begin{equation}
    \mb T(z) := \ee^{\frac{N \ell}{2} \sigma_3} \widehat{\mb Y}(z) \ee^{-{N g(z)} \sigma_3} \ee^{-\frac{N \ell}{2} \sigma_3}.
    \label{eq:T-def}
\end{equation}
Then, it follows from Proposition \ref{prop:g-fun} and \eqref{eq:g-prime-infty} that $\mb T(z)$ solves the following Riemann-Hilbert problem:
\begin{rhp}
    Seek a $2 \times 2$ matrix-valued function $\mb T(z)$ satisfying the following conditions: 
    \begin{enumerate}[(a)]
        \item $\mb T(z)$ is analytic in $ \C \setminus \gamma$, 
        \item $\mb{T}(z)$ has continuous boundary values on $\gamma\setminus \{-\ee^{c/2}\}$ satisfying
        \begin{equation}
            \mb{T}_+(z) = \mb T_-(z) \begin{cases}
                \begin{bmatrix}
                    \ee^{-2N \phi_+(z)} & \ee^{-\nu(z)} \\ 0 & \ee^{-2N \phi_-(z)}
                \end{bmatrix}, & z \in \gamma_0, \medskip \\
                \begin{bmatrix}
                    1 & \ee^{-\nu(z)} \ee^{2N \phi(z)} \\ 0 & 1
                \end{bmatrix}, & z \in \gamma \setminus \gamma_0. 
            \end{cases}
        \end{equation}
        \item As $z \to -\ee^{c/2}$, $\mb T(z)$ satisfies \eqref{eq:Y-log}. Furthermore, 
        \begin{equation}
            \mb T(z) = \I +   \dfrac{\mb T_1}{z} + \Oo(z^{-2}) \qasq z \to \infty. 
        \end{equation}
    \end{enumerate}
\end{rhp}
\subsection{Opening lenses}
The factorization identity 
\begin{equation}
    \begin{bmatrix}
        \ee^{-2N \phi_+(z)} & \ee^{-\nu(z)} \\ 0 & \ee^{-2N \phi_-(z)}
    \end{bmatrix} = \begin{bmatrix}
        1 & 0 \\ \ee^{-2N \phi_-(z) + \nu(z)} & 1
    \end{bmatrix} \begin{bmatrix}
        0 & \ee^{-\nu(z)} \\ -\ee^{\nu(z)} & 0
    \end{bmatrix}
    \begin{bmatrix}
        1 & 0 \\ \ee^{-2N \phi_+(z) + \nu(z)} & 1
    \end{bmatrix}, \quad z \in \gamma_0,
\end{equation}
motivates the next transformation. Let $\gamma_{\pm}$ be two smooth arcs on the left/right sides of $\gamma_0$ contained (except for their end-points) in the set $\{z : \re(\phi(z)) >0 \}$ and entering $z = z_{\pm}$ at fixed angles to be prescribed later (cf. Section \ref{subsec:conf-map}).

Define
\begin{equation}
    \mb S(z) := \mb T(z) \begin{cases}
        \begin{bmatrix}
            1 & 0 \\ \mp \ee^{-2N\phi(z) + \nu(z)} & 1
        \end{bmatrix} & \text{in the region between $\gamma_{\pm}$ and $\gamma_0$}, \medskip \\
        \I, & \text{otherwise}.
    \end{cases} 
    \label{eq:S-def}
\end{equation}
Then, $\mb S(z)$ solves the following RHP:
\begin{rhp}
   Seek a $2 \times 2$ matrix-valued function $\mb S(z)$ satisfying the following conditions: 
    \begin{enumerate}[(a)]
        \item $\mb S(z)$ is analytic in $ \C \setminus (\gamma \cup \gamma_{\pm })$, 
        \item $\mb{S}(z)$ has continuous boundary values on $(\gamma \cup \gamma_{\pm } )\setminus \{-\ee^{c/2}\}$ satisfying
        \begin{equation}
            \mb{S}_+(z) = \mb S_-(z) \begin{cases}
                \begin{bmatrix}
                    0 & \ee^{-\nu(z)} \\  -\ee^{\nu(z)} & 0
                \end{bmatrix}, & z \in \gamma_0, \medskip \\
                 \begin{bmatrix}
            1 & 0 \\  \ee^{-2N\phi(z) + \nu(z)} & 1
        \end{bmatrix}, & z \in \gamma_{\pm} \medskip \\
                \begin{bmatrix}
                    1 & \ee^{-\nu(z)} \ee^{2N \phi(z)} \\ 0 & 1
                \end{bmatrix}, & z \in \gamma \setminus \gamma_0. 
            \end{cases}
        \end{equation}
        \item $\mb S(z)$ satisfies the same estimates as in \eqref{eq:Y-log} and has the following asymptotic expansion: 
        \begin{equation}
            \mb S(z) = \I +   \dfrac{\mb T_1}{z} + \Oo(z^{-2}) \qasq z \to \infty. 
        \end{equation}
    \end{enumerate}
    \label{rhp:S}
\end{rhp}
This problem is fairly standard and a solution can be approximated in patches; two parametrices near the end-points $z = z_\pm$ and one global parametrix. We start with the latter.

\subsection{Global parametrix}
By dropping the jumps which are close to identity from Riemann-Hilbert problem \ref{rhp:S}, we find the following Riemann-Hilbert problem.
\begin{rhp}
    Seek a $2\times 2$ matrix-valued function $\mb N(z)$ satisfying the following conditions:
    \begin{enumerate}[(a)]
        \item $\mb N(z)$ is analytic in $ \C \setminus \gamma_0$
        \item $\mb N(z)$ has continuous boundary values on $\gamma_0$ satisfying 
        \begin{equation}
            \mb N_+(z) = \mb N_-(z) \begin{bmatrix}
                0 & \ee^{-\nu (z)} \\ -\ee^{\nu (z)} & 0
            \end{bmatrix}
        \end{equation}
        \item $\mb N(z)$ satisfies the following asymptotic expansion: 
        \begin{equation}
            \mb N(z) = \I +   \dfrac{\mb N_1}{z} + \Oo(z^{-2}) \qasq z \to \infty, 
        \end{equation}
        \item $\mb N (z) = \Oo (|z - z_\pm|^{-1/4})$ as $z \to z_\pm$. 
    \end{enumerate}
    \label{rhp:global}
\end{rhp}
This is a, by now, standard problem that can be (uniquely) solved using the Szeg\H{o} function defined in \eqref{eq:szego-fun} and the auxiliary function $a(z)$ defined in \eqref{eq:a-fun}. Indeed, recalling the definition of the Szeg\H{o} function \eqref{eq:szego-fun}, it is a straightforward calculation to check that $\varsigma(z)$ solves the following scalar Riemann-Hilbert problem 
\begin{rhp}
Find a scalar function $\varsigma(z)$ satisfying the following conditions:
    \begin{enumerate}[(a)]
    \item $\varsigma(z)$ is analytic in $\C \setminus \gamma_0$. 
    \item $\varsigma(z)$ has continuous boundary values on $\gamma_0$ which satisfy 
    \begin{equation}
        \varsigma_+(z) \varsigma_-(z) = \ee^{-\nu(z)}.
            \label{eq:szego-jump}
    \end{equation}
    \item $\varsigma(z) = \Oo(1)$ as $z \to z_\pm$.
    \item $\lim_{z \to \infty} \varsigma(z)$ exists and is non-vanishing. 
    \end{enumerate}
    \label{rhp:szego}
\end{rhp} 
Next, let
\begin{equation}
    \mb P^{(\infty)}(z) := \begin{bmatrix}
        \dfrac{1}{2}\left(a(z) + \dfrac{1}{a(z)} \right) & \dfrac{1}{2\ii} \left( a(z) - \dfrac{1}{a(z)} \right)\medskip \\
        -\dfrac{1}{2\ii} \left( a(z) - \dfrac{1}{a(z)} \right) & \dfrac{1}{2}\left(a(z) + \dfrac{1}{a(z)} \right) 
    \end{bmatrix}. 
\end{equation}
It follows from the definition of $a(z)$ that $\mb P^{(\infty)}(z)$ is the unique solution to the following Riemann-Hilbert problem: 
\begin{rhp}
Find a $2\times 2$ matrix-valued function $\mb P^{(\infty)}(z)$ satisfying the following conditions:
    \begin{enumerate}[(a)]
    \item $\mb P^{(\infty)}(z)$ is analytic in $\C \setminus \gamma_0$. 
    \item $\mb P^{(\infty)}(z)$ has continuous boundary values on $\gamma_0$ which satisfy 
    \begin{equation}
        \mb P^{(\infty)}_+(z)  =\mb P^{(\infty)}_-(z) \begin{bmatrix}
            0 & 1 \\ -1 & 0
        \end{bmatrix}
    \end{equation}
    \item $\mb P^{(\infty)}(z) = \mb I + \Oo(z^{-1})$ as $z \to\infty$.
    \end{enumerate}
    \label{rhp:global-basic}
\end{rhp}
It is now a matter of simple checks to see that the solution to Riemann-Hilbert problem \ref{rhp:global} is given by
\begin{equation}
    \mb N(z) := \varsigma^{\sigma_3} (\infty) \mb P^{(\infty)}(z) \varsigma^{-\sigma_3} (z).
\end{equation}

To arrive at Riemann-Hilbert problem \ref{rhp:global}, we dropped jumps on $\gamma_{\pm}$ which are close to the identity. However, this is not true uniformly on $\gamma_{\pm}$, and thus we need to solve Riemann-Hilbert problem \ref{rhp:S} exactly in small neighborhoods of $z = z_{\pm}$, which we do in the next section.

\subsection{Local parametrices}
Let $U^{(\pm)}_\delta = \{z \ : |z - z_\pm| < \delta \}$, where the radius $\delta$ is to be determined later. We now construct the (standard) solution to the following Riemann-Hilbert problem. 
\begin{rhp}
    Seek a $2\times 2$ matrix-valued function $\mb P^{(\pm)}(z)$ satisfying the following conditions:
    \begin{enumerate}[(a)]
    \item $\mb P^{(\pm)}(z)$ is analytic in $U^{(\pm)}_\delta \setminus (\gamma \cup \gamma_\pm)$. 
    \item $\mb P^{(\pm)}(z)$ has continuous boundary values on $\gamma \cup \gamma_\pm$ which satisfy 
    \begin{equation}
        \mb P^{(\pm)}_+(z)  =\mb P^{(\pm)}_-(z)  \begin{cases}
                \begin{bmatrix}
                    0 & \ee^{-\nu(z)} \\  -\ee^{\nu(z)} & 0
                \end{bmatrix}, & z \in \gamma_0, \medskip \\
                 \begin{bmatrix}
            1 & 0 \\  \ee^{-2N\phi(z) + \nu(z)} & 1
        \end{bmatrix}, & z \in \gamma_{\pm} \medskip \\
                \begin{bmatrix}
                    1 & \ee^{-\nu(z)} \ee^{2N \phi(z)} \\ 0 & 1
                \end{bmatrix}, & z \in \gamma \setminus \gamma_0. 
            \end{cases}
    \end{equation}
    \item $\mb P^{(\pm)}(z) =  \mb N(z) \left(\mb I + \Oo(N^{-1}) \right)$ as $N \to\infty$ uniformly for $z \in \partial U^{(\pm)}_\delta$  .
    \end{enumerate}
    \label{rhp:local}
\end{rhp}
We will construct $\mb P^{(+)}(z)$ explicitly. The construction of $\mb P^{(-)}(z)$ is analogous and is thus omitted.

\subsubsection{Conformal Map}\label{subsec:conf-map} It follows from the definition of $\phi(z)$ and the first equality in \eqref{eq:psi-def} that $|\phi(z)| \sim |z - z_+|^{\frac32}$ as $z \to z_{+}$. Furthermore, for $z \in \gamma_0$, we have 
\begin{equation}
    \phi_{\pm}(z) = \pm \pi \ii \mu([z, z_+]) = -\pi \ee^{\pm \frac{3\pi \ii}{2}} \mu([z, z_+])
    \label{eq:phi-jump-arg}
\end{equation}
where $[z, z_+]$ is the subarc of $\gamma_0$ connecting $z, z_+$. Thus, for a small enough $\delta>0$, one may define a branch of 
\begin{equation}
    \zeta^+(z) := \left(-\frac32 \phi(z) \right)^{\frac23}
\label{eq:zeta-phi}    
\end{equation}
which is analytic and conformal on a domain containing $\overline{U_\delta^{(+)}}$. It is clear from \eqref{eq:phi-jump-arg} that $\zeta^+(z) \in \R_-$ when $z \in U^{(+)}_\delta \cap \gamma_0$. Similarly, \eqref{eq:g-phi-ineq-off-arc} implies that $\zeta^+(z) \in \R_+$ when $z \in U^{(+)}_\delta \cap (\gamma \setminus \gamma_0)$. Since we had a freedom in choosing $\gamma_\pm$, we now choose these so that $\zeta^+(\gamma_\pm) \subset \ee^{\pm \frac{2\pi \ii}{3}} \R_+$. 

\subsubsection{Model problem} \label{subsec:airy-para} Let $\mb A(\zeta)$ be the unique solution to the following standard Riemann-Hilbert problem 

\begin{rhp}
    Find a $2\times 2$ matrix-valued function $\mb A(\zeta )$ satisfying the following conditions:
\begin{enumerate}[(a)]
    \item $\mb A(\zeta )$ is analytic in $\C \setminus (\R \cup \ee^{\pm \frac{2\pi \ii}{3}} \R_+)$, 
    \item $\mb A(\zeta )$ has continuous boundary values on $\R \cup\ee^{\pm \frac{2\pi \ii}{3}} \R_+$ which satisfy
    \[
    \mathbf A_+(\zeta) = \mathbf A_-(\zeta) \left\{
    \begin{array}{ll}
    \begin{bmatrix} 0 & 1 \\ -1 & 0 \end{bmatrix}, & \zeta\in (-\infty,0), \medskip \\
    \begin{bmatrix} 1 & 0 \\ 1 & 1 \end{bmatrix}, & \zeta\in \ee^{\pm \frac{2\pi \ii}{3}} \R_+, \medskip \\
    \begin{bmatrix} 1 & 1 \\ 0 & 1 \end{bmatrix}, & \zeta\in (0,\infty),
    \end{array}
    \right.
    \]
    where the real line is oriented from $-\infty$ to $\infty$ and the rays $\ee^{\pm \frac{2\pi \ii}{3}} \R_+$ are oriented towards the origin. 
    \item $\mb A(\zeta)$ has the following asymptotic expansion:
    \begin{equation}
    \mathbf A(\zeta) =  \frac{\zeta^{-\sigma_3/4}}{\sqrt2} \begin{bmatrix} 1 & \ii \\  \ii & 1 \end{bmatrix}\left( \I + \Oo \left( \zeta^{-\frac{3}{2}}\right)\right) \ee^{-\frac23\zeta^{3/2}\sigma_3} \qasq \zeta \to \infty.
    \label{eq:airy-asymptotics}
    \end{equation}
\end{enumerate}
\label{rhp:airy}
\end{rhp}
Matrix $\mb A(\zeta)$ can be written explicitly in terms of Airy functions, but we will not use this expression so we omit it; the interested reader can consult \cite{MR1677884}*{Section 7.6}, for example. 

\subsubsection{Construction of \texorpdfstring{$\mb P^{(+)}(z)$}{solution to local RHP}}
Let 
\begin{equation}
    \mb P^{(+)}(z) := \mb E^{(+)}(z) \mb A \left(N^{\frac23} \zeta^+(z) \right) \ee^{\frac12 \nu(z)\sigma_3} \ee^{- N\phi(z)\sigma_3},
    \label{eq:local-def}
\end{equation}
where $\mb E^{(+)}(z)$ is analytic in $U^{(+)}_\delta$ and given by
\begin{equation}
    \mb E^{(+)}(z) = \mb N(z) \ee^{-\frac12 \nu(z) \sigma_3} \frac{1}{\sqrt2} \begin{bmatrix} 1 & -\ii \\  -\ii & 1 \end{bmatrix} \left( \zeta^+(z)\right)^{\frac14\sigma_3} N^{\frac16 \sigma_3}, 
    \label{eq:E-p-def}
\end{equation}
where $(\zeta^+(z))^\frac14 = \left(-\frac32 \phi(z) \right)^{\frac16}$ is the branch analytic in $U^{(+)}_\delta \setminus \gamma_0$ and positive on $\gamma\setminus \gamma_0$. It follows from conformality of $\zeta^+(z)$ and choice of branch of $\diamond^\frac14$ that for $z \in \gamma_0$, 
\begin{equation}
    (\zeta^+(z))^\frac14_+ = \ii (\zeta^+(z))^\frac14_-
\end{equation} 
Using this, one can readily check that $\mb E^{(+)}(z)$ has no jump across $\gamma_0$ and is thus analytic in $U_\delta^{(+)} \setminus \{z_+\}$. Finally, by observing that the singularity at $z = z_+$ can be at most square root singularity, we conclude that it is removable and $\mb E^{(+)}(z)$ is actually analytic in $U_\delta^{(+)}$. It is now a short exercise to verify that $\mb P^{(+)}(z)$ solves Riemann-Hilbert problem \ref{rhp:local}. 
One can construct a local parametrix $\mb P^{(-)}(z)$ in $U^{(-)}_\delta$ in an analogous manner. This yields two neighborhoods with two different radii $\delta^\pm >0$ and we choose $\delta = \min \{\delta^\pm\}$. 

\begin{remark}
    It will be important later to observe that $\mb P^{(\pm)}(z) = \Oo(N^{\frac16})$ as $N \to \infty$ uniformly for $z \in \overline{U_\delta^{(\pm)}}$. Indeed, it follows from the explicit construction of $\mb A(\zeta)$ that \eqref{eq:airy-asymptotics} holds uniformly for all $\zeta \in \C \setminus (\R \cup \ee^{\pm \frac{2\pi \ii}{3}} \R_+)$, and the same holds for the boundary values on $\R \cup \ee^{\pm \frac{2\pi \ii}{3}} \R_+$. Thus, there is some $R>0$ such that for $|\zeta| > R$  we have 
    \begin{equation}
       \left|[ \mb A(\zeta) \ee^{\frac23 \zeta^{3/2} \sigma_3}]_{1k}\right| \leq C (1 + |\zeta|)^{-\frac14}, \qandq \left|[ \mb A(\zeta) \ee^{\frac23 \zeta^{3/2} \sigma_3}]_{2k} \right| \leq C (1 + |\zeta|)^{\frac14}
       \label{eq:airy-ineq}
    \end{equation}
    Since the rows of $\mb A(\zeta)$ are piece-wise analytic functions, they are bounded on $|\zeta| < R$ and so \eqref{eq:airy-ineq} holds uniformly for all $\zeta \in \C \setminus (\R \cup \ee^{\pm \frac{2\pi \ii}{3}} \R_+)$, with the same holding for their boundary values on $\R \cup \ee^{\pm \frac{2\pi \ii}{3}} \R_+$. Recalling \eqref{eq:zeta-phi}, we estimate the rows of the $N$-dependent factors in \eqref{eq:local-def}:
    \begin{multline}
        \left| \left[N^{\frac16 \sigma_3}  \mb A(N^\frac23 \zeta^+(z)) \ee^{\frac23 (N^\frac23 \zeta^+(z))^{3/2} \sigma_3} \right]_{1k}  \right| \leq C N^{\frac16} \qandq 
          \left| \left[N^{\frac16 \sigma_3}  \mb A(N^\frac23 \zeta^+(z)) \ee^{\frac23 (N^\frac23 \zeta^+(z))^{3/2} \sigma_3} \right]_{2k}  \right| \leq \widetilde{C}
    \end{multline}
    uniformly for $z \in \overline{U_\delta^{(+)}}$. This implies the desired estimate for $\mb P^{(+)}(z)$, and an identical proof applies for $\mb P^{(-)}(z)$. An application of Liouville's Theorem implies $\det \mb P^{(\pm)} = 1$, and so the estimate also applies to $(\mb P^{(\pm)})^{-1}$. 

    Note that the estimate is sharp. Indeed, when $z = z_+$, we have $\phi(z_+) = 0 = \zeta^+(z_+)$ and thus, $\mb P^{(+)}(z_+)$ only depends on $N$ through $\mb E^{(+)}(z) = \Oo(N^{\frac16})$.
    

    \label{remark:growing-parametrix}
\end{remark}

\subsection{Small norm Riemann-Hilbert problem} 
Let 
\[
\Sigma_{\mb R} := \left[(\gamma_{+} \cup \gamma_{-} \cup \gamma) \setminus \left(\gamma_0 \cup U^{(+)}_\delta\cup U^{(-)}_\delta \right) \right] \cup \left( \partial U^{(+)}_\delta\cup \partial U^{(-)}_\delta \right)
\] 
and consider the matrix 
\begin{equation}
    \mb R(z) := \mb S(z) 
    \begin{cases}
       \left( \mb P^{(\pm)}(z)\right)^{-1}, & z \in U^{(\pm)}_\delta, \medskip \\
       \left( \mb N(z) \right)^{-1}, & z \in \C \setminus \left(\Sigma_{\mb R} \cup (U^{(+)}_\delta \cup U^{(-)}_\delta) \right), 
    \end{cases}
    \label{eq:R-matrix-def}
\end{equation}
where we orient $\partial U_\delta^{{(\pm)}}$ clockwise. Then, $\mb R(z)$ solves the following RHP:
\begin{rhp}
    Seek a $2 \times 2$ matrix-valued function $\mb R(z)$ satisfying the following conditions: 
    \begin{enumerate}[(a)]
        \item $\mb R(z)$ is analytic in $ \C \setminus \Sigma_{\mb R}$, 
        \item $\mb{R}(z)$ has continuous boundary values on $\Sigma_{\mb R} \setminus \{-\ee^{c/2}\}$ satisfying
        \begin{equation}
            \mb{R}_+(z) = \mb R_-(z) \begin{cases}
                \mb N(z) \begin{bmatrix}
            1 & 0 \\  \ee^{-2N\phi(z) + \nu(z)} & 1
        \end{bmatrix} \mb N^{-1}(z), & z \in \gamma_{\pm} \setminus U_\delta^{(\pm)} \medskip \\
                \mb N(z) \begin{bmatrix}
                    1 & \ee^{-\nu(z)} \ee^{2N \phi(z)} \\ 0 & 1
                \end{bmatrix} \mb N^{-1}(z), & z \in \gamma \setminus (\gamma_0 \cup  U_\delta^{(\pm)}), \medskip \\
                \mb P^{(\pm)}(z) \mb N^{-1}(z), & z \in \partial U^{(\pm)}_\delta
            \end{cases}
        \end{equation}
        \item As $z \to -\ee^{c/2}$, $\mb R(z) = \Oo(\log|z + \ee^{c/2}|)$. Furthermore, $\lim_{z \to \infty} \mb R(z) = \I.$
    \end{enumerate}
    \label{rhp:R}
\end{rhp}
Since $\mb N(z)$ is independent of $N$, one can check that the jumps of $\mb R(z)$ are of the form $\I + \Oo(N^{-1})$ and by the now standard theory of small-norm Riemann-Hilbert problems (see, e.g. \cite{MR1677884}), we have that a solution $\mb R(z)$ exists and satisfies the estimate
\begin{equation}
    \mb R(z)  = \I + \Oo(N^{-1}) \qasq N \to \infty, 
    \label{eq:R-asymptotic}
\end{equation}
locally uniformly for $z \in \overline{\C} \setminus \Sigma_{\mb R}$. It follows from \eqref{eq:R-matrix-def} and \eqref{eq:S-def} that $\mb T(z)$ is uniformly bounded away from the disks $U_\delta^{(\pm)}$, while $\mb T(z) = \Oo(N^{\frac16})$ for $z \in U_\delta^{(\pm)}$.

\subsection{Asymptotics of \texorpdfstring{$\widehat{P}_n(z; c, N)$}{the approximating orthogonal polynomials}}

For any $z \in \C \setminus \gamma_0$, one can choose the arcs $\gamma_\pm$ in \eqref{eq:S-def} such that 
\begin{equation}
   \ee^{\frac{N\ell}{2}\sigma_3} \widehat{\mb Y}_{\pm}(z) \ee^{-N g(z)\sigma_3}\ee^{-\frac{N\ell}{2}\sigma_3} =\mb T_\pm(z) = \mb S_\pm(z) = \mb R_\pm(z) \mb N(z).
\end{equation}
It follows from \eqref{eq:Y-hat-formula} that
\begin{equation}
    \widehat{P}_N(z;c, N) = [\widehat{\mb Y}]_{11}(z) = \ee^{Ng(z)} \left(\dfrac{1}{2}\dfrac{\varsigma(\infty)}{\varsigma(z)} \left( a(z) + \frac{1}{a(z)}\right) [\mb R]_{11}(z) -\dfrac{1}{2\ii} \dfrac{1}{\varsigma(\infty)\varsigma(z)}\left( a(z) - \frac{1}{a(z)}\right) [\mb R]_{12}(z) \right).
\end{equation}
Using \eqref{eq:R-asymptotic}, we find 
\begin{equation}
    \widehat{P}_N(z;c, N) =  \ee^{Ng(z)} \left(\dfrac{1}{2}\dfrac{\varsigma(\infty)}{\varsigma(z)} \left( a(z) + \frac{1}{a(z)}\right)  +\Oo(N^{-1}) \right) \qasq N \to \infty.
    \label{eq:approximate-asymptotics}
\end{equation}
Similarly, we have 
\begin{equation}
    -\dfrac{2\pi \ii}{\widehat{\kappa}_{N-1}(c, N)} \widehat{P}_{N-1}(z;c, N)  = \ee^{N(g(z)+\ell)} \left( -\dfrac{1}{2\ii} \dfrac{1}{\varsigma(\infty)\varsigma(z)}\left( a(z) - \frac{1}{a(z)}\right) + \Oo(N^{-1}) \right).
    \label{eq:approximate-asymptotics-2}
\end{equation}


Before moving on, we remark that  the existence of a solution $\mb R(z)$ for $N$ large enough implies that assumptions \eqref{eq:rhp-initial-assumptions} are satisfied for $N$ large enough. In particular, $\deg \widehat{P}_N(z; c, N) = N$ for such values of $N$. This, of course, can also be seen from \eqref{eq:approximate-asymptotics} and the logarithmic behavior of $g(z)$ at $z = \infty$. 

\subsection{Approximation of the Riemann-Hilbert problem for \texorpdfstring{$P_n(z)$}{orthogonal polynomials}}
\label{subsec:approx-step}

To arrive at asymptotics of $P_n(z; q, N)$, we use their characterization in terms of a Riemann-Hilbert problem similar to Riemann-Hilbert problem \ref{rhp:initial}. Indeed, recalling the definition \eqref{eq:kappa-n}, the matrix\footnote{Note that it follows from Proposition \ref{prop:jacobi} that the assumptions analogous to \eqref{eq:rhp-initial-assumptions} are automatically satisfied.}
\begin{equation}
  {\mb Y}(z;q, N) = \begin{bmatrix}
        {P}_n(z; q, N) & \mc C[{P}_n(z)  \prod_{j = 1}^{2N}(1 + z^{-1}q^j)](z; q, N) \medskip \\ - \kappa^{-1}_{n- 1}(q, N) {P}_{n - 1}(z; q,N) & - \kappa^{-1}_{n- 1}(q, N) \mc C [{P}_{n - 1}(z)  \prod_{j = 1}^{2N}(1 + z^{-1}q^j)](z; q, N)
    \end{bmatrix},
    \label{eq:Y-formula}
\end{equation}
solves the following Riemann-Hilbert problem:
\begin{rhp}
    Fix $q > 1$ and seek a $2 \times 2$ matrix-valued function ${\mb Y}_n(z;q, N) \equiv {\mb Y}(z;q, N)$ satisfying the following conditions:
    \begin{enumerate}[(a)]
        \item ${\mb Y}(z;q, N)$ is analytic in $\C \setminus \gamma$, 
        \item ${\mb Y}(z;q, N)$ has continuous boundary values on $\gamma$ satisfying 
        \[
          {\mb Y}_+(z;q, N) = {\mb Y}_-(z;q, N) \begin{bmatrix}
                1 & \prod_{j = 1}^{2N}(1 + z^{-1}q^j) \\ 0 & 1
            \end{bmatrix}, 
        \]
        \item As $z \to \infty$, we have 
        \begin{equation}
            {\mb Y}(z;q, N) = \left( \I + \dfrac{{{\mb Y}}^{(1)}}{z} + \mc O \left( z^{-2} \right) \right)z^{n \sigma_3}. 
        \end{equation}
    \end{enumerate}
    \label{rhp:initial-original}
\end{rhp}

Let $N>0$ be large enough so that $\widehat{\mb Y}_N(z; c, N)$ exists and let 
\begin{equation}
    \mb X(z) \equiv \mb X(z; c, q, N) := \ee^{\frac{1}{2}N\ell \sigma_3} \mb Y_N(z) \left(\widehat{\mb Y}_N(z)\right)^{-1} \ee^{-\frac12 N\ell \sigma_3}. 
    \label{eq:X-def}
\end{equation}
Then, $\mb X(z)$ satisfies the following Riemann-Hilbert problem.
\begin{rhp}
    Seek a $2 \times 2$ matrix-valued function $\mb X(z;c, q, N)$ satisfying the following conditions:
    \begin{enumerate}[(a)]
        \item $\mb X(z)$ is analytic in $\C \setminus \gamma$, 
        \item $\mb X(z)$ has continuous boundary values on $\gamma \setminus \{-\ee^{c/2}\}$ satisfying 
        \begin{equation}
            \mb X_{+}(z) = \mb X_{-}(z) \left(\I + \left( {\prod_{j = 1}^{2N}(1 + z^{-1}q^j)} - \ee^{-NV(z) - \nu(z)}\right) \ee^{\frac12 N\ell \sigma_3}\widehat{\mb Y}_{ -}(z) \begin{bmatrix}
                0 & 1 \\ 0 & 0
            \end{bmatrix}   \widehat{\mb Y}^{-1}_{-}(z) \ee^{-\frac12 N\ell \sigma_3} \right).
            \label{eq:X-jump}
        \end{equation}
        \item $\mb X(z) = \Oo(\log|z + \ee^{c/2}|)$ as $z \to -\ee^{c/2}$. Furthermore, as $z \to \infty$, 
        \[
        \mb X(z) = \I + \mc O(z^{-1}).
        \]
    \end{enumerate}
    \label{rhp:ratio}
\end{rhp}
We first observe that the jump matrix for $\mb X(z)$ only involves the first column of $\widehat{\mb Y}_-(z)$ which is analytic and is thus equal to the first column of $\widehat{\mb Y}(z)$. Let $\mb J_{\mb X}(z)$ be the jump matrix in \eqref{eq:X-jump} and define 
\[
    \mb K(z) := \begin{cases}
        \mb N(z), & z \in \gamma\setminus\left(\gamma_0 \cup U_\delta^{(+)} \cup U_\delta^{(-)} \right), \medskip \\
        \mb N_-(z) \begin{bmatrix}
            1 & 0 \\ \ee^{-2N \phi_-(z) + \nu(z)} & 1
        \end{bmatrix}, & z \in \gamma_0  \setminus U_\delta^{(\pm)}, \medskip \\
        \mb P_-^{(\pm)}(z), & z \in (\gamma\setminus \gamma_0)  \cap U_\delta^{(\pm)}, \\
        \mb P_-^{(\pm)}(z) \begin{bmatrix}
            1 & 0 \\ \ee^{-2N \phi_-(z) + \nu(z)} & 1
        \end{bmatrix}, & z \in \gamma_0  \cap U_\delta^{(\pm)}.

    \end{cases}
\]
Using \eqref{eq:T-def}, \eqref{eq:S-def}, \eqref{eq:R-matrix-def} we find 
\begin{equation}
    \mb J_{\mb X}(z) = \I + \left( \ee^{NV(z) + \nu(z)} {\prod_{j = 1}^{2N}(1 + z^{-1}q^j)} - 1\right) \ee^{N(2g_-(z) -V(z) + \ell)}\ee^{- \nu(z)} \mb R_-(z) \mb K(z) \begin{bmatrix}
                0 & 1 \\ 0 & 0
    \end{bmatrix} \mb K^{-1}(z) \mb R^{-1}_-(z).
    \label{eq:X-jump-1}
\end{equation}
Recall that $\mb N(z)$ is independent of $N$ and that $\re(\phi_{\pm}(z)) = 0$ for $z \in \gamma_0$ (cf. Figure \ref{fig:phi-levels}). This, combined with Remark \ref{remark:growing-parametrix} implies that $\mb K(z) = \Oo(N^{\frac16})$  and $\mb K^{-1}(z) = \Oo(N^{\frac16})$ as $N \to \infty$ uniformly for $z \in \gamma$, and $\mb K(z) = \Oo(1)$ for $z \in \gamma\setminus\left(\gamma_0 \cup U_\delta^{(+)} \cup U_\delta^{(-)} \right)$. It follows from this discussion and \eqref{eq:R-asymptotic} that 
\begin{equation}
    \ee^{- \nu(z)} \mb R_-(z) \mb K(z) \begin{bmatrix}
                0 & 1 \\ 0 & 0
            \end{bmatrix} \mb K^{-1}(z) \mb R^{-1}_-(z) = \Oo(N^{\frac13}), \quad z \in \gamma.
\end{equation}
On the one hand, for $z \in \gamma_0$, using \eqref{eq:g-phi-jump-1}, \eqref{eq:g-phi-jump-2}, we have
\[
    \ee^{N(2g_-(z) - V(z) + \ell)} = \ee^{N(g_-(z) - g_+(z))} = \ee^{-2N\phi_+(z)},
\]
which is unimodular for $z \in \gamma_0$. Thus, it follows from Proposition \ref{prop:weight-approx} that 
\begin{equation}
    \mb J_{\mb X}(z) = \I + \Oo(N^{-\frac23}) \qasq N \to \infty, \quad  z \in \gamma_0 \cup \set{z_{+}, z_-}. 
\end{equation}
On the other hand, using \eqref{eq:g-phi-jump-3} and analyticity of $g(z)$ on $\C \setminus \Gamma$, we have
\begin{equation}
    \ee^{N(2g_-(z) - V(z) + \ell)} = \ee^{N(g_+(z) + g_-(z) - V(z) + \ell)}, \quad z \in \gamma\setminus \gamma_0.
    \label{eq:exp-expr}
\end{equation}
Combining this expression with \eqref{eq:g-phi-ineq-off-arc}, Proposition \ref{prop:weight-approx}, and Remark \ref{remark:growing-parametrix}, we find 
\begin{equation}
    \mb J_{\mb X}(z) = \I + \Oo(N^{-\frac23}) \qasq N \to \infty, \quad z \in (\gamma \setminus \gamma_0 ) \cap \overline{U_{\delta}^{(\pm)}}.
\end{equation}
Finally, for $z \in \gamma\setminus\left(\gamma_0 \cup U_\delta^{(+)} \cup U_\delta^{(-)} \right)$, using equality \eqref{eq:exp-expr}, inequality \eqref{eq:g-phi-ineq-off-arc}, boundedness of $\mb N(z)$, and Proposition \ref{prop:weight-approx}, there exists some $k \in \R_+$ such that
\begin{equation}
    \mb J_{\mb X}(z) =\I + \Oo(N \ee^{-kN}) \qasq N \to \infty, \quad z \in \gamma\setminus\left(\gamma_0 \cup U_\delta^{(+)} \cup U_\delta^{(-)} \right).
\end{equation}
Applying standard small norm Riemann-Hilbert problem theory (see, e.g. \cite{MR1677884}), we find that 
\begin{equation}
    \mb X(z) = \I + \Oo(N^{-\frac{2}{3}}) \qasq N \to \infty,
    \label{eq:X-estimate}
\end{equation}
locally uniformly for $z \in \overline{\C} \setminus \gamma$, and the same estimate holds for the boundary values of $\mb X(z)$ uniformly on compact subsets of $\gamma \setminus \{-\ee^{c/2}\}$. From the definition of $\mb X(z)$, it follows that
\begin{equation}
     P_N(z; \ee^{\frac{c}{2N}}, N) = [\mb X(z)]_{11} \widehat{P}_N(z; c, N) - [\mb X(z)]_{12} \ee^{-N\ell} \dfrac{2\pi \ii}{\widehat{\kappa}_{N-1}(c, N)} \widehat{P}_{N-1} (z; c, N).
     \label{eq:polynomial-relation}
\end{equation}
Plugging \eqref{eq:X-estimate}, \eqref{eq:approximate-asymptotics}, and \eqref{eq:approximate-asymptotics-2} into \eqref{eq:polynomial-relation} gives Theorem \ref{thm:polynomial-asymptotics}. 

\subsection{An avatar of \texorpdfstring{$R_N(w, z)$}{the CD kernel}}

To obtain asymptotics of the correlation kernel, we will rewrite formula \eqref{eq:corr-kernel} in terms of the function  
\begin{equation}
    \mc R_N(w, z) := \frac{1}{2\pi \ii}\int_{\gamma} R_N(t, z) \prod_{j = 1}^{2N} \left(1 + \dfrac{q^j}{t} \right) \dfrac{t - z}{t - w} \dd t, \quad z \in \C, \ w \in \C \setminus \gamma. 
    \label{eq:mc-r-def}
\end{equation}
Luckily, the above Riemann-Hilbert analysis yields the following useful result.
\begin{proposition}
    Let $g(z)$ be as in \eqref{eq:g-def} and $\mc R_n(w, z)$ as in \eqref{eq:mc-r-def}. For any $z \in \C \setminus \left(U_\delta^{(+)} \cup U_\delta^{(-)} \right)$, let $w \mapsto \widetilde{\mc R}_n(w, z)$ be the analytic continuation of $\mc R_n(w, z)$ from $\{w \ : \ |w| < \ee^{\frac{c}{2}}\}$ to the domain bounded by $\gamma_0 \cup \gamma_{out}$ (cf. Figure \ref{fig:phi-levels}). Then, for $q = \ee^{\frac{c}{2N}}$, we have that 
    \begin{equation}
        \widetilde{\mc R}_N(w, z) \ee^{N(g(w) - g(z))} = [\mb N^{-1}(w) \mb N(z)]_{11} + \Oo(N^{-\frac23})
        \label{eq:mc-R-estimate}
    \end{equation}
    locally uniformly as $N \to \infty$. 
    \label{prop:mc-R-estimate}
\end{proposition}
\begin{proof}
    Using the Christoffel-Darboux formula, one can directly verify that 
    \begin{equation}
    \mc R_n(w, z) =\begin{bmatrix}
        1 & 0
    \end{bmatrix} \mb Y^{-1}(w; q, N) \mb Y(z; q, N) \begin{bmatrix}
        1 \\ 0
    \end{bmatrix} . 
    \label{eq:mc-r-rh}
\end{equation}
Using the definition of $\mb X(z)$ and estimate \eqref{eq:X-estimate} we have the following estimate for $z, w$ inside $\gamma$:
\begin{align*}
    \mc R_N(w, z) &= \begin{bmatrix}
        1 & 0
    \end{bmatrix} \mb Y^{-1}(w) \mb Y(z) \begin{bmatrix}
        1 \\ 0
    \end{bmatrix}\\
    &= \begin{bmatrix}
        1 & 0
    \end{bmatrix} \left(\ee^{\frac12 N\ell \sigma_3}\widehat{\mb Y}\right)^{-1}(w) (\I + \Oo( N^{-\frac23})) \ee^{\frac12 N\ell \sigma_3}\widehat{\mb Y}(z) \begin{bmatrix}
        1 \\ 0
    \end{bmatrix}\\
    &= \ee^{N(g(z) - g(w))} \begin{bmatrix}
        1 & 0
    \end{bmatrix} \left(\mb T\right)^{-1}(w) (\I + \Oo( N^{-\frac23})) \mb T(z) \begin{bmatrix}
        1 \\ 0
    \end{bmatrix}.
\end{align*}
For $z, w$ as described, we have $\mb T(z) = \mb R(z) \mb N(z)$, and using this, $\det \mb N(z) \equiv 1$, and estimate \eqref{eq:R-asymptotic}, we find \eqref{eq:mc-R-estimate}. To extend to the region bounded by $\gamma_0 \cup \gamma_{out}$, we must analytically continue (in the $w$ variable) $\mc R(w, z)$ across the circle. Denoting this continuation by $\widetilde{\mc R}(w, z)$ and using the known jumps of $\mb T(z)$ and the definition of $\mb X(z)$, for $w$ in the region bounded by $\gamma_{out} \cup \gamma$ we have
\begin{equation}
   \ee^{N(g(w) - g(z))} \widetilde{\mc R}_N(w, z) =  \begin{bmatrix}
    1 & -W_N(w)\,\ee^{-\nu(w)} \ee^{2N\phi(w)}
\end{bmatrix} \left(\mb T\right)^{-1}(w) (\I + \Oo( N^{-\frac23})) \mb T(z) \begin{bmatrix}
        1 \\ 0
    \end{bmatrix}
\end{equation}
where $W_N(w) = \ee^{NV(w) + \nu(w)} \prod_{j = 1}^{2N}(1 + \frac{q^j}{w})$. The Euler--Maclaurin argument in the proof of Proposition \ref{prop:weight-approx} applies verbatim to $w$ in compact subsets of $\C \setminus [-\ee^c, -1]$ and gives $W_N(w) = 1 + \Oo(N^{-1})$ locally uniformly there, so
\[
 \ee^{N(g(w) - g(z))} \widetilde{\mc R}_N(w, z) = [\mb N^{-1}(w) \mb N(z)]_{11} + \Oo(N^{-\frac23})
\]
locally uniformly as $N \to \infty$. 

\end{proof}

\section{Preliminaries to the saddle point analysis}
\label{sec:saddle-point-preliminaries}

In this section, we will record basic facts about the function $\Phi_c(z; \xi, \eta)$, including proofs of Lemma \ref{lemma:saddle-pt}
and Proposition \ref{prop:inflection}. Before embarking on these technical calculations, we record some details of the $c = 0$ degeneration and deduce the classical arctic circle (see Figure \ref{fig:arctic-a}). We hope this is helpful in forming the reader's intuition for the calculations to follow.

\subsection{Uniform tiling: \texorpdfstring{$c = 0$}{}} 
\label{sec:contours-c-0}
It is straightforward to check that
\[
\lim_{c \to 0} \Phi_c(z; \xi, \eta) = g_0(z) + (1 + \xi)\log (1 + z) - (1 + \eta) \log z + \frac{\ell_0}{2},
\]
where $ {\ell_0} = -2g_0(\ee^{\frac23 \pi \ii}) + 2\pi \ii/3$, 
\[
g'_0(z) = \lim_{c \to 0} g'_c(z) =\dfrac{1}{z(z+1)} + \dfrac{R_0(z)}{z(z+1)}, 
\]
and \( R_0(z) = \lim_{c \to 0} R(z)\). Similarly, using Lemma \ref{lemma:g-prime-identity}, 
\[
\lim_{c \to 0} \psi(z) = \lim_{c \to 0} \left(g'(z) - \frac12 V'(z) \right) = \dfrac{R_0(z)}{z(z + 1)}. 
\]
Using this, we find that 
\[
\dod{\Phi_0}{z}(z; \xi, \eta) = \dfrac{R_0(z)}{z(z+1)} + \dfrac{1}{z} (\xi - \eta) + \dfrac{1}{z}\dod{}{c} \left( \log \left( \dfrac{(z+1)(z+\ee^c)}{(z + \ee^{\frac{c}{2}(1 + \xi)})^2} \right)\right)\biggl|_{c = 0} = \dfrac{R_0(z)}{z(z+1)} + \dfrac{1}{z} (\xi - \eta) - \dfrac{\xi}{z(1 + z)}.
\]
Now, supposing $s \equiv s(\xi, \eta) \neq 0, -1$ is a zero of $\Phi_0'(z)$, then it must satisfy 
\[
s(s+1) \Phi'_0(s) = {R_0(s)}+ (s+1) (\xi - \eta) - {\xi} =0 
\]
Collecting terms and squaring both sides yields the equation 
\[
(s - z_+)(s - z_-) = ((s + 1)\eta - s\xi)^2 .
\]
Recalling that $\lim_{c \to 0} z_\pm = \ee^{\pm \frac{2\pi \ii}{3} }$, we find the equation 
\[
s^2 + s + 1 = ((s + 1)\eta - s\xi)^2
\]
which agrees with calculations in\footnote{In the notation of \cite{MR4124992}, the uniform tiling corresponds to $\alpha = 1$.} \cite{MR4124992}. Using the definition of the liquid region \eqref{eq:liquid-def}, $(\xi, \eta)$ belongs to $\partial \mc L_0$ exactly when the two complex roots of this equation coincide; i.e. when the discriminant vanishes. This yields the equation of the (sheared) ellipse
\[
4\xi^2 - 4\eta \xi + 4\eta^2 = 3.
\]
In what follows, our goal is to extend this analysis to $c >0$ and arrive at a description of the arctic curve $\partial \mc L_c$.
\subsection{Saddle points as solution to polynomial equation}
Our first result is a statement about critical points of $\Phi_c(z; \xi, \eta)$, i.e. solutions to the equation 
\begin{equation}
{\dod{\Phi_c}{z}(z)} = \frac{1}{2} V'(z) + \psi(z) + \frac{1}{z} (\xi - \eta) + \dfrac{2}{cz} \log \left( \dfrac{z + 1}{z + \ee^{\frac{c}{2}(1 + \xi)}} \right) = 0.
\label{eq:sad}
\end{equation}
\begin{lemma} 
    Fix $(\xi, \eta) \in \mc{H}$ (see \eqref{eq:hexagon-set}) and let $\Phi_c(z; \xi, \eta)$ be as in \eqref{eq:phase-def}. If $s \in \C \setminus \{0, -1, -\ee^{c}\}$ is a complex number satisfying 
    \[
    \dod{\Phi_c}{z}\biggl|_{z = s} = 0,
    \]
    then $s$ satisfies the polynomial equation
    \begin{equation}
         (1-t^2)\pi_1(s)^2 - 2\big(s+\ee^{\frac{c}{2}}\big)^2\big(\pi_1(s)+\pi_2(s)\big) = 0,
        \label{eq:Pi-eq}
    \end{equation}
    where $t$ is defined in \eqref{eq:t-def} and
\begin{equation}
    \begin{aligned}
        \pi_1(s) &= (s+1)(s+\ee^c)+\ee^{c(\eta-\xi)}\big(s+\ee^{\frac c2(1+\xi)}\big)^2,\\
        \pi_2(s) &= \ee^{c(\eta-\xi)}\big(s+\ee^{\frac c2(1+\xi)}\big)^2-(s+1)(s+\ee^c).
        \label{eq:pi-def}
    \end{aligned}
\end{equation}
\label{lemma:polynomial-eq}
\end{lemma}

\begin{proof}
Since $\xi, \eta \in \R$, exponentiating both sides yields that $s \neq 0$ is a saddle point iff 
\begin{equation}
\exp \left( c s \cdot {\dod{\Phi_c}{z}(s)}\right)  = 1 \implies 
\ee^{cs\psi(s)}(s+1)(s+\ee^c)=\ee^{c(\eta-\xi)}\big(s+\ee^{\frac c2(1+\xi)}\big)^2 .
\label{eq:exp-eta-minus-xi}
\end{equation}
Using \eqref{eq:psi-def}, \eqref{eq:exp-eta-minus-xi} becomes
\begin{equation}
{(s + \ee^{\frac{c}{2}})^2(Q(s) + 1)^2} = (1 - t^2)\ee^{c(\eta-\xi)}\big(s+\ee^{\frac c2(1+\xi)}\big)^2.
\end{equation}
Thus, for any given $s$, we may write 
\begin{equation}
    (s+\ee^{\frac{c}{2}})\,Q(s)=\pm\sqrt{1-t^2}\, \ee^{\frac c2(\eta-\xi)}(s+\ee^{\frac c2(1+\xi)})-(s+\ee^{\frac{c}{2}}).
    \label{eq:intermediate}
\end{equation}
To simplify the resulting expression, recall that, using \eqref{eq:angle}, we have
\begin{equation}
    R^2(s) = s^2 - 2\ee^{\frac{c}{2}}\cos (\theta_c) s + \ee^c \implies (s+ \ee^{\frac{c}{2}})^2 Q^2(s) = (s+\ee^{\frac{c}{2}})^2-(1-t^2)(s+1)(s+\ee^c).
    \label{eq:identity}
\end{equation}
Thus, squaring both sides of \eqref{eq:intermediate} and using \eqref{eq:identity}, we find 
\begin{equation}
   (s+1)(s+\ee^{c})+\ee^{c(\eta-\xi)}\big(s+\ee^{\frac c2(1+\xi)}\big)^2 =\dfrac{\pm2}{\sqrt{1-t^2}}\,\ee^{\frac c2(\eta-\xi)}(s+\ee^{\frac{c}{2}})(s+\ee^{\frac c2(1+\xi)}).
   \label{eq:pi-sqrt}
\end{equation}
Squaring both sides and using the definitions \eqref{eq:pi-def} yields the result.

\end{proof}

\begin{proof}[Proof of Lemma \ref{lemma:saddle-pt}]
    By Lemma \ref{lemma:polynomial-eq} and the fact that \eqref{eq:Pi-eq} is a polynomial equation with real coefficients, it suffices to show that \eqref{eq:Pi-eq} has at least one real root. Denote the right hand side of \eqref{eq:Pi-eq} by $\Pi(s; \xi, \eta)$. Then,
    \[
        \Pi(-\ee^{\frac{c}{2}(1 + \xi)}; \xi, \eta) = (1-t^2) \pi_1^2(-\ee^{\frac{c}{2}(1 + \xi)}) > 0.
    \]
    Furthermore, the leading coefficient of $\Pi(s)$ is given by 
    \[
        (1 - t^2) (1 + \ee^{c(\eta - \xi)})^2 - 4\ee^{c(\eta - \xi)} = 4\ee^{c(\eta - \xi)} \left ( \dfrac{\cosh^2(\frac{c}{2} (\eta - \xi))}{\cosh^2(\frac{c}{2} )} - 1\right) < 0
    \]
    whenever $|\eta - \xi| < 1$ (i.e. $(\xi, \eta)$ is in the interior of $\mc H$) and so $\lim_{s \to \infty} \Pi(s; \xi, \eta) = -\infty$. Thus, there is at least one real root on $(-\ee^{\frac{c}{2}(1+\xi)}, \infty)$, and the result follows.
\end{proof}
\subsection{The liquid region \texorpdfstring{$\mc L_c$}{}} To justify the definition of the liquid region \eqref{eq:liquid-def}, we demonstrate that it is non-empty. When $\xi = \eta = 0$, we have
\[
\Pi(s; \xi, \eta) = -\frac{4(\ee^{c/2}-1)^2}{(\ee^c+1)^2}(s^2 + \ee^{\frac{c}{2}} s + \ee^c )\left(\left(2 \ee^{\frac{c}{2}} + \ee^c + 1\right) s^2 + \left(3 \ee^{\frac{c}{2}} + 2 \ee^c + 3 \ee^{\frac{3c}{2}}\right) s + 2 \ee^{\frac{3c}{2}} + \ee^c + \ee^{2c}\right).
\]
A discriminant calculation shows that the first factor has two complex-conjugate roots and the second has two real roots. Thus, $(0, 0) \in \mc L_c$ for all $c > 0$. By the definition of $\mc L_c$ and the fact that the real roots of $\Pi(s; \xi, \eta)$ are separated by $s = -\ee^{\frac{c}{2}(1 + \xi)}$ we see that the zero-set of the discriminant of $\Pi(s; \xi, \eta)$ in the ($\xi, \eta$)-plane is exactly the boundary of $\mc L$. This, however, is computationally expensive, prohibitively so for large values of $c$. Instead, we now demonstrate how one can arrive at a parametrization of the boundary of the liquid region in terms of the saddle point $s$. 

\subsubsection{Parametrization of \texorpdfstring{$\partial \mc L_c$}{the boundary of the liquid region}}
Recall that $s(\xi, \eta) \in \overline{\C_+}$ denotes the critical point of $\Phi_c(\cdot\,; \xi, \eta)$ from Lemma \ref{lemma:saddle-pt}. On the arctic curve it is real and the critical point has multiplicity $\geq 2$, so
\begin{equation}
    \partial \mc L_c = \left \{ (\xi, \eta) \in \mc H \ : \  \Phi_c''(s(\xi, \eta), \xi, \eta) = 0 \right\}.
    \label{eq:arctic-circle-def}
\end{equation}
Assuming for a moment that $s \neq 0$, this system is then equivalent to 
\begin{equation}
    cs \cdot \Phi_c'(s; \xi, \eta ) = (cs \cdot \Phi_c')'(s; \xi, \eta) = 0.
    \label{eq:sad-1}
\end{equation}
The second of these equations can be rewritten as 
\begin{equation}
    \left(cs\psi(s)\right)' + \dfrac{1}{s + 1} + \dfrac{1}{s + \ee^{c}} - \dfrac{2}{s + \ee^{\frac{c}{2}(1 + \xi)}}= 0. 
    \label{eq:sad-2}
\end{equation}
Observe that equation \eqref{eq:sad-2} depends on $\xi$ but not $\eta$, and can be readily solved for $\ee^{\frac{c}{2}(1+\xi)}$. As an intermediate step, observe that
\begin{equation}
    s + \ee^{\frac{c}{2}(1 + \xi)} = \dfrac{2(s+1)(s+\ee^c)}{(1 + \ee^c + 2s) + (1+s)(\ee^c + s) \left(cs\psi(s)\right)' }. 
    \label{eq:xi-intermediate}
\end{equation}
One could now evaluate the derivative above using \eqref{eq:psi-def} to arrive at an explicit formula. A simpler expression can be found by observing the following identity: 
\begin{equation}
    \left(cs\psi(s)\right)' + \dfrac{1}{s + 1} + \dfrac{1}{s + \ee^{c}} = \dfrac{(1 + \ee^c + 2s) + (1+s)(\ee^c + s) \left(cs\psi(s)\right)' }{(s+1)(s+\ee^c)} = \dod{}{s} \log \left( \ee^{cs\psi(s)} (s+1) (s+\ee^c)\right).
    \label{eq:psi-der-intermediate}
\end{equation}
Using \eqref{eq:psi-def} and the algebraic identity 
\begin{equation}
    (s+\ee^{\frac{c}{2}})^2 - t^2 R^2(s) = (1 - t^2)(s+1)(s+\ee^c),
    \label{eq:R-identity}
\end{equation}
we find 
\begin{equation}
    \ee^{cs\psi(s)} (s+1)(s + \ee^c) = \dfrac{1}{1 - t^2} (s + \ee^{\frac c2} + tR(s))^2 
    \label{eq:exp-psi-id}
\end{equation}
Computing the logarithmic derivative:
\begin{equation}
     \dfrac{(1 + \ee^c + 2s) + (1+s)(\ee^c + s) \left(cs\psi(s)\right)' }{(s+1)(s+\ee^c)} = \dfrac{2(1 + tR'(s))}{s + \ee^{\frac{c}{2}} + tR(s)} = 2\dfrac{R(s) + t(s - \ee^{\frac{c}{2}} \cos \theta_c) }{R(s) (s+ \ee^{\frac{c}{2}} + tR(s))}.
     \label{eq:xi-log-der}
\end{equation}
Plugging this into \eqref{eq:xi-intermediate}, we find 
\begin{equation}
    s + \ee^{\frac{c}{2} (1 + \xi)} =\dfrac{R(s)(s + \ee^{\frac{c}{2}} + tR(s))}{R(s) + t(s - \ee^{\frac{c}{2}} \cos \theta_c)} \implies \ee^{\frac{c}{2} \xi} = \dfrac{R(s) + t(\ee^{\frac{c}{2}} - s \cos \theta_c)}{R(s) + t(s - \ee^{\frac{c}{2}}\cos \theta_c)}.
  \label{eq:exp-xi-formula}
\end{equation}
Plugging the first expression in \eqref{eq:exp-xi-formula} into \eqref{eq:exp-eta-minus-xi}, we find 
\begin{equation}
    \ee^{\frac{c}{2} \eta} = \dfrac{1}{\sqrt{1 - t^2}} \dfrac{R(s) + t(\ee^{\frac{c}{2}} - s\cos \theta_c)}{R(s)} = \cosh\left(\frac{c}{2} \right)\dfrac{R(s) + t(\ee^{\frac{c}{2}} - s\cos \theta_c)}{R(s)}
    \label{eq:exp-eta-formula}
\end{equation}
Equations \eqref{eq:exp-xi-formula}, \eqref{eq:exp-eta-formula} can be interpreted as a parametrization of the arctic curve. These formulas were used to produce the left panel of Figure \ref{fig:arctic-circle-parametrization}. A comparison with sample random tilings for various choices of $c$ is shown in Figure \ref{fig:comparison}. Using \eqref{eq:exp-xi-formula}, \eqref{eq:exp-eta-formula}, one can now verify that $\eta(0) = -1$ and $\xi(\ee^{\frac c2}) = 0$, and so the parameters $s = 0$ and $s = \ee^\frac{c}{2}$ correspond to the endpoints $C, D$ of $\mathfrak{S}$, respectively, as indicated in Figure \ref{fig:arctic-circle-parametrization}.
\begin{figure}[t]
    \begin{subfigure}[b]{0.3\textwidth}
        \centering
        \includegraphics[width=\linewidth]{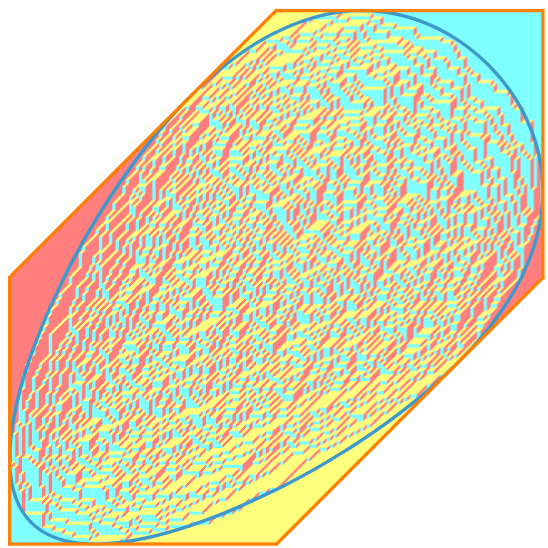}
        \caption{$c = 1$}
        \label{fig:compare-1}
    \end{subfigure}
    \begin{subfigure}[b]{0.3\textwidth}
        \centering
        \includegraphics[width=\linewidth]{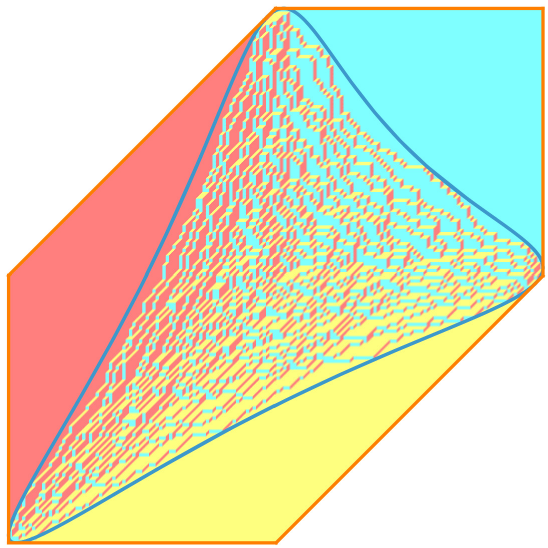}
        \caption{$c = 5$}
        \label{fig:compare-5}
    \end{subfigure}
    \begin{subfigure}[b]{0.3\textwidth}
        \centering
        \includegraphics[width=\linewidth]{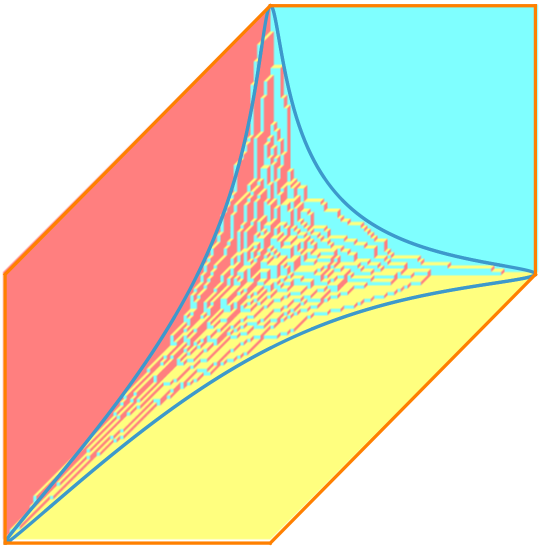}
        \caption{$c = 10$}
        \label{fig:compare-10}
    \end{subfigure}
    \caption{A comparison between the arctic curve generated using \eqref{eq:exp-xi-formula}, \eqref{eq:exp-eta-formula} and a sample tiling for various choices of $c$.}
    \label{fig:comparison}
\end{figure}
\subsection{The appearance of inflection points} 
Recall the characterization of $\partial \mc L_c$ given in \eqref{eq:arctic-circle-def}. Away from turning points (where the arctic curve touches the boundary of the hexagon) and $s = \ee^\frac{c}{2}$ (where $\gamma_0$ intersects the real line), formulas \eqref{eq:exp-xi-formula}, \eqref{eq:exp-eta-formula} produce a parametrization $(\xi(s), \eta(s))$ which is piece-wise analytic. We now identify conditions for inflection points to occur by looking for real values of $s$ such that the quantity $(\xi'\eta'' - \xi''\eta')(s)$ changes signs. To this end, let  \[
\mb F(s) := \begin{bmatrix}
    \Phi'_c(s; \xi(s), \eta(s)) \medskip \\ \Phi''_c(s; \xi(s), \eta(s))
\end{bmatrix}.
\]
Then, $\partial \mc L_c$ is defined by the condition $\mb F(s) = \mb 0$. Recalling the notation \eqref{eq:Phi-derivatives} and the definition of $s(\xi, \eta)$, we have $\varphi_{100} = \varphi_{200} = 0$ and
\begin{align}
    \dod{}{s} \mb F(s) =  \mb F_s(s) + \xi'(s)  \mb F_\xi(s) + \eta'(s) \mb F_\eta(s) = \begin{bmatrix}
        0 \\ \varphi_{300}
    \end{bmatrix} + \xi'(s) \begin{bmatrix}
        \varphi_{110} \\ \varphi_{210}
    \end{bmatrix} + \eta'(s) \begin{bmatrix}
        \varphi_{101} \\ \varphi_{201}
    \end{bmatrix}.
\end{align}
Rewriting this, we have 
\begin{equation}
    \dod{}{s} \mb F(s) = \mb 0 \implies \begin{bmatrix}
        \varphi_{110} & \varphi_{101} \\ \varphi_{210} & \varphi_{201} \end{bmatrix} \begin{bmatrix} \xi'(s) \\ \eta'(s) \end{bmatrix} = \begin{bmatrix} 0 \\ - \varphi_{300}\end{bmatrix} 
        \label{eq:tangent-lin-sys}
\end{equation}
\begin{lemma}
    Let $s \in \R \setminus \{ -\ee^c, -1, 0\}$, then $\varphi_{110} \varphi_{201} - \varphi_{101} \varphi_{210} > 0.$
    \label{lemma:det}
\end{lemma}
\begin{proof}
From the definition of $\Phi_c(z; \xi, \eta)$, we compute 
\begin{align}
    \varphi_{101} &= -\dfrac{1}{s}, & \varphi_{110} &= \dfrac{1}{s + \ee^{\frac{c}{2}(1 + \xi(s))}}, & \varphi_{201} &= \dfrac{1}{s^2}, & \varphi_{210} &= -\dfrac{1}{(s + \ee^{\frac{c}{2}(1 + \xi(s))})^2}.
    \label{eq:phi-derivatives-expressions}
\end{align}
It is elementary to verify using \eqref{eq:exp-xi-formula} that 
\begin{equation}
    s + \ee^{\frac{c}{2}(1 + \xi(s))} = 0 \Leftrightarrow s \in \{ -1, -\ee^{c}, z_+, z_- \},
    \label{eq:exp-xi-equal-s}
\end{equation}

and thus 
\[
\varphi_{110} \varphi_{201} - \varphi_{101} \varphi_{210} = \dfrac{\ee^{\frac{c}{2}(1 + \xi(s))}}{s^2 (s + \ee^{\frac{c}{2}(1 + \xi(s))})^2} >0.
\]
\end{proof}
Hence, it follows from \eqref{eq:tangent-lin-sys} and Lemma \ref{lemma:det} that 
\begin{equation}
    \begin{bmatrix} \xi'(s) \\ \eta'(s) \end{bmatrix} = \dfrac{\varphi_{300}}{\varphi_{110} \varphi_{201} - \varphi_{101} \varphi_{210}} \begin{bmatrix}
        \varphi_{101} \\ - \varphi_{110}
    \end{bmatrix}.
    \label{eq:xi-eta-first-der}
\end{equation}
Continuing in the same way and noting that $\mb F_{\xi \eta} \equiv 0 \equiv \mb F_{\eta \eta}$, we have 
\begin{equation}
    \dod[2]{}{s} \mb F(s) = \mb F_{ss} + 2\xi'(s) \mb F_{s\xi} +  2\eta'(s) \mb F_{s \eta} + (\xi'(s))^2 \mb F_{\xi \xi} + \xi''(s) \mb F_\xi + \eta''(s) \mb F_\eta = \mb 0.
\end{equation}
Written as a linear system, this yields 
\begin{equation}
    \begin{bmatrix}
        \varphi_{110} & \varphi_{101} \\ \varphi_{210} & \varphi_{201} \end{bmatrix} \begin{bmatrix} \xi''(s) \\ \eta''(s) \end{bmatrix} + 2 \begin{bmatrix} \varphi_{210} & \varphi_{201} \\ \varphi_{310} & \varphi_{301} \end{bmatrix} \begin{bmatrix} \xi'(s) \\ \eta'(s) \end{bmatrix} + (\xi'(s))^2 \begin{bmatrix} \varphi_{120} \\ \varphi_{220}\end{bmatrix}
             + \begin{bmatrix}
                 \varphi_{300} \\ \varphi_{400}
             \end{bmatrix} = \mb 0.
\end{equation}
By Lemma \ref{lemma:det}, this system has a unique solution which can be computed via a tedious calculation the result of which we omit and instead record the formula 
\begin{equation}
    (\xi'\eta'' - \xi''\eta')(s) = \dfrac{\varphi_{300}^2}{(\varphi_{110} \varphi_{201} - \varphi_{101} \varphi_{210})^3} \cdot \left( (\varphi_{110} \varphi_{201} - \varphi_{101} \varphi_{210})^2 - \varphi_{101}^2 \varphi_{120} \varphi_{300} \right).
    \label{eq:curvature}
\end{equation}
In the sequel, we will conduct a steepest descent analysis of the correlation kernel \eqref{eq:corr-kernel} at the inflection point on the segment $\mathfrak{S}_c$ (cf. Figure \ref{fig:arctic-circle-parametrization}); this point occurs for an $s \in (0, \ee^{\frac{c}{2}})$. 

\begin{lemma}
    For
    $s \in (0, \ee^{\frac{c}{2}})$, $\varphi_{300} < 0$.
    \label{lemma:phi300-negative}
\end{lemma}
\begin{proof}
    Starting from the (1,1) entry of \eqref{eq:xi-eta-first-der} and using Lemma \ref{lemma:det}, it suffices to show that $\xi'(s) > 0$, which we do using the second expression in \eqref{eq:xi-log-der}. Indeed, denoting 
    \begin{equation}
        W(s) := s + \ee^{\frac{c}{2}} + tR(s),
        \label{eq:W-def}
    \end{equation}
    equation \eqref{eq:xi-intermediate} implies  
    \begin{equation*}\label{eq:W-intermediate}
        s + \ee^{\frac{c}{2}(1 + \xi(s))} = \dfrac{W(s)}{W'(s)} \implies \dod{}{s} \ee^{\frac{c}{2}(1 + \xi(s))} = -\dfrac{W(s)W''(s)}{W'(s)^2}.
    \end{equation*}
    Using \eqref{eq:R-identity}, one can directly verify that $W(s) >0$ for $s \in (0, \ee^{\frac c2})$. For the first derivative,
    \[
        W'(s) = 1 + tR'(s) = 1 + t \dfrac{s - \ee^{\frac c2} \cos \theta_c }{R(s)}, 
    \]
   we note that 
    \begin{equation}
        R^2(s) = (s - \ee^{\frac c2} \cos \theta_c)^2 + \ee^{c} \sin^2 \theta_c \implies  \left| \dfrac{s - \ee^{\frac c2} \cos \theta_c }{R(s)}\right| < 1    \label{eq:R-id-2}
    \end{equation}
    and since $t\in (0,1)$, we find $W'(s) > 0$. In fact, \eqref{eq:R-id-2} combined with $R(s)<0$ on $(0, \ee^{\frac c2})$ implies $R'(s) \in (-1, 0)$. Since $(R'(s))^2 + R(s) R''(s) = 1$, we find $R''(s) < 0$ and so 
    \[
        W''(s) = tR''(s) < 0.
    \]
    Thus,
    \[
        \dod{}{s} \ee^{\frac{c}{2}(1 + \xi(s))}  = \frac{c}{2}\xi'(s) \ee^{\frac{c}{2}(1 + \xi(s))} > 0 \implies \xi'(s) > 0.
    \]

\end{proof}
With Lemmas \ref{lemma:det} and \ref{lemma:phi300-negative}, we now see that the condition for the occurrence of an inflection point is the change in sign of the second factor in \eqref{eq:curvature}. Using \eqref{eq:phi-derivatives-expressions} and 
\[
\varphi_{120} = -\dfrac{c}{2} \dfrac{\ee^{\frac{c}{2}(1 + \xi(s))}}{(s + \ee^{\frac{c}{2}(1 + \xi(s))})^2},
\]
it follows that an inflection point exists at $(\xi(s), \eta(s)), \ s \in (0, \ee^{\frac{c}{2}})$ iff the following expression changes sign:
\begin{equation}
    cs^2\varphi_{300}(s) + \dfrac{2\ee^{\frac{c}{2}(1 + \xi(s))}}{(s + \ee^{\frac{c}{2}(1 + \xi(s))})^2} .
    \label{eq:inflection-eq}
\end{equation}
It is clear that when $c = 0$, \eqref{eq:inflection-eq} is strictly positive. Thus, let $c_* \in \R$ be the first time \eqref{eq:inflection-eq} has a zero. To gain information on $c_*$ requires a more explicit expression for \eqref{eq:inflection-eq}; starting from \eqref{eq:sad-2}, we take a derivative and find that for any $(\xi, \eta) \in \partial \mc L$ 
\begin{equation}
    cs\varphi_{300}(s) = (cs\psi(s))'' - \dfrac{1}{(s+\ee^c)^2} - \dfrac{1}{(s + 1)^2} + \dfrac{2}{(s + \ee^{\frac{c}{2}(1+\xi)})^2}. 
    \label{eq:3rd-derivative-1}
\end{equation}
Specializing $\xi = \xi(s)$, using the notation \eqref{eq:W-def}, differentiating \eqref{eq:psi-der-intermediate}, and recalling \eqref{eq:exp-psi-id} - \eqref{eq:exp-xi-formula}, we find 
\[
     cs\varphi_{300}(s) = \dfrac{2W''(s)}{W(s)}. 
\]
Combined with the first equation in \eqref{eq:exp-xi-formula}, we find 
\begin{multline}
    cs^2\varphi_{300}(s) + \dfrac{2\ee^{\frac{c}{2}(1 + \xi(s))}}{(s + \ee^{\frac{c}{2}(1 + \xi(s))})^2} = \dfrac{2(sW(s)W''(s) + W(s) W'(s) - s(W'(s))^2)}{W^2(s)} \\
    = 2 \dod{}{s} \left(s\dfrac{W'(s)}{W(s)} \right) = 2\dod{}{s}\left(\dfrac{s}{s+\ee^{\frac{c}{2}(1 + \xi(s))}} \right).
    \label{eq:mc-I}
\end{multline}
By the vertical symmetry of the arctic curve, one can deduce that the first inflection point occurs at $s = \ee^{\frac{c_*}{2}}$, i.e. $\xi = 0$. A more detailed understanding of the number and location of inflection points amounts to studying the solutions of \eqref{eq:mc-I}. Since $W(s) > 0$ for $s\in (0, \ee^{\frac c2}]$, it suffices to study the numerator, which is done in the following lemma. 

\begin{lemma}
    Recall \eqref{eq:W-def} and let $$\mc I_c(s) := W^2(s) \dod{}{s} \left(s\dfrac{W'(s)}{W(s)} \right) =  sW(s)W''(s) + W(s) W'(s) - s(W'(s))^2$$ and $c_*$ be the solution of
    \(
        \cosh \left(\frac {c_*}{2} \right) = 1 + \sqrt{3}.
    \) Then, 
    \begin{enumerate}[(a)]
        \item for $c < c_*$, $\mc I_c(s) > 0$ for $s \in (0, \ee^{\frac{c}{2}}]$,
        \item for $c = c_*$, $\mc I_{c_*}(\ee^{\frac{c_*}{2}}) = 0$ and $\mc I_c(s) > 0$ for $s \in (0, \ee^{\frac{c}{2}})$,
        \item for $c > c_*$, $\mc I_{c}(s)$ has exactly one simple zero on $(0, \ee^{\frac{c}{2}})$ .
    \end{enumerate}
    \label{lemma:inflection}
\end{lemma}
\begin{proof}
    First, we collect some facts and then proceed to prove the claims above. Using the explicit expression for $W(s)$ and the identities
    \[
        R^2(s) = s^2 - 2\ee^{\frac c2}\cos (\theta_c ) s + \ee^c, \quad R(s)R'(s) = {s - \ee^{\frac c2} \cos\theta_c}, \qquad (R'(s))^2 + R(s) R''(s) = 1 
    \]
    we find 
    \begin{equation}
       R^3(s) \mc I_c(s) = A_1(s) R(s) + A_0(s)
       \label{eq:I-As}
    \end{equation}
    where $A_0(s), A_1(s)$ are polynomials given by 
    \begin{align*}
        A_1(s) &= \ee^{\frac c2} R^2(s) + t^2 \ee^{\frac c2}\left( 2\ee^\frac{c}{2} s - \cos (\theta_c) (s^2 + \ee^{c})\right) > 0, \\
        A_0(s) &= \ee^{\frac c2} t(1 - \cos \theta_c) \left(s + \ee^{\frac c2}\right)  \left(s^2 + \ee^{\frac c2}(1 - \cos \theta_c) s + \ee^c\right) > 0
    \end{align*}
    Since $R(s) < 0$ on $(0, \ee^{\frac c2})$, \eqref{eq:I-As} implies
    \begin{equation}
        \text{sign}(\mc I_c(s)) = \text{sign}(A_1(s)  |R(s)| - A_0(s)) = \text{sign}(A_1^2(s)  R(s)^2 - A_0^2(s))
        \label{eq:I-sign}
    \end{equation}
    where the last equality follows from $R(s) \in \R$ and $A_1(s) |R(s)| + A_0(s) > 0$. Similar to \eqref{eq:psi-symmetry}, we have the identities
    \[
        R^2 \left( \dfrac{\ee^c}{s} \right) = \dfrac{\ee^c}{s^2} R^2(s), \quad A_1\left( \dfrac{\ee^c}{s} \right) =  \dfrac{\ee^c}{s^2} A_1(s), \quad A_0\left( \dfrac{\ee^c}{s} \right) =  \dfrac{\ee^{\frac 32 c}}{s^3}A_0(s).
    \]
    Thus, denoting $P(s):=A_1^2(s)  R(s)^2 - A_0^2(s)$ we find the symmetry
    \[
        \frac{\ee^{\frac{3}{2}c}}{s^3} P(s) = \frac{s^3}{\ee^{\frac{3}{2}c}} P\left( \frac{\ee^{c}}{s}\right).
    \]
    Stated differently, the expression $s^{-3} P(s)$ is invariant under the map $s \mapsto \ee^c/s$ and can thus be written as a polynomial in $$v(s) := \frac{s}{\ee^{\frac{c}{2}}} + \frac{\ee^{\frac{c}{2}}}{s}.$$ Recalling the definition of $t$ and denoting $b = \cosh(c/2)$, we have
    \begin{multline}
        P(s) = \frac{\ee^{\frac{5 c}{2}} }{b^4 (b+1)^3}s^3 \mc K(v(s)) \quad \text{ where }\\
        \mc K(v) :=  2 b \left(-4 b^4-8 b^3+3 b^2+9 b+4\right)+(-4 b^5-8 b^4+3 b^3+17 b^2+12 b+4)v \\
        +4 b \left(b^3+2 b^2+2 b+1\right) v^2 +2 b^2 (b+1)^2 v^3
        \label{eq:P-K}
    \end{multline}
    Note that $s \in (0, \ee^{\frac{c}{2}}] \iff v \in [2, \infty)$, and so we record the following: 
    \begin{align}
        \mc K (2) &= -4(b + 1)(2b+1)^2(b^2 - 2b - 2), \label{eq:K-2}\\
        \mc K'(2) &= -(2b+1)^2 \left(b^3-9 b^2-12 b-4\right), \label{eq:K-prime-2}\\
        \mc K''(v) &= 4 b (b+1) \left((b^2 + b) (3 v+2) + 2\right). \label{eq:K-pprime-2}
    \end{align}
    In particular, for all $v \geq 2$, $\mc K(v)$ is convex. We are ready to prove the claims above.
    \begin{enumerate}[(a)]
        \item By definition, $c < c_*$ implies $b < 1 + \sqrt{3}$. It follows from \eqref{eq:K-prime-2} that 
        \begin{equation}
            b < 9 \implies \mc K'(2)  > 0. \label{eq:prime-pos}
        \end{equation}
        When $b < 1 + \sqrt{3}$, $\mc K(v)$ is convex and increasing on $[2, \infty)$ and $\mc K(2) > 0$. Thus, by \eqref{eq:I-sign}, $\mc I_c(s) > 0$ for $s\in (0, \ee^{\frac c2}]$.
        
        \item Suppose $b = 1 + \sqrt{3}$. Then, $\mc K(2) = 0$ and, by \eqref{eq:prime-pos} and convexity, we have that $\mc K(v) > 0$ for all $v > 2$ and Lemma \ref{lemma:inflection}(b) follows. 

        \item Suppose $b > 1 + \sqrt{3}$. Then, $\mc K(2) < 0$ and $\lim_{v\to +\infty} \mc K(v) = +\infty$. Thus, there exists $v_* > 2$ such that $\mc K(v_*) = 0$. Uniqueness of the solution follows from convexity. 
    \end{enumerate}
\end{proof}

\begin{remark}
    All other inflection points can be deduced using similar methods and symmetry. Indeed, there is another zero of \eqref{eq:inflection-eq} in the interval $(-\infty, -1)$ (cf. Figure \ref{fig:arctic-circle-parametrization}). The remaining inflection points can be found by rotation; the $\frac{2\pi}{3}$ rotational symmetry is more easily seen in the symmetric hexagon shown in, e.g., Figure \ref{fig:sample-arctic-circles}. The transition around $c \approx c_*$ is shown in Figure \ref{fig:inflection-transition}.
\end{remark}
\begin{figure}[t]
    \begin{subfigure}[b]{0.3\textwidth}
        \centering
        \includegraphics[width=\linewidth]{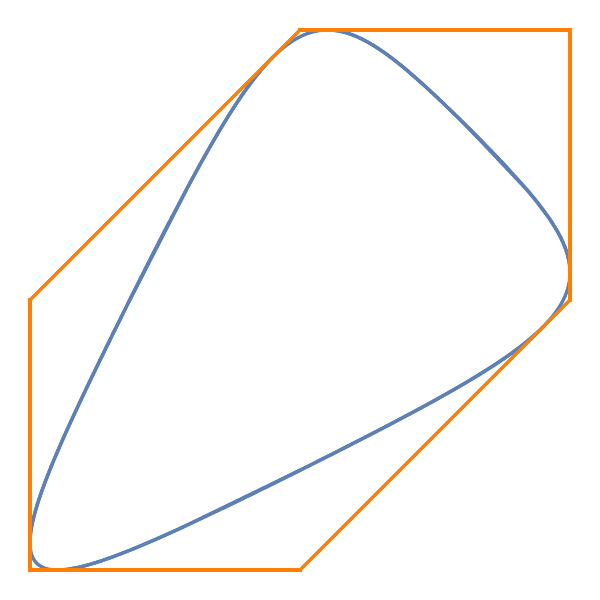}
        \caption{$c = 3$}
        \label{fig:inflection-3}
    \end{subfigure}
    \begin{subfigure}[b]{0.3\textwidth}
        \centering
        \includegraphics[width=\linewidth]{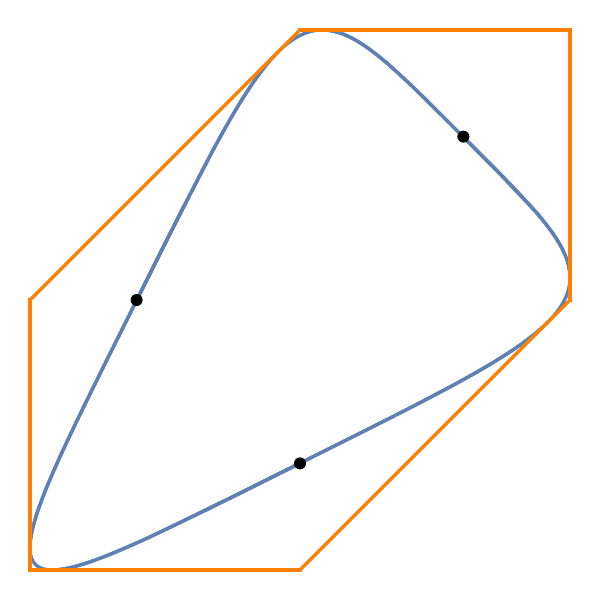}
        \caption{$c = c_* \approx 3.32577...$}
        \label{fig:inflection-critical}
    \end{subfigure}
    \begin{subfigure}[b]{0.3\textwidth}
        \centering
        \includegraphics[width=1\linewidth]{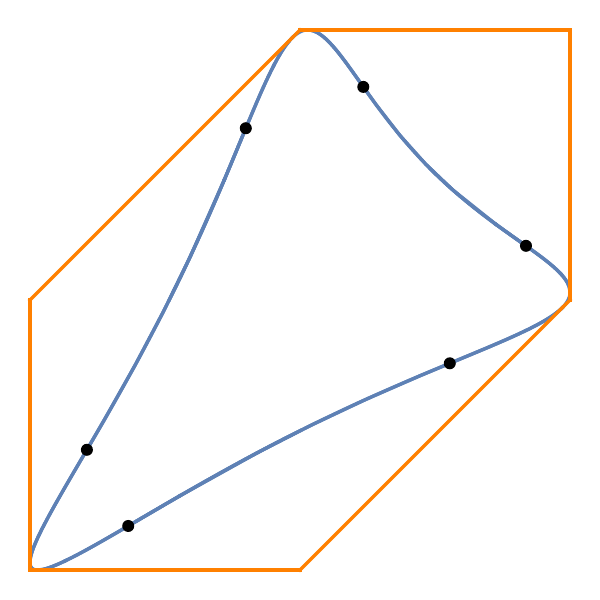}
        \caption{$c = 5$}
        \label{fig:inflection-5}
    \end{subfigure}
    \caption{The appearance of inflection points on the arctic curve}
    \label{fig:inflection-transition}
\end{figure}

\subsection{Level sets of \texorpdfstring{$\re\left(\Phi_c(z) - \Phi_c(s) \right)$}{real part of the phase function}} \label{subsec:level-sets}
In the subsequent steepest descent analysis, it will be essential to understand the geometry of the set 
\[
\mc N_c  := \{z \ : \ \re\left(\Phi_c(z) - \Phi_c(s) \right) = 0\}.
\]
Recall that when $(\xi, \eta) \in \mathfrak S$, we have $s(\xi, \eta) \in (0, \ee^{\frac{c}{2}})$. Since $\Phi_c(z)$ is analytic off of $(-\infty, 0) \cup \{ \ee^{\frac{c}{2}} \ee^{\ii \theta} \ : \ \theta \in [-\pi, \theta_c] \} $, its real part is harmonic outside this set. Furthermore, it is straightforward to verify that the jumps of $\Phi_c(z)$ across $ (-\infty, 0) \cup \{ \ee^{\frac{c}{2}} \ee^{\ii \theta} \ : \ \theta \in [-\pi, \theta_c] \}$ are purely imaginary, and thus $\re(\Phi_c(z))$ is continuous in $\C \setminus \{0, -1, -\ee^c, -\ee^{\frac{c}{2}} \}$. It follows from the definition of $s(\xi, \eta)$ and Lemma \ref{lemma:phi300-negative} that $s(\xi, \eta)$ is a zero of order exactly 2 of $\Phi'_c(z; \xi, \eta)$. 
\begin{figure}
    \begin{subfigure}{0.49\textwidth}
        \centering
    \includegraphics[width=\linewidth]{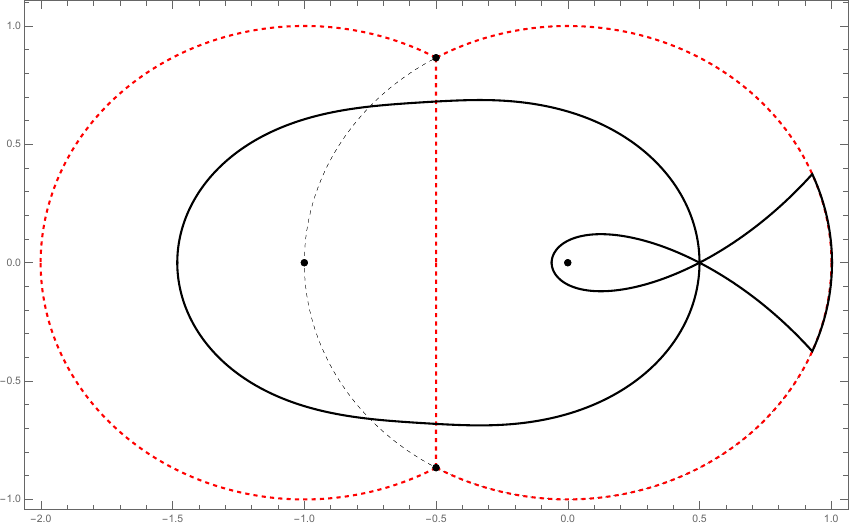}
    \put(-225,135){$\gamma_{out}$}
    \put(-217,100){$\Gamma_1(0)$}
    \put(-95,92){$\Gamma_2(0)$}
    \put(-135,92){$\gamma_{in}$}
    \put(-45,104){$\Gamma_3(0)$}
    \put(-85,20){$+$}
    \put(-85,40){$-$}
    \put(-75,75){$+$}
    \put(-25,75){$-$}
    \caption{$c = 0$}
    \end{subfigure}
    \begin{subfigure}{0.49\textwidth}
        \centering
    \includegraphics[width=\linewidth]{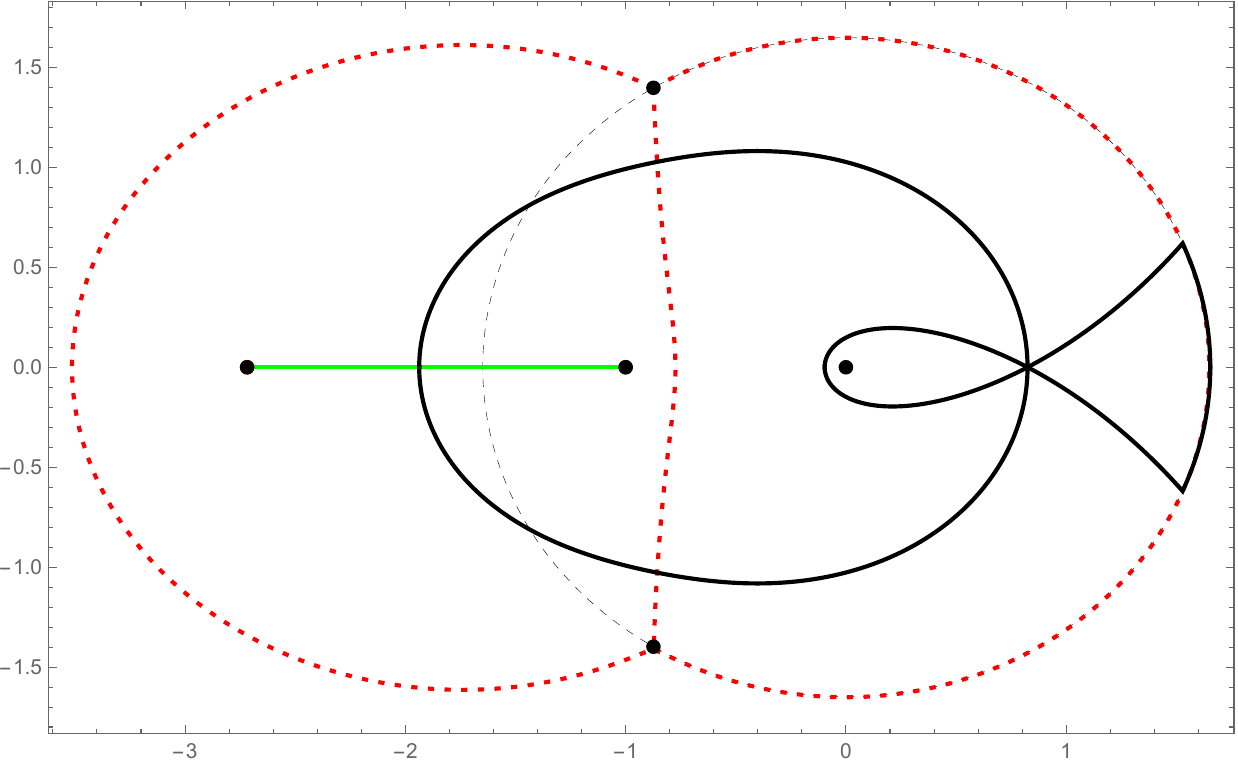}
    \put(-225,135){$\gamma_{out}$}
    \put(-130,95){$\gamma_{in}$}
    \put(-190,100){$\Gamma_1(1)$}
    \put(-95,92){$\Gamma_2(1)$}
    \put(-47,107){$\Gamma_3(1)$}
    \put(-95,25){$+$}
    \put(-95,45){$-$}
    \put(-70,78){$+$}
    \put(-25,78){$-$}
    \caption{$c = 1$}
    \end{subfigure}
        \begin{subfigure}{0.49\textwidth}
        \centering
    \includegraphics[width=\linewidth]{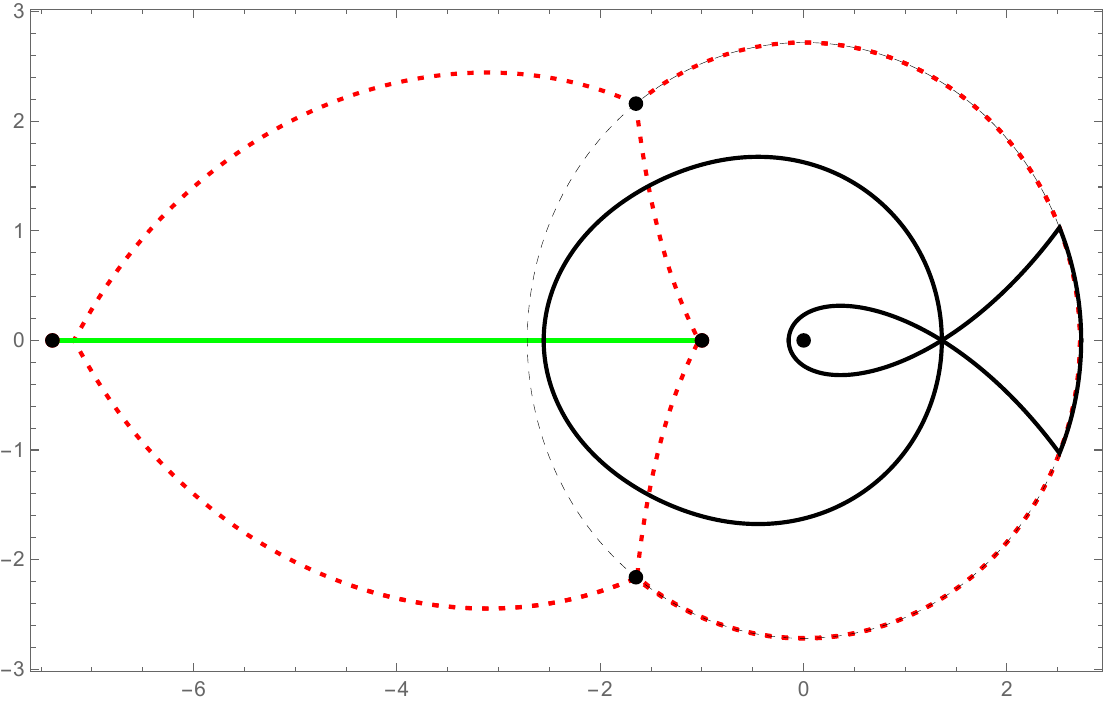}
    \put(-85,30){$+$}
    \put(-85,50){$-$}
    \put(-60,79){$+$}
    \put(-25,79){$-$}
    \caption{$c = 2$}
    \end{subfigure}
        \begin{subfigure}{0.49\textwidth}
        \centering
    \includegraphics[width=\linewidth]{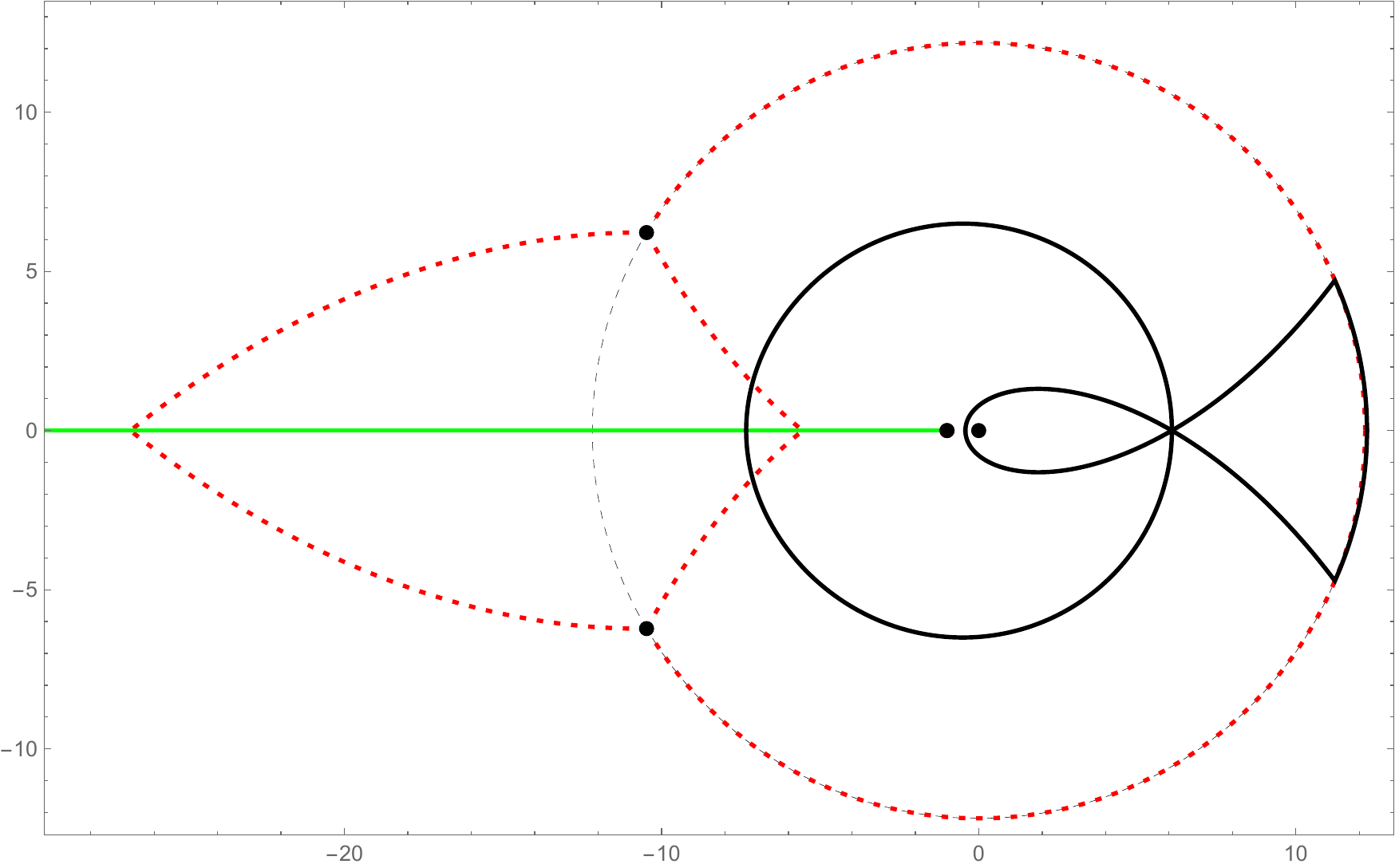}
    \put(-85,30){$+$}
    \put(-85,50){$-$}
    \put(-63,75){$+$}
    \put(-25,75){$-$}
    \caption{$c = 5$}
    \end{subfigure}
    \caption{\centering The set $\mc N_c$ (solid) and the level set $\re(\phi_c(z)) = 0$ (red, dashed) for various choices of $c$ and $s = \frac12 \ee^{\frac{c}{2}}$. The circle of radius $\ee^{\frac c2}$ is indicated with a thin dashed line and the interval $[-\ee^c, -1]$ is shown in green. The signs indicate the sign of $\re(\Phi_c(z) - \Phi_c(s))$.}
    \label{fig:Phi-c-level-set}
\end{figure}
Figure \ref{fig:Phi-c-level-set} shows the set $\mc N_0$ and the set $\{z \ : \ \re(\phi_0(z)) = 0\}$. An important feature of $\mc N_0$ is that the arcs $\Gamma_1(0), \Gamma_2(0)$ remain in the set bounded by $\gamma_{0} \cup \gamma_{out}$ (recall the notation from Figure \ref{fig:phi-levels} and that $\gamma_0 = \text{supp}(\mu)$), which is demonstrated in Figure \ref{fig:Phi-c-level-set} and proven in \cite{MR4124992}. By continuity, it follows that this remains the case for $c>0$ small enough. In the next section, we extend this to all $c>0$. 

\subsubsection{Deformation in \texorpdfstring{$c$}{c}} 
\label{sec:contours-c-general}
In this section, we prove the relevant structure of $\mc N_c$ displayed in Figure \ref{fig:Phi-c-level-set}. Precisely, we seek  to prove the following Lemma. 
\begin{lemma}
    Let $s \in (0, \ee^{\frac{c}{2}})$. Then, there exist contours $\gamma_z, \gamma_w \subset \C\setminus U_\delta^{(\pm)}$ passing through $z = s$ such that 
    \begin{itemize}
        \item $\gamma_z$ remains in the domain bounded by $\ee^{\frac{c}{2}} \T$ and satisfies 
        \[
        \re \left(\Phi_c(z)  - \Phi_c(s)\right) < 0, \quad z \in \gamma_z \setminus \{s\}.
        \]
        \item $\gamma_w$ remains in the domain bounded by $\gamma_0 \cup \gamma_{out}$ and satisfies 
        \[
        \re \left(\Phi_c(w)  - \Phi_c(s)\right) > 0, \quad w \in \gamma_w \setminus \{s\}.
        \]
    \end{itemize}
    \label{lemma:contours}
\end{lemma}
It follows from Lemma \ref{lemma:phi300-negative} that exactly three trajectories of $\mc N_c$ emanate from $z = s$. It follows from 
\begin{equation*}
    \begin{aligned}
         \Phi_c(z) &= (1 + \xi - \eta) \log z + \Oo(1) \qasq z \to \infty, \\
    \Phi_c(z) &= -(1 + \eta) \log z + \Oo (1) \qasq z \to 0,
    \end{aligned}
\end{equation*}
that these trajectories remain bounded, and thus $\mc N_c$ intersects the real line at three points; consideration of the sign of $\re(\Phi_c(z) - \Phi_c(s))$ near $z = 0$ and harmonicity implies that two of the intersection points are to the left of the origin and one is to the right of the origin. Following the strategy of \cite{MR4124992}, we will show that 
\begin{equation}
    z \mapsto \re \left(\Phi_c(z)\right), \quad z \in (\gamma_{out} \cup \gamma_0) \cap \C_+
\end{equation}
is a decreasing function. It would then follow that $\mc N_c$ intersects $\{z \ : \ \re(\phi_c(z)) = 0\}$ exactly once in $\C_+$, and that this intersection point must belong to $\Gamma_3(c) \cap \gamma_0$. Indeed, if this were not the case, then the region bounded by $\Gamma_3(c)$ would be entirely contained in the region bounded by $\gamma_0 \cup \gamma_{out}$ and $\re (\Phi_c(z) - \Phi_c(s))$ would be a non-constant, harmonic function in the region bounded by $\Gamma_3(c)$ and identically zero on $\Gamma_3(c)$, contradicting the maximum modulus principle. To prove the desired monotonicity of $\re \left(\Phi_c(z)\right)$, we first need the following lemma. 
\begin{lemma}
    Fix $\xi \in (-1, 0)$ and let $\psi(z)$ be as in \eqref{eq:psi-def}. Define  
    \begin{equation}
        L(z) := \log(z + \ee^c) + \log(z + 1) - 2 \log \left(z + \ee^{\frac{c}{2}(1 + \xi)} \right),
        \label{eq:L-def}
    \end{equation}
    where all logarithms are principal. Then, $L(z)/(cz\psi(z))$ is analytic in $\C_+ \setminus \ee^{\frac{c}{2}}\overline{\D}$ and satisfies 
    \[
        \im\left(\dfrac{L(z)}{cz\psi(z)} \right) < 0, \qquad z \in \C_+ \setminus \ee^{\frac{c}{2}}\overline{\D}.
    \]
    \label{lemma:sign}
\end{lemma}
\begin{proof}
    By definition, $L(z)$ and $cz\psi(z)$ are analytic on $z \in \C_+ \setminus \ee^{\frac{c}{2}}\overline{\D}$. Furthermore, it follows from Lemma \ref{lemma:pre-image-psi}(c) and the Open Mapping Theorem that for any $z \in \C_+ \setminus \ee^{\frac{c}{2}}\overline{\D}$, $cz\psi(z) \neq 0$, and thus the ratio $L(z)/(cz\psi(z))$ is analytic in $\C_+ \setminus \ee^{\frac{c}{2}}\overline{\D}$, and its imaginary part is harmonic there. For brevity, denote 
    \[
        u(z) := \im\left(\dfrac{L(z)}{cz\psi(z)} \right).
    \]

    To complete the proof, we will apply the maximum modulus to $u(z)$. To do so, we need to first study the singularity of $u(z)$ at $z = z_+$. Observe that for any $z\in \gamma \cap \C_+$
    \[
        \arg (z + \ee^c) = \arg (z) + \arg \left( 1 + \dfrac{\ee^c}{z}\right) = \arg(z) - \arg \left( 1 + z\right) \implies  \arg (z + \ee^c) +  \arg \left( 1 + z\right) = \arg(z).
    \]
    By the Inscribed Angle Theorem, $\arg(z) = 2\arg(z + \ee^{\frac{c}{2}})$, and since $\ee^{\frac{c}{2}(1 + \xi)} < \ee^\frac{c}{2}$, we have 
    \begin{equation}
        \im (L(z)) = \arg (z + \ee^c) +  \arg \left( 1 + z\right) - 2 \arg\left(z + \ee^{\frac{c}{2}(1 + \xi)}\right) = 2 \left(\arg\left(z + \ee^{\frac{c}{2}}\right)  - \arg\left(z + \ee^{\frac{c}{2}(1 + \xi)}\right)\right) < 0.
        \label{eq:L-im-negative}
    \end{equation}
    Furthermore, since $|z + \ee^c| = |z| |1 + \overline{z}| = \ee^{\frac{c}{2}}|1 + z|$, and $1 < \ee^{\frac{c}{2}(1 + \xi)}$, we have
    \begin{equation}
        \re(L(z)) = \log  \left(\dfrac{|z +1||z +\ee^c|}{\left|z +\ee^{\frac{c}{2}(1 + \xi)}\right|^2} \right) = \frac{c}{2} + 2\log  \left(\dfrac{|z +1|}{\left|z +\ee^{\frac{c}{2}(1 + \xi)}\right|} \right). 
        \label{eq:L-real}
    \end{equation}
    From the final expression, we see that 
    \[
        \re(L(z)) > 0 \Leftrightarrow \left|z +\ee^{\frac{c}{2}(1 + \xi)}\right| < \ee^{\frac c4}|z +1|,
    \]
    and the latter can be directly verified. Thus, by the Open Mapping Theorem, $L(z)$ maps a neighborhood of $z_+$ to the interior of the fourth quadrant. Combining this with Lemma \ref{lemma:pre-image-psi}(c), we find 
    \[
        \lim_{\substack{z \to z_+\\ z \in \C_+ \setminus \ee^{\frac{c}{2}}\D}} u(z) = -\infty.
    \]
    Choose $\epsilon > 0$ small enough such that $u(z) < -100$ on $\overline{B(z_+, \epsilon)} := \{z \ : \ |z - z_+| \leq \epsilon\}$ and let $\Omega := \C_+ \setminus (\ee^{\frac{c}{2}}\D \cup \overline{B(z_+, \epsilon)})$. It follows from the definition of $L(z)$ and \eqref{eq:psi-infnity-residue} that 
    \[
        \lim_{\substack{z \to \infty\\ z \in \C_+ \setminus \ee^{\frac{c}{2}}\D}} u(z) = 0
    \]
    and so, $u(z)$ is bounded and continuous on $\overline{\Omega} \setminus \{\pm\ee^\frac{c}{2}, -\ee^c\}$. In fact, $u(z)$ is bounded at each of these points: 
    \begin{itemize}
        \item it follows from the definition of $L(z)$ and \eqref{eq:psi-interval-decay-1} that $$\displaystyle \lim_{\substack{z \to -\ee^{c}\\ z \in \C_+ \setminus \ee^{\frac{c}{2}}\D}} u(z)= 0.$$
        
        \item Similarly, the definition of $L(z)$ and Lemma \ref{lemma:pre-image-psi}(a) imply $\im\left(L \left(\ee^{\frac{c}{2}}\right) \right) = 0 = \im \psi(\ee^{\frac{c}{2}})$, and so 
        \[
            \displaystyle \lim_{\substack{z \to \ee^{c/2}\\ z \in \C_+ \setminus \ee^{\frac{c}{2}}\D}} u(z)= 0.
        \]

        \item Finally, one can directly verify 
        \[
            \lim_{z \to -\ee^{\frac{c}{2}}} \dfrac{1 + Q(z)}{1 - Q(z)} < 0 \implies \lim_{z \to -\ee^{\frac{c}{2}}} cz \psi(z) = \ii \pi 
        \]
        where the implication follows from Lemma \ref{lemma:pre-image-psi}. Thus, using \eqref{eq:L-real}, we find  
        \[
             \lim_{\substack{z \to -\ee^{c/2}\\ z \in \C_+ \setminus \ee^{\frac{c}{2}}\D}} u(z) = -\dfrac{\re(L_{\pm}(-\ee^{\frac{c}{2}}))}{\pi} < 0. 
        \]
    \end{itemize}
    We conclude that $u(z)$ is harmonic in $\Omega$ and continuously extends to $\overline{\Omega}$. It remains to show that $u(z) \leq 0$ on $\overline{\Omega}$, which we do in pieces. 
    \begin{itemize}
        \item For $ z > \ee^{\frac{c}{2}}$ and $z < -\ee^{c}$, $\im(L(z)) = \im(cz\psi(z)) = 0$ (cf. \eqref{eq:psi-def}), and so $u(z) = 0$ there. 
        \item For $z\in \gamma_0$, $cz\psi_-(z) > 0$ (cf. \eqref{eq:zpsi-positivity-support}) and, combining this with \eqref{eq:L-im-negative}, we have $u(z) < 0$. 
        \item For $z\in \gamma \setminus \gamma_0$, $cz\psi(z) \in \ii \R_+$ (cf. \eqref{eq:re-ineq-intermediate}) and combining this with \eqref{eq:L-real}, we have 
        \[
            u(z) = -\dfrac{\re(L(z))}{|cz \psi(z)|} < 0.
        \]
        \item For $z \in (-\ee^c, -\ee^{\frac{c}{2}})$, we require explicit calculation. Using \eqref{eq:L-im-negative}, \eqref{eq:L-real}, we may write 
        \[
            L_+(z) = \log \left( \dfrac{|z + 1| |z + \ee^c|}{|z + \ee^{\frac{c}{2}(1 + \xi)}|^2} \right) - \pi \ii =:\ell(z) - \pi \ii.
        \]
        Similarly, using Lemma \ref{lemma:pre-image-psi}(c), we can write 
        \[
            cz\psi_+(z) = \log  \left|\dfrac{1 + Q(z)}{1 - Q(z)} \right| + \pi \ii =: r(z) + \pi \ii. 
        \]
        Thus, 
        \[
            \im \left(\dfrac{L_+(z)}{cz \psi_+(z)} \right) = - \pi \dfrac{ r(z) + \ell(z)}{r^2(z) + \pi^2}
        \]
        and it remains to compute the value of 
        \[
            r(z) + \ell(z) = 2 \log \left( \cosh \left(\frac{c}{2} \right)  \dfrac{t |R(z)| + |z + \ee^{\frac{c}{2}}|}{| z +  \ee^{\frac{c}{2}(1 + \xi)}|} \right).
        \]
        It suffices to show that for $x \in [-\ee^c, -\ee^{\frac{c}{2}}]$ we have (recall \eqref{eq:t-def} and that $R(x) < 0$ for $x < 0$) 
        \[
            G(x) :=  -\sinh \left(\frac{c}{2} \right) R(x) - \cosh \left(\frac{c}{2} \right) (x + \ee^{\frac{c}{2}}) + x +  \ee^{\frac{c}{2}(1 + \xi)} > 0.
        \]
        At the right endpoint, direct algebraic manipulations yield $R(-\ee^{\frac{c}{2}}) = - \dfrac{\ee^{\frac{c}{2}}}{\cosh (c/4)}$, so using the (hyperbolic) double-angle formula 
        \[
            G(-\ee^{\frac{c}{2}}) = 2\ee^{\frac{c}{2} } \sinh \left(\frac{c}{4} \right) - \ee^{\frac{c}{2} } +   \ee^{\frac{c}{2}(1 + \xi)} > \ee^{\frac{3c}{4} } - \ee^{\frac{c}{4} } - \ee^{\frac{c}{2} } +   1 = (\ee^{\frac{c}{4}} - 1)^2 (\ee^{\frac{c}{4}} + 1) > 0.
        \]
        Furthermore, 
        \[
            G'(x) = -\sinh \left(\frac{c}{2} \right) \dfrac{x - \ee^{\frac{c}{2}} \cos\theta_c}{R(x)} - \cosh \left(\frac{c}{2} \right)  + 1 < 0.
        \]
        Thus, $G(x) > 0$ for all $x \in [-\ee^c, -\ee^{\frac{c}{2}}]$.
    \end{itemize}
    
    We conclude that $u(z) \leq 0$ for all $z\in \partial \Omega$. By the maximum modulus principle, see e.g. \cite{MR503901}*{Theorem X.1.7},  we have $u(z) < 0$ for all $z\in \Omega$ and, by choice of $\epsilon>0$,  we conclude that $u(z) < 0$ for all $z \in \C_+ \setminus \ee^{\frac{c}{2}} \overline{\D}$.
\end{proof}

\begin{lemma}
    Let $(\xi, \eta) \in \mathfrak S$ and fix $c >0$. Then, the map $z \mapsto \re (\Phi_c(z))$ is decreasing as $z$ traverses $(\gamma_{out} \cup \gamma_0) \cap \C_+$ from left to right.
\end{lemma}
\begin{figure}
    \begin{subfigure}{0.49\textwidth}
        \centering
    \includegraphics[width=\linewidth]{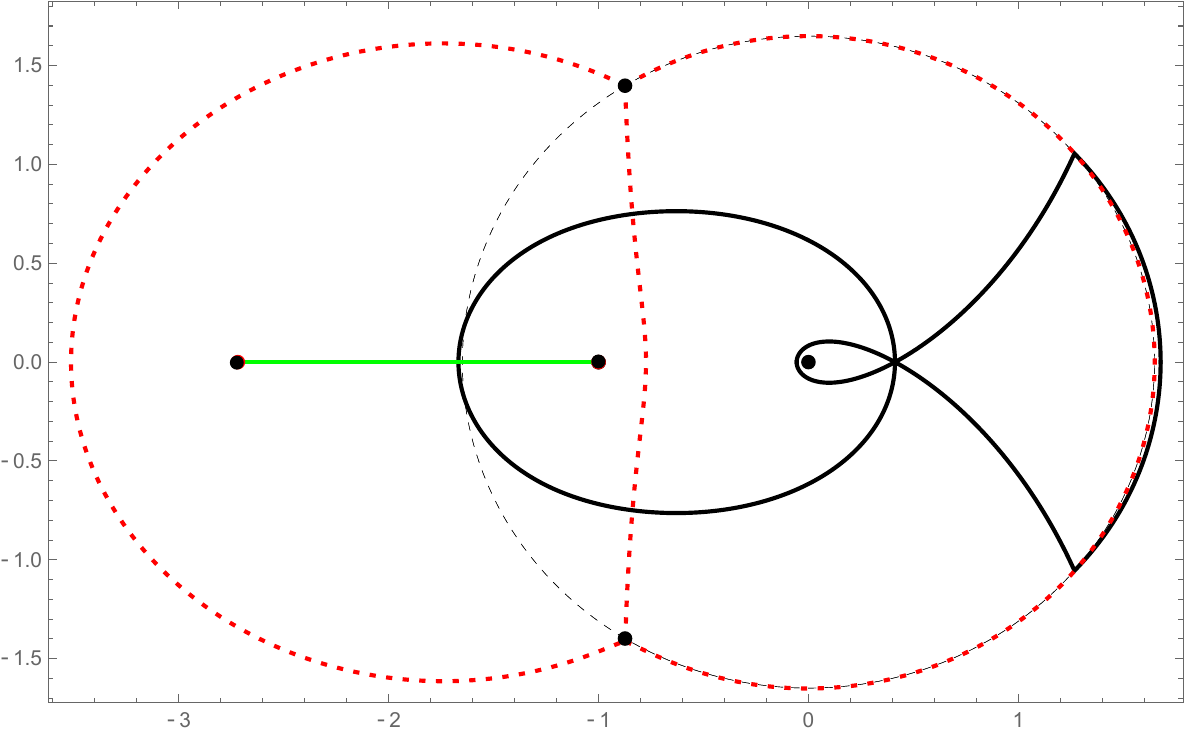}
    \caption{$c = 1, \ s = \frac{1}{4}\ee^{\frac{c}{2}}$}
    \end{subfigure}
    \begin{subfigure}{0.49\textwidth}
        \centering
    \includegraphics[width=\linewidth]{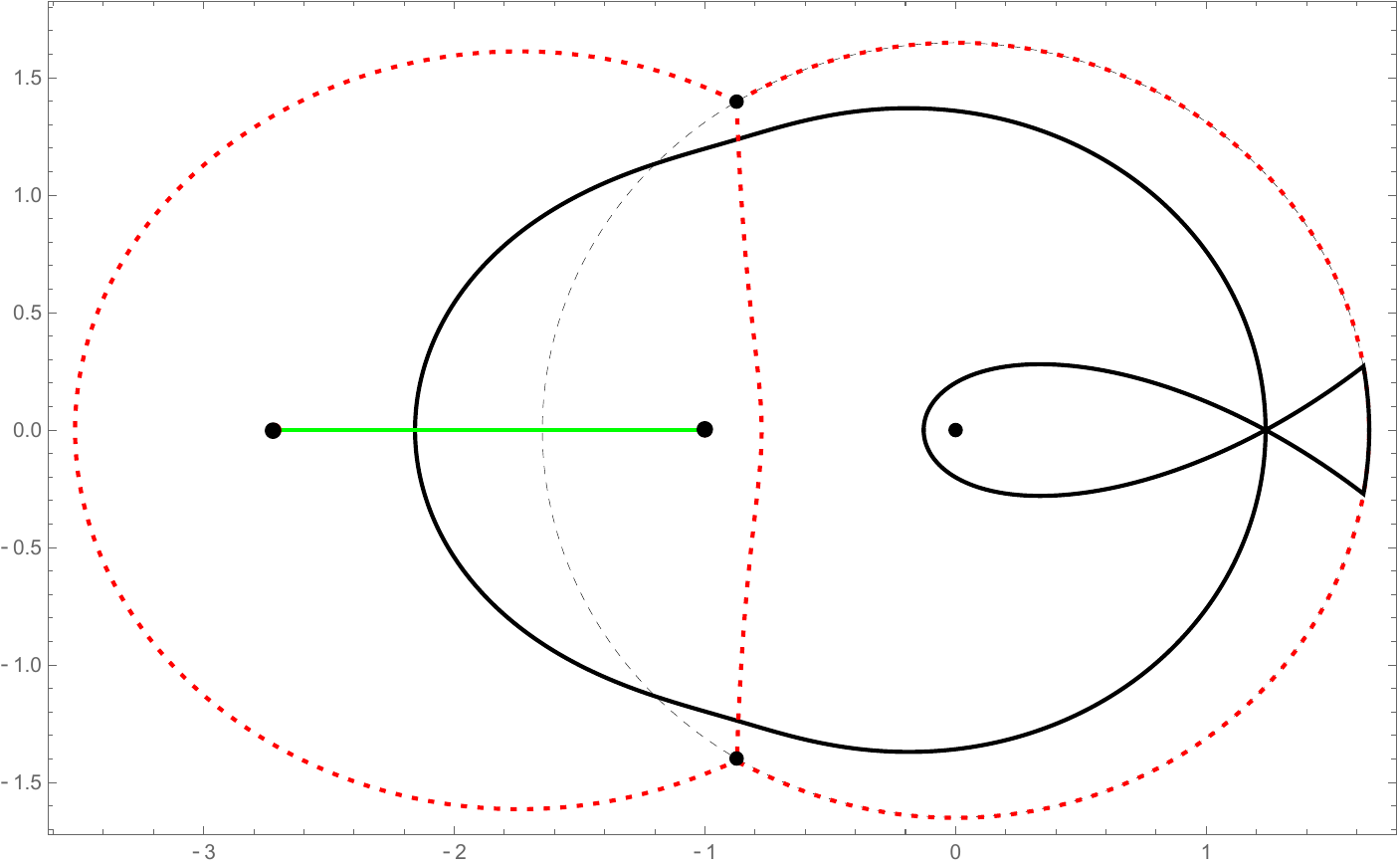}
    \caption{$c = 1, \ s = \frac{3}{4}\ee^{\frac{c}{2}}$}
    \end{subfigure}
        \begin{subfigure}{0.49\textwidth}
        \centering
    \includegraphics[width=\linewidth]{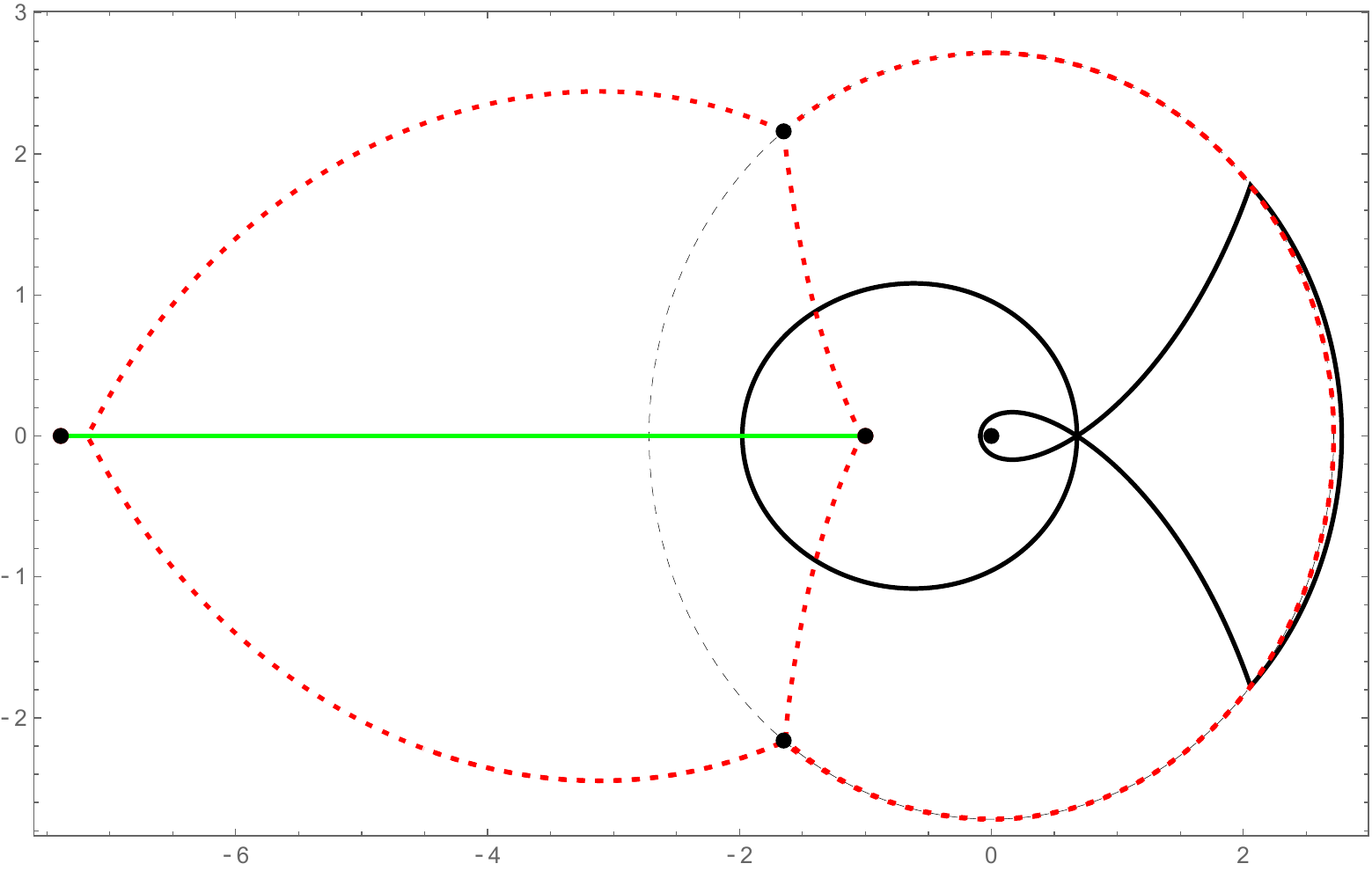}
    \caption{$c = 2, \ s = \frac{1}{4}\ee^{\frac{c}{2}}$}
    \end{subfigure}
        \begin{subfigure}{0.49\textwidth}
        \centering
    \includegraphics[width=\linewidth]{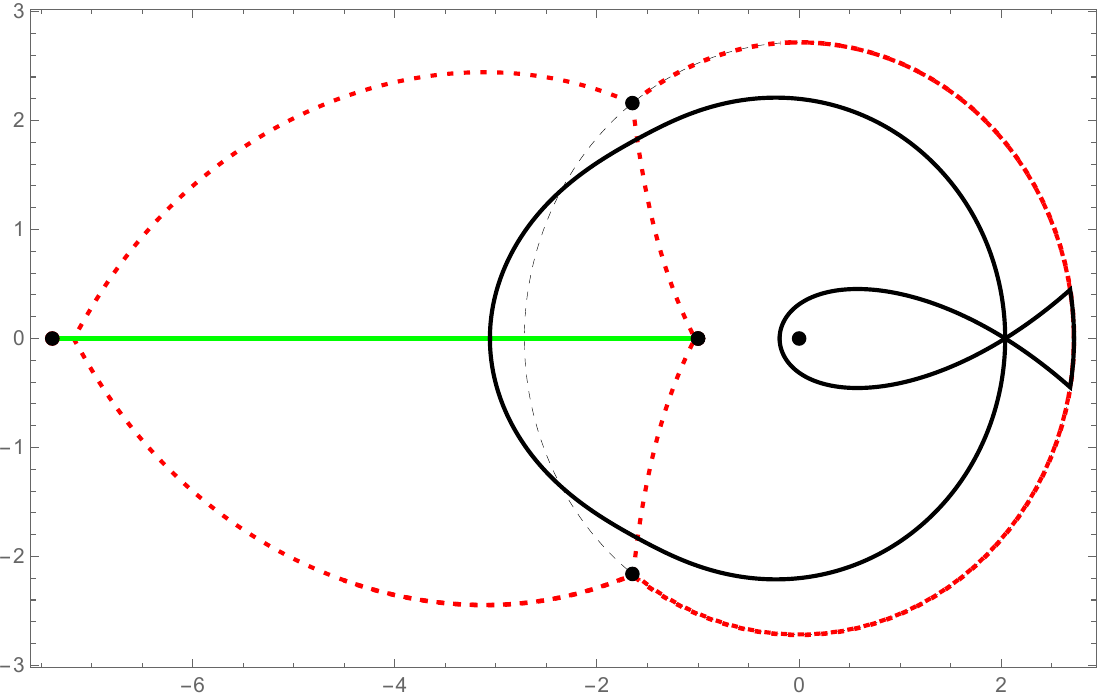}
    \caption{$c = 2, \ s = \frac{3}{4}\ee^{\frac{c}{2}}$}
    \end{subfigure}
    \begin{subfigure}{0.49\textwidth}
        \centering
    \includegraphics[width=\linewidth]{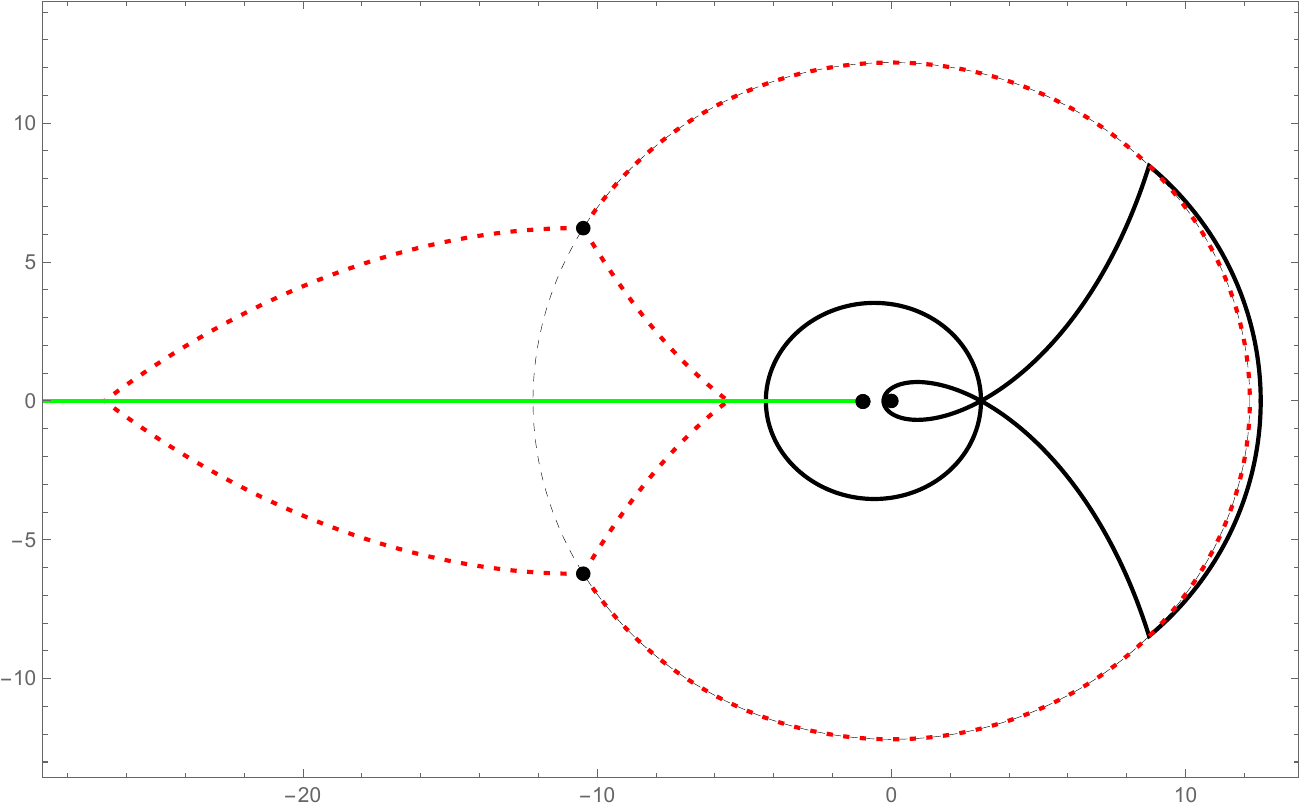}
    \caption{$c = 5, \ s = \frac{1}{4}\ee^{\frac{c}{2}}$}
    \end{subfigure}
        \begin{subfigure}{0.49\textwidth}
        \centering
    \includegraphics[width=\linewidth]{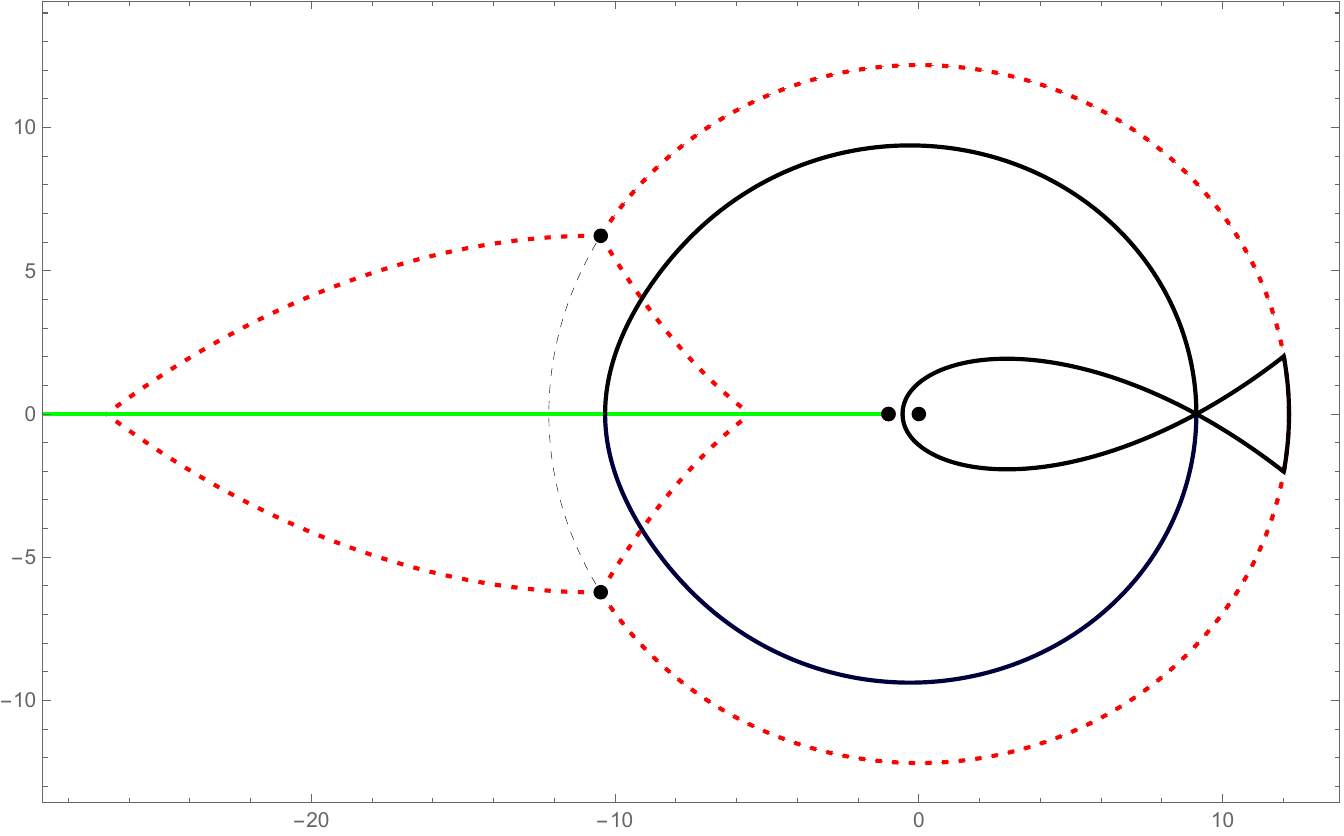}
    \caption{$c = 5, \ s = \frac{3}{4}\ee^{\frac{c}{2}}$}
    \end{subfigure}
    \caption{\centering The set $\mc N_c$ (solid) and the level set $\re(\phi_c(z)) = 0$ (red, dashed) for various choices of $c$ and $s$. The circle $\ee^{\frac c2} \T$ is indicated with a thin dashed line and the interval $[-\ee^c, -1]$ is in green.}
    \label{fig:Phi-c-level-set-2}
\end{figure}
\begin{proof}
    Using the first expression in \eqref{eq:V-def}, \eqref{eq:phase-def} and \eqref{eq:g-V-phi}, we can write 
    \begin{equation}
        \Phi_c(z) = \phi_c(z) - \int_0^1 \log \left( 1 + \dfrac{\ee^{ct}}{z}\right) \dd t + 2\int_0^{\frac{1+\xi}{2}} \log (1 + z\ee^{-c t}) \dd t - (1 + \eta) \log z 
    \end{equation}
    Taking the real part and writing 
    \[
    \log|1 + z\ee^{-ct}| = \log \left| 1 + \dfrac{\ee^{ct}}{z} \right| + \log|z| - c t
    \]
    we find 
    \begin{equation}
    \re (\Phi_c(z)) = \re(\phi_c(z)) - \int_0^1 \log \left| 1 + \dfrac{\ee^{ct}}{z}\right| \dd t + 2\int_0^{\frac{1+\xi}{2}} \log \left| 1 + \dfrac{\ee^{ct}}{z}\right|  \dd t + (\xi - \eta) \log |z| - \frac{c}{4}(1 + \xi)^2. 
    \label{eq:real-part-Phi}
    \end{equation}
    Let $z(r), r \in [0, 1]$ (depending on $c$) be an analytic parametrization of $\gamma_{out}$. Orienting the upper half of $\gamma_{out}$ towards $z_+$, there exists a function $f(r) >0$ such that
    \begin{equation}
        \dod{z}{r} = \dfrac{-\ii f(r)}{\psi_c(z(r))}, \quad f(r) >0.
        \label{eq:diff-eq}
    \end{equation}
    Thus, 
    \[
    \dod{}{r} \log |z(r)| = \re \left ( \dfrac{z'(r)}{z(r)}\right) = f(r) \im \left(\dfrac{1}{z(r) \psi_c(z(r))}
 \right) < 0
    \]
    where the last inequality follows from Lemma \ref{lemma:pre-image-psi}(c). Thus, since $(\xi, \eta) \in \mathfrak S \implies \xi > \eta$, we have that $(\xi - \eta) \log |z|$ is strictly decreasing on $\gamma_{out} \cap \C_+$ and constant on $\gamma_0$. Moving on to the remaining terms, 
    \begin{equation}
        \dod{}{r} \log \left| 1 + \dfrac{\ee^{ct}}{z(r)}\right| = - f(r) \im \left(  \dfrac{1}{z(r) \psi(z(r))} \dfrac{\ee^{ct}}{z(r) + \ee^{ct}}\right).
    \end{equation}
    Computing the two integrals in \eqref{eq:real-part-Phi}, we find 
    \begin{equation}
    \dod{}{r}  \left( - \int_0^1 \log \left| 1 + \dfrac{\ee^{ct}}{z}\right| \dd t + 2\int_0^{\frac{1+\xi}{2}} \log \left| 1 + \dfrac{\ee^{ct}}{z}\right|  \dd t\right) 
    = f(r) \im \left( \dfrac{L(z(r))}{cz(r) \psi(z(r))} \right) < 0
    \label{eq:the-inequality}
    \end{equation}
    with $L(z)$ as in \eqref{eq:L-def}, and where the inequality follows from Lemma \ref{lemma:sign}.

    It remains to show that $\re (\Phi_c(z))$ is decreasing along $\gamma_0 \cap \C_+$, but this is much simpler. Indeed, $\log |z|$ is constant there, and so we need only verify inequality \eqref{eq:the-inequality} on $\gamma_0 \cap \C_+$. It follows from \eqref{eq:zpsi-positivity-support} that $z\psi_-(z) >0$ and so we need to show 
    \begin{equation}
    \im \left(  L( z(r) ) \right) \leq 0
    \label{eq:log-inequality}
    \end{equation}
    but this was established in \eqref{eq:L-im-negative}.
\end{proof}

\section{Proof of Theorems \texorpdfstring{\ref{thm:airy-convex} and \ref{thm:airy-inflection}}{on Airy process}}
\label{sec:kernel-asymptotics}

Let $c > 0$ (cf. Proposition \ref{prop:inflection}) and let\footnote{This is slightly different notation from the Section \ref{sec:statement-of-results} where we denoted the inflection point by $(\xi_*, \eta_*)$ and other points by $(\xi, \eta)$. We make this change to avoid confusion in the calculations to follow.} $(\xi_*, \eta_*) \in \partial \mc L \cap \mathfrak S$ be any point; we will distinguish between the cases where $(\xi_*, \eta_*)$ is or is not an inflection point at the end of this section.  Throughout this section, we will denote
\begin{equation}
    x_j = N(1 + \xi_{j, N}), \quad y_j = N(1 + \eta_{j, N}), \qforq j = 1, 2,
    \label{eq:scaled-variables}
\end{equation}
and consider $(\xi_{j, N}, \eta_{j, N}) \to (\xi_*, \eta_*)$ in the following way:
\begin{equation}
    \begin{bmatrix}
        \xi_{j, N}\\ \eta_{j, N} 
    \end{bmatrix} =  \begin{bmatrix}
        \xi_*\\ \eta_*
    \end{bmatrix}  + \frac{\alpha_j}{N^{\frac23}} \mb n +  \frac{\beta_j}{N^{\frac13}} \mb n^\perp,
\end{equation}
and $\mb n, \mb n^\perp$ are vectors normal and tangent to $\partial \mc L$, respectively, at $(\xi_*, \eta_*)$, given by 
\begin{equation}
    \mb n := \begin{bmatrix} \varphi_{110} \\ \varphi_{101}\end{bmatrix}, \qandq \mb n^\perp := \begin{bmatrix} -\varphi_{101} \\ \varphi_{110}  \end{bmatrix}.
    \label{eq:frame}
\end{equation}
Our first step is to rewrite \eqref{eq:corr-kernel} in terms of $ {\mc R}_N(w, z)$. More precisely, let $\rho = \ee^{10 c}$ and, with a slight abuse of notation\footnote{The notation $\widetilde{\mc R}_N(w, z)$ already appeared in Proposition \ref{prop:mc-R-estimate}. Note that the two definitions agree in the domain bounded by $\gamma_0 \cup \gamma_{out}$ (cf. \eqref{eq:mc-R-analytic-cont}) and thus Proposition \ref{prop:mc-R-estimate} applies as stated to \eqref{eq:mc-r-tilde} in the specified domain.}, 
\begin{equation}
    \widetilde{\mc R}_N(w, z) := \dfrac{1}{2\pi \ii} \oint_{\rho \T} R_N(t, z) \prod_{j = 1}^{2N} \left( 1 + \frac{q^{j}}{t} \right) \dfrac{t - z}{t - w}\dd t. 
    \label{eq:mc-r-tilde}
\end{equation}
$\widetilde{\mc R}_N(w, z)$ is the analytic continuation of $ \mc{R}_N(w, z)$ in $w$. Indeed, we have that 
\begin{equation}
    \widetilde{\mc R}_N(w, z) = \begin{cases}
        \mc R_N(w, z), & |w| < \ee^{\frac{c}{2}} \text{ or }  |w| > \rho,\medskip \\
        \ee^{N(g(z) - g(w))}\begin{bmatrix}
        1 & - W_N(w)\, \ee^{-\nu(w)} \ee^{2N\phi(w)}
    \end{bmatrix} \mb T^{-1}(w)\mb X^{-1}(w) \mb X(z) \mb T(z) \begin{bmatrix}
        1 \\ 0 \end{bmatrix}, & \ee^{\frac{c}{2}} < |w| < \rho,
    \end{cases}
    \label{eq:mc-R-analytic-cont}
\end{equation}
where $W_N(w) := \ee^{NV(w) + \nu(w)} \prod_{j = 1}^{2N}(1 + \frac{q^j}{w})$. By \eqref{eq:X-estimate}, $\mb X^{-1}(w)\mb X(z) = \I + \Oo(N^{-\frac23})$ uniformly for $z, w \in \gamma_z \cup \gamma_w$, and $W_N(w)\,\ee^{2N\phi(w)}$ is exponentially small for $w \in \gamma_w \setminus \overline{\{|w| < \ee^{\frac c2}\}}$. Thus, we have that $\widetilde{\mc R}_N(w, z) \ee^{N(g(w) - g(z))}$ is uniformly bounded on $\gamma_z \cup \gamma_w$. 
\begin{proposition}
    Suppose the zeros of $\prod_{j = 1}^{x_2}(1+q^{-j}z)$ lie outside the region bounded by $\gamma_w \cup \rho \T$. Then, we have the identity 
    \begin{multline*}
    \dfrac{1}{(2\pi \ii)^2} \oint_{\gamma_w} \oint_{\gamma_z} \left(\prod_{j = x_2}^{2N - 1} (1 + q^{-(j + 1)}w) \right) q^{N(2N+1)}R_N(w, z) \left(\prod_{j = 0}^{x_1 - 1}(1 + q^{-(j + 1)}z) \right) \dfrac{w^{y_2}}{z^{y_1 + 1}w^{2N}} \dd z \dd w \\
    = \dfrac{1}{(2\pi \ii)^2} \oint_{\gamma_w} \oint_{\gamma_z} \widetilde{\mc R}_N(w, z) \dfrac{F(z; x_1, y_1)}{F(w; x_2, y_2)} \dfrac{1}{z} \dfrac{1}{w - z} \dd z \dd w,
    \end{multline*}
    where 
    \[
    F(z; x, y) = \dfrac{1}{z^y} \prod_{j = 1}^x (1 + q^{-j}z).
    \]
\end{proposition}
     \begin{proof}
        This is the same proof as in \cite{MR4124992}*{Proposition 7.9}. The condition on the zeros of $F(w; x_2, y_2)$ is so that we do not encounter poles of the integrand in the final contour deformation.
    \end{proof}
We have shown in Section \ref{sec:prop-g-proofs} that $\re(\phi(z))<0$ on $\gamma \setminus \gamma_0$, that is, $z = -\ee^{\frac{c}{2}}$ is in the interior of $\{z \ : \ \re(\phi(z)) <0\}$. Since we are concerned with $(\xi, \eta) \in \mathfrak{S}$, and particularly $\xi<0$, there is $N$ large enough such that $q^{x_2} < \ee^{\frac{c}{2}} $. That is to say, $\gamma_w$ may be chosen so that the assumption on the zeros of $F(w; x_2, y_2)$ is satisfied. 

Now we can start analyzing the double integral by rewriting
\begin{equation}
        F(z; x, y) = \exp \left\{ \sum_{j = 1}^{x} \log \left( 1 + z \ee^{-\frac{c}{2}\frac{j}{N}}\right) - y \log z \right\}. 
\end{equation}

\begin{proposition}
    Let $x_j$, $j = 1, 2$, be as in \eqref{eq:scaled-variables} and suppose $z \not \in \left[-\ee^{\frac{c}{2N} x_j}, -1 \right]$. Then, 
    \[
    \lim_{N \to \infty} \left[\sum_{k = 1}^{x_j} \log \left( 1 + z \ee^{-\frac{c}{2} \frac{k}{N}}\right) - N(1 + \xi_{j, N}) \int_0^1 \log \left( 1 + z \ee^{-\frac{c}{2}(1 + \xi_{j, N})u} \right) \dd u \right]= \frac{1}{2} \left( \log (1 + z \ee^{-\frac{c}{2} (1 + \xi_*)}) - \log (1 + z)\right),
    \]
    locally uniformly in $z \not \in \left[-\ee^{\frac{c}{2N} x_j}, -1 \right].$
    \label{prop:RS}
\end{proposition}

\begin{proof}
    This follows from the Euler-Maclaurin formula (see e.g. \cite{MR435697}*{Section 8.1}). Indeed, let $K\subset \C \setminus \left[-\ee^{\frac{c}{2N} x_j}, -1 \right]$ and set 
    \[
    f(x, z) := \log \left( 1 + z \ee^{-\frac{c}{2} x} \right). 
    \]
    Then, 
    \[
    \sum_{k = 1}^{x_j} f\left( \dfrac{k}{N} \right) = \int_0^{x_j} f\left( \frac{x}{N} , z \right) \dd x + \frac{1}{2} \left( f\left( \frac{x_j}{N}, z\right) - f(0, z) \right) + \frac{1}{N}\int_{0}^{x_j} B_1(x - \lfloor x \rfloor)  f_x \left( \frac{x}{N} , z \right)  \dd x, 
    \]
    where $B_1(x)$ is the first Bernoulli polynomial. Rescaling in the second integral by letting $x = x_j u$ gives 
    \[
    \frac{1}{N}\int_{0}^{x_j} B_1(x - \lfloor x \rfloor)  f_x \left( \frac{x}{N} , z \right)  \dd x = \frac{x_j}{N}\int_{0}^{1} B_1\left(x_ju - \lfloor x_j u \rfloor \right)  f_x \left( \frac{x_j}{N} u , z \right)  \dd u.
    \]
    Finally, using integration-by-parts and recalling that  $B_2'(x) = 2B_1(x)$, we find 
    \begin{multline*}
    \frac{1}{N}\int_{0}^{x_j} B_1(x - \lfloor x \rfloor)  f_x \left( \frac{x}{N} , z \right)  \dd x \\= \dfrac{x_j}{N} \left(  \frac{1}{2 x_j} f_x \left(\frac{x_j}{N} u, z \right) B_2\left(x_ju - \lfloor x_j u \rfloor \right) \biggl|_{u = 0}^{u = 1} -  \frac{x_j}{N}  \frac{1}{2x_j}\int_0^1 f_{xx}\left(\frac{x_j}{N} u, z \right) B_2 \left(x_ju - \lfloor x_j u \rfloor \right) \dd u\right).
    \end{multline*}
    Thus, we find that 
    \[
     \sum_{k = 1}^{x_j} f\left( \dfrac{k}{N} \right) = \int_0^{x_j} f\left( \frac{x}{N} , z \right) \dd x + \frac{1}{2} \left( f\left( \frac{x_j}{N}, z\right) - f(0, z) \right) + \Oo(N^{-1}).
    \]
    Recalling the definition of $x_j$ and rescaling the first integral using $x = x_j u$ yields the result. 
\end{proof}

We now split the proof into two cases. 

\subsection{Case 1: \texorpdfstring{$x_1 \leq x_2$}{}} 
\label{subsec:case-1-pf}
In this case, only the double integral in \eqref{eq:corr-kernel} appears. Using Proposition \ref{prop:RS}, we rewrite the integrand as 
\begin{equation}
\widetilde{\mc R}_N(w, z) \ee^{N(g(w) - g(z))} \dfrac{1}{z} \dfrac{1}{w - z}  \dfrac{D_N(z; x_1)}{D_N(w; x_2)} \ee^{N (\Phi(z; \xi_{1, N}, \eta_{1, N}) - \Phi(w; \xi_{2, N}, \eta_{2, N})) }
\end{equation}
where 
\[
D_N(z; x) := \exp \left\{ \sum_{k = 1}^{x} \log \left( 1 + z \ee^{-\frac{c}{2} \frac{k}{N}}\right) - x \int_0^1 \log \left( 1 + z \ee^{-\frac{c}{2}x \frac{u}{N}} \right) \dd u\right\}.
\]
We showed in the previous section that $\Phi(z; \xi, \eta)$ has a unique pair of complex-conjugate critical points which collide on the real line when $(\xi, \eta) \to (\xi_*, \eta_*) \in \partial \mc L$ at a location which we denote $s_*(\xi_*, \eta_*)$. Let $s_{N, j} = s(\xi_{j, N}, \eta_{j, N})$ denote the critical point in the upper half-plane of $\Phi(z; \xi_{j, N}, \eta_{j, N})$. Then, by continuity, we have that $s_{N, j} \to s_*$ as $N\to \infty$. Let $U_{\delta}$ be a neighborhood of $w = z = s_*$ small enough such that we can take $\gamma^{(N)}_z \cap U_\delta$ and $\gamma^{(N)}_w \cap U_\delta$ to be a union of line segments and consider $N$ large enough such that $\cup_{j = 1}^2\{s_{N, j}, \overline{s_{N, j}}\} \subset U_{\delta}$. The main contribution to the double integral will come from this neighborhood. Indeed, observe that, by the choice of contours $\gamma_z, \gamma_w$, we have that 
\[
\left| \widetilde{\mc R}_N(w, z) \ee^{N(g(w) - g(z))}   \right| < C_1
\]
for some constant $C_1 >0$ independent of $N, z, w$. Furthermore, since $\gamma_{z}, \gamma_w$ avoid the intervals $\left[-\ee^{\frac{c}{2N} x_1}, -1 \right]$ and $\left[-\ee^{\frac{c}{2N} x_2}, -1 \right]$, respectively, we have by Proposition \ref{prop:RS} that 
\[
\left| \frac{D_N(z; x_1)}{D_{N}(w; x_2)} \right| < C_2
\]
for some $C_2 >0$ independent of $N, z, w$. Putting this together with Lemma \ref{lemma:contours}, we have that, away from $U_\delta$ the double integral is exponentially small.

It remains to compute the limit of 
\[
 \dfrac{1}{(2\pi \ii)^2} \oint_{\gamma_w \cap U_\delta} \oint_{\gamma_z \cap U_\delta}\widetilde{\mc R}_N(w, z) \ee^{N(g(w) - g(z))} \dfrac{1}{z} \dfrac{1}{w - z}  \dfrac{D_N(z; x_1)}{D_N(w; x_2)} \ee^{N (\Phi(z; \xi_{ 1, N}, \eta_{1, N}) - \Phi(w; \xi_{2, N}, \eta_{2, N})) }\dd z \dd w.
\]
For $z \in U_\delta$ and $(\xi, \eta)$ in a small enough neighborhood of $(\xi_*, \eta_*)$, $\Phi(z; \xi, \eta)$ is analytic in all three variables and admits a Taylor expansion centered at $(s_*, \xi_*, \eta_*)$ which, due to the linearity of $\Phi(z;\xi,\eta)$ in $\eta$ (so that $\varphi_{ijk}=0$ for $k\ge2$ and $\varphi_{ij1}=0$ for $j\ge1$) and the definition of $s_*$ ($\varphi_{100}=\varphi_{200}=0$), takes the form
\begin{equation}
\begin{aligned}
    \Phi(z; \xi, \eta) = \varphi_{000}  &+ \varphi_{010}(\xi - \xi_*) + \varphi_{001} (\eta - \eta_*) + \frac12 \varphi_{020}(\xi - \xi_*)^2 + \frac16 \varphi_{030}(\xi - \xi_*)^3 \\
    & \hspace{200 pt} +  \frac{1}{24} \varphi_{040}(\xi - \xi_*)^4 + \Oo((\xi - \xi_*)^5) \\
    &+ (z - s_*) \left( \varphi_{110}(\xi - \xi_*) + \varphi_{101}(\eta - \eta_*) + \frac{1}{2} \varphi_{120}(\xi - \xi_*)^2 + + \frac16 \varphi_{130} (\xi - \xi_*)^3 + \Oo \left((\xi - \xi_*)^4\right) \right)\\
    &+ (z - s_*)^2 \left( \dfrac{1}{2}\varphi_{210} (\xi - \xi_*) + \frac{1}{2}\varphi_{201}(\eta - \eta_*) + \frac14 \varphi_{220}(\xi - \xi_*)^2 + \Oo \left( (\xi - \xi_*)^3 \right)  \right)\\
    & +  (z - s_*)^3 \left( \frac16 \varphi_{300} + \frac16 \varphi_{310}(\xi - \xi_*) +  \frac16 \varphi_{301}(\eta - \eta_*) + \Oo ((\xi - \xi_*)^2) \right)\\
    & + (z - s_*)^4\left( \frac{1}{24} \varphi_{400} + \Oo(\xi - \xi_*) + \Oo (\eta - \eta_*) \right) + \Oo((z - s_*)^5)
    \end{aligned}
    \label{eq:taylor}
\end{equation}

Making the change-of-variables 
\[
z = s_* + u N^{-\frac13}, \quad w = s_* + v N^{-\frac13},
\]
re-centers the integrals. Plugging this into \eqref{eq:taylor} and recalling the definitions of $\xi_{j, N}, \eta_{j, N}$, we find
\begin{equation}
\begin{aligned}
    N\Phi(s_* + uN^{-\frac13}; \xi_{j, N}, \eta_{j, N}) &= N \varphi_{000} + N^{\frac23} (\varphi_{001}\varphi_{110}-\varphi_{010}\varphi_{101}) \beta_j + N^{\frac13}\left( (\varphi_{010} \varphi_{110} + \varphi_{001}\varphi_{101})\alpha_j +\frac12 \varphi_{020} \varphi_{101}^2 \beta_j^2\right) \\
    & - \varphi_{020} \varphi_{110} \varphi_{101} \alpha_j \beta_j - \frac16 \varphi_{030} \varphi_{101}^3 \beta_j^3 + \Oo(N^{-\frac13})\\
    &+  u \left(\alpha_j \|\mb n\|^2 + \frac{1}{2} \varphi_{120} \varphi_{101}^2 \beta_j^2 \right) + u^2 \frac12 \beta_j \left( \varphi_{201} \varphi_{110} - \varphi_{210} \varphi_{101}\right) + u^3\frac{1}{6} \varphi_{300} \\
    &+ N^{-\frac{1}{3}} \left[ -u \left( \alpha_j \beta_j \varphi_{120} \varphi_{110} \varphi_{101} + \frac16 \beta_j^3 \varphi_{130} \varphi_{101}^3\right) \right.\\
 & \left. +  u^2 \left( \frac{1}{2} \alpha_j \varphi_{210} \varphi_{110} 
 + \frac12 \alpha_j \varphi_{201} \varphi_{101} + \frac14 \beta_j^2 \varphi_{101}^2 \varphi_{220}   \right) \right.\\
 & \left. + u^3 \frac{1}{6} \beta_j ( \varphi_{110}\varphi_{301} - \varphi_{101} \varphi_{310}) + u^4 \frac{1}{24} \varphi_{400}\right] + \Oo\left(N^{-\frac23} \right) 
\end{aligned}
\end{equation}
We introduce one more shift and rescaling to suppress the $u^2$ term; let 
\[
u = \left( \frac{2}{\varphi_{300}}\right)^\frac13 \left(U + \tau_1 \right), \quad v =\left( \frac{2}{\varphi_{300}}\right)^\frac13  \left(V + \tau_2\right), \quad \text{ where } \quad \tau_j \equiv \tau(\beta_j) := -\beta_j \left( \frac{\varphi_{300}}{2}\right)^\frac13 \frac{\varphi_{201} \varphi_{110} - \varphi_{210} \varphi_{101}}{\varphi_{300}}.
\]
Then, 
\begin{equation}
\begin{aligned}
    &N\Phi\left(s_* + \left( \frac{2}{\varphi_{300}}\right)^\frac13(U + \tau_j)N^{-\frac13}; \xi_{j, N}, \eta_{j, N} \right) = N \varphi_{000} + N^{\frac23} (\varphi_{001} \varphi_{110} - \varphi_{010}\varphi_{101}) \beta_j  \\ &+N^{\frac13}\left( (\varphi_{010} \varphi_{110} + \varphi_{001}\varphi_{101})\alpha_j +\frac12 \varphi_{020} \varphi_{101}^2 \beta_j^2\right) 
     - \varphi_{020} \varphi_{110} \varphi_{101} \alpha_j \beta_j - \frac16 \varphi_{030} \varphi_{101}^3 \beta_j^3 + \Oo(N^{-\frac13})\\
    &+  \left( \frac13 \beta_j^3 \frac{(\varphi_{201} \varphi_{110} - \varphi_{210} \varphi_{101})^3}{\varphi_{300}^2} 
 - \beta_j \frac{\varphi_{201} \varphi_{110} - \varphi_{210} \varphi_{101}}{\varphi_{300}} (\alpha_j \|\mb n\|^2 + \frac12 \varphi_{120} \varphi_{101}^2 \beta_j^2) \right) \\
     &+U \left( \frac{2}{\varphi_{300}}\right)^\frac13 \left(\alpha_j \|\mb n \|^2 + \beta_j^2 \left( \frac12 \varphi_{120} \varphi_{101}^2 - \frac{1}{2} \frac{(\varphi_{201} \varphi_{110} - \varphi_{210} \varphi_{101})^2}{\varphi_{300}}\right) \right) + \frac{1}{3} U^3 \\
    &+ N^{-\frac{1}{3}} \left[ -(U + \tau_j)  \left( \frac{2}{\varphi_{300}}\right)^\frac13 \left(\alpha_j \beta_j \varphi_{120} \varphi_{110} \varphi_{101} + \frac16 \beta_j^3 \varphi_{130} \varphi_{101}^3\right) \right. \\
  &\left.+  (U + \tau_j)^2 \left( \frac{2}{\varphi_{300}}\right)^\frac23\left( \frac{1}{2} \alpha_j \varphi_{210} \varphi_{110} 
 + \frac12 \alpha_j \varphi_{201} \varphi_{101} + \frac14 \beta_j^2 \varphi_{101}^2 \varphi_{220} 
  \right) + (U + \tau_j)^3 \frac{1}{3} \beta_j \frac{ \varphi_{110}\varphi_{301} - \varphi_{101} \varphi_{310}}{\varphi_{300}} \right. \\
  &\left. + (U + \tau_j)^4 \left( \frac{2}{\varphi_{300}}\right)^\frac43 \frac{1}{24}\varphi_{400}\right] + \Oo\left(N^{-\frac23} \right). 
  \label{eq:taylor-2}
\end{aligned}
\end{equation}
To summarize the above calculation, we have the following identity: 
\begin{multline}
    \lim_{N \to \infty} \dfrac{ \exp \left\{ - \left(N^{\frac{2}{3}} k_5 \beta_1 + N^{\frac{1}{3}} (k_3 \alpha_1 + k_4 \beta_1^2)\right) \right\} }{ \exp \left\{ - \left(N^{\frac{2}{3}} k_5 \beta_2 + N^{\frac{1}{3}} (k_3 \alpha_2 + k_4 \beta_2^2)\right)  \right\}}\dfrac{\exp\left\{N\Phi(s_* + (U + \tau_1)N^{-\frac13}\left( \frac{2}{\varphi_{300}}\right)^\frac13; \xi_{1, N}, \eta_{1, N})\right\} }{\exp\left\{ N\Phi(s_* + (V + \tau_2)N^{-\frac13}\left( \frac{2}{\varphi_{300}}\right)^\frac13; \xi_{2, N}, \eta_{2, N})\right\}} \\
    = \dfrac{\exp\left\{ k_1 \beta_1^3 - k_2 \alpha_1 \beta_1 \right\} }{\exp\left\{ k_1 \beta_2^3 - k_2 \alpha_2 \beta_2 \right\}}\dfrac{\exp\left\{\frac13 U^3 - r(\alpha_1, \beta_1) U  \right\}}{\exp\left\{ \frac13 V^3 - r(\alpha_2, \beta_2) V  \right\}}, 
    \label{eq:limit-identity}
\end{multline}
where the functions $r(\alpha, \beta), \tau(\beta)$ as in \eqref{eq:r-fun}, \eqref{eq:tau-fun}, respectively, and the geometric constants $k_i, i = 1,..., 6$ are given by 
 \begin{equation}
    \begin{aligned}
        k_1 &:= \frac{1}{3}  \frac{(\varphi_{201} \varphi_{110} - \varphi_{210} \varphi_{101})^3}{\varphi_{300}^2} - \frac12 \frac{\varphi_{201} \varphi_{110} - \varphi_{210} \varphi_{101}}{\varphi_{300}}  \varphi_{120} \varphi_{101}^2 - \frac16 \varphi_{030} \varphi_{101}^3, \\
        k_2 &:= \| \mb n \|^2 \frac{\varphi_{201} \varphi_{110} - \varphi_{210} \varphi_{101}}{\varphi_{300}} + \varphi_{020} \varphi_{110} \varphi_{101}, \\
        k_3 &:= \varphi_{010} \varphi_{110} + \varphi_{001}\varphi_{101}, \qquad k_4 :=  \frac12 \varphi_{020} \varphi_{101}^2, \\
        k_5 &:= \varphi_{001} \varphi_{110} - \varphi_{010}\varphi_{101}, \\
        k_6 &:= \left( \frac{2}{\varphi_{300}}\right)^\frac13.
    \end{aligned}
    \label{eq:k-values}
    \end{equation}
    The remaining factors in the integrand are (the $N^{-\frac23}k_6^2$ factor is from the change-of-variables)
\begin{multline}
N^{-\frac23}k_6^2\widetilde{\mc R}_N\left(s_* + k_6(V + \tau_2)N^{-\frac13}, s_* + k_6(U + \tau_1)N^{-\frac13} \right) \ee^{N\left(g(s_* + k_6(V + \tau_2)N^{-\frac13}) - g(s_* + k_6(U+ \tau_1)N^{-\frac13})\right)} \\
\times \dfrac{1}{s_* + k_6(U + \tau_1)N^{-\frac13}} \dfrac{1}{s_* + k_6(V + \tau_2)N^{-\frac13} -( s_* + k_6(U + \tau_1)N^{-\frac13})}  \dfrac{D_N(s_* + k_6(U + \tau_1)N^{-\frac13}; x_1)}{D_N(s_* + k_6(V + \tau_2)N^{-\frac13}; x_2)} \\
 = -N^{-\frac13}k_6^2  \widetilde{\mc R}_N\left(s_* + k_6(V + \tau_2)N^{-\frac13}, s_* + k_6(U + \tau_1)N^{-\frac13} \right) \ee^{N\left(g(s_* + k_6(V + \tau_2)N^{-\frac13}) - g(s_* + k_6(U + \tau_1)N^{-\frac13})\right)} \\
 \times \dfrac{1}{s_* + k_6(U + \tau_1)N^{-\frac13}} \dfrac{k_6^{-1}}{U + \tau_1 - (V + \tau_2) }  \dfrac{D_N(s_* + k_6(U + \tau_1)N^{-\frac13}; x_1)}{D_N(s_* + k_6(V + \tau_2)N^{-\frac13}; x_2)}. 
\end{multline}
Recalling Proposition \ref{prop:RS}, we can compute 
\begin{equation}
    \lim_{N \to \infty} D_N \left(s_* + k_6 (U + \tau_j)N^{-\frac13}; x_j \right) = \exp \left\{ \frac{1}{2} \left( \log \left( 1 + s_* \ee^{-\frac{c}{2} (1 + \xi_*)}\right) - \log (1 + s_*) \right) \right\}. 
    \label{eq:RS-limit}
\end{equation}
Furthermore, it follows from Proposition \ref{prop:mc-R-estimate} that 
\begin{equation}
    \lim_{N \to \infty} \widetilde{\mc R}_N\left(s_* + k_6(V+ \tau_2)N^{-\frac13}, s_* + k_6(U + \tau_1)N^{-\frac13} \right) \ee^{N\left(g(s_* + k_6(V + \tau_2)N^{-\frac13}) - g(s_* + k_6(U + \tau_1)N^{-\frac13})\right)} = 1.
        \label{eq:rep-kernel-limit}
\end{equation}
To finish our calculation, we must now make the choice of contour of integration in $U_\delta$ precise. To do so, observe that since $\varphi_{101} < 0$, $\varphi_{201} \varphi_{110} - \varphi_{210} \varphi_{101} > 0$, and $\varphi_{300} <0$ (see \eqref{eq:phi-derivatives-expressions}, Lemma \ref{lemma:det}, and Lemma \ref{lemma:phi300-negative}, respectively), we have 
\begin{equation}
    x_1 \leq x_2 \Leftrightarrow \beta_1 \leq \beta_2 \Leftrightarrow \tau_1 \geq \tau_2.
    \label{eq:tau-ineq}
\end{equation}
Let $\sigma, \sigma' >0$ and deform $\gamma_z \cap U_\delta$ to a subset of the straight line segment (recall \eqref{eq:k-values})
\[
\gamma_z \cap U_\delta  \mapsto \{z \in U_\delta \ : \ \re(z) = s_* + N^{-\frac{1}{3}}k_6 (\sigma+\tau_1) \} \qandq  \gamma_w \cap U_\delta \mapsto \{w \in U_\delta \ : \ \re(w) = s_* - N^{-\frac{1}{3}} k_6 (\sigma'-\tau_2) \}.
\]
Then, it follows from the dominated convergence theorem that
\begin{multline}
    \lim_{N \to \infty}  -s_* k_6^{-1} N^{\frac13} \dfrac{ \ee^{- \left(N^{\frac{2}{3}} k_5 \beta_1 + N^{\frac{1}{3}} (k_3 \alpha_1 + k_4 \beta_1^2)\right) } }{ \ee^{- \left(N^{\frac{2}{3}} k_5 \beta_2 + N^{\frac{1}{3}} (k_3 \alpha_2 + k_4 \beta_2^2)\right) } }\dfrac{\ee^{ k_1 \beta_2^3 - k_2 \alpha_2 \beta_2 }}{\ee^{ k_1 \beta_1^3 - k_2 \alpha_1 \beta_1 } } \\
    \times \dfrac{1}{(2\pi \ii)^2} \int_{\gamma_w \cap U_\delta} \int_{\gamma_z \cap U_\delta}\widetilde{\mc R}_N(w, z) \ee^{N(g(w) - g(z))} \dfrac{1}{z} \dfrac{1}{w - z}  \dfrac{D_N(z; x_1)}{D_N(w; x_2)} \ee^{N (\Phi(z; \xi_{1, N}, \eta_{1, N}) - \Phi(w; \xi_{2, N}, \eta_{2, N})) }\dd z \dd w \\
    = \dfrac{1}{(2\pi \ii)^2} \int_{C_U} \int_{C_V} \dfrac{\exp\left\{\frac13 U^3 - r(\alpha_1, \beta_1) U  \right\}}{\exp\left\{ \frac13 V^3 - r(\alpha_2, \beta_2) V  \right\}} \dfrac{1}{U + \tau_1 - (V + \tau_2) }  \dd U \dd V,
    \label{eq:double-int-limit}
\end{multline}
where 
\[
C_U = \{U \in \C \ : \ \re(U) = \sigma\} \qandq 
C_V = \{V \in \C \ : \ \re(V) = - \sigma' \}.
\]
The final double integral can be recognized as the Airy kernel (this particular form has appeared in, e.g., \cite{MR2018275}*{Section 2.2}). Indeed, it follows from \eqref{eq:tau-ineq} that, on these contours, $\re(U + \tau_1 - V - \tau_2) > 0$ and thus
\[
\int_0^\infty \ee^{-t (U + \tau_1 - V - \tau_2)} \dd t = \dfrac{1}{U + \tau_1 - (V + \tau_2)}.
\]
Combining this with the classical integral formula 
\begin{equation}
    \mathsf{Ai} (z) = \frac{1}{2\pi \ii }  \int \displaylimits_{\re(t) = \text{const.} >0} \ee^{\frac{1}{3}t^3 - zt } \dd t =  \frac{1}{2\pi \ii }\int \displaylimits_{\re(t) = \text{const.} < 0} \ee^{- \frac{1}{3}t^3 + zt } \dd t,
    \label{eq:airy-formula}
\end{equation}
we arrive at \eqref{eq:kernel-limit-non-inflection} in the case $x_1 \leq  x_2$. 

\subsubsection{The case of an inflection point} 
\label{sububsec:inflection-pf} If we impose the extra assumption that $(\xi_*, \eta_*)$ is an inflection point, then it follows from \eqref{eq:curvature} and Lemmas \ref{lemma:det}, \ref{lemma:phi300-negative} that 
\[
(\varphi_{110} \varphi_{201} - \varphi_{101} \varphi_{210})^2 - \varphi_{101}^2 \varphi_{120} \varphi_{300} = 0.
\]
In this case, $r(\alpha, \beta)$ in \eqref{eq:limit-identity} no longer depends on $\beta$. To arrive at Theorem \ref{thm:airy-inflection}, let
\begin{equation}
    \beta_j = \widetilde{\beta}_j + \omega N^{\delta}, \quad \delta <\frac{1}{9}.
    \label{eq:beta-inflection}
\end{equation}
Then every term of total degree $4$ in \eqref{eq:taylor} contributes $N \cdot \Oo(N^{4\delta - \frac43}) = \Oo(N^{4\delta - \frac13})$, while terms of degree $\geq 5$ are $\Oo(N^{5\delta - \frac23}) = o(1)$. Writing $z - s_* = k_6(U + \widetilde\tau_j + \tau(\omega) N^{\delta})N^{-\frac13}$ with $\widetilde\tau_j$ defined in \eqref{eq:tau-tilde}, the leading $N^{4\delta - \frac13}$ part of each degree-$4$ term depends only on $\omega$ (not on $j$, $U$ or $V$) and cancels between the numerator and denominator of \eqref{eq:limit-identity}; the remainder is $\Oo(N^{3\delta - \frac13}) = o(1)$ when $\delta < \frac19$. Thus, making the same sequence of algebraic manipulations yields \eqref{eq:limit-identity} with $\beta_j$ as in \eqref{eq:beta-inflection} and 
\[
r(\alpha, \beta) \equiv r(\alpha) = -\left( \dfrac{2}{\varphi_{300}}\right)^{\frac13} \|\mb n \|^2 \alpha.
\]
In place of \eqref{eq:double-int-limit}, we find the same identity but with 
\begin{equation}
\tau_j \mapsto \widetilde{\tau}_j := -\widetilde{\beta}_j \left( \frac{\varphi_{300}}{2}\right)^\frac13\frac{\varphi_{201} \varphi_{110} - \varphi_{210} \varphi_{101}}{\varphi_{300}}.
\label{eq:tau-tilde}
\end{equation}
The rest of the proof is identical. 
\begin{remark}
    One can observe the necessity of taking $\delta < \frac19$ in \eqref{eq:taylor-2}. Indeed, while in the present case the coefficient of $\beta^2_j$ in the fourth line vanishes, the omitted terms from the coefficient of $U$ contain a $N^{-1/3}\beta_j^3$ (see the last line of \eqref{eq:taylor-2}). It is expected that at a higher order inflection point this coefficient of $\beta_j^3$ would vanish. Plugging in $\beta_j = \tilde{\beta}_j + \omega N^{{1}/{9}}$ and computing the next omitted term in the coefficient of $U$ reveals terms proportional to $\omega^2 N^{-1/9}$. Thus, one can expect to let $\omega = o(N^{{1}/{18}})$ or, in other words, take $\delta < \frac{3}{18}$ in Theorem \ref{thm:airy-inflection}.
    \label{remark:exponents}
\end{remark}

\subsection{Case 2: \texorpdfstring{$x_1 > x_2$}{}}

While the analysis of the double integral in this case is basically unchanged, we need to obtain asymptotics of the single integral appearing in \eqref{eq:corr-kernel} in this case. The idea will be to observe that the main contribution of this integral comes from a neighborhood of $z = s_*$ and can be interpreted as a residue contribution resulting from contour deformations in the double integral. We make this more precise now. First, rewriting 
\[
\prod_{j = x_2}^{x_1 - 1} \left( 1 + q^{-(j+1)}z \right) = \exp\left \{ \sum_{j = x_2 + 1}^{x_1} \log \left( 1 + z \ee^{-\frac{c}{2} \frac{j}{N}} \right) \right\} = \exp \left \{ \sum_{j = 1}^{x_1} \log \left( 1 + z \ee^{-\frac{c}{2} \frac{j}{N}} \right) -  \sum_{j = 1}^{x_2} \log \left( 1 + z \ee^{-\frac{c}{2} \frac{j}{N}} \right) \right\},
\]
and applying Proposition \ref{prop:RS} to the two sums, we find
\[
\lim_{N\to \infty} \left[\sum_{j = x_2+1}^{x_1} \log \left( 1 + z \ee^{-\frac{c}{2} \frac{j}{N}}\right) - (x_1 - x_2) \int_0^1 \log \left( 1 + z \ee^{-\frac{c}{2} \left(\frac{u(x_1 - x_2)}{N} + \frac{x_2}{N}\right)} \right) \dd u \right] = 0.
\]
With this, we can rewrite the integrand in the first integral of \eqref{eq:corr-kernel} as
\begin{multline}
    \exp \left \{ \sum_{j = x_2+1}^{x_1} \log \left( 1 + z \ee^{-\frac{c}{2} \frac{j}{N}}\right) - (x_1 - x_2) \int_0^1 \log \left( 1 + z \ee^{-\frac{c}{2} \left( \frac{u(x_1 - x_2)}{N} + \frac{x_2}{N}\right)} \right) \dd u \right\} \\
    \cdot \exp \left \{ (x_1 - x_2) \int_0^1 \log \left( 1 + z \ee^{-\frac{c}{2} \left(\frac{u(x_1 - x_2)}{N} + \frac{x_2}{N}\right)} \right) \dd u - (y_1 - y_2 + 1) \log z \right\}
\end{multline}
Now, recall \eqref{eq:scaled-variables} and note that 
\begin{align*}
x_1 - x_2 &= -\varphi_{101} N^{\frac{2}{3}}(\beta_1 - \beta_2) + \varphi_{110} N^{\frac{1}{3}}(\alpha_1 - \alpha_2), \\
y_1 - y_2 &= \varphi_{110} N^{\frac{2}{3}}(\beta_1 - \beta_2) + \varphi_{101} N^{\frac{1}{3}}(\alpha_1 - \alpha_2).
\end{align*}
In this regime, we have 
\[
\lim_{N \to \infty} \int_0^1 \log \left( 1 + z \ee^{-\frac{c}{2} \left(\frac{u(x_1 - x_2)}{N} + \frac{x_2}{N}\right)} \right) \dd u = \log \left( 1 + z \ee^{-\frac{c}{2} (1 + \xi_*)}\right)
\]
and 
\[
\lim_{N\to \infty} \dfrac{y_1 - y_2 + 1}{x_1 - x_2} = -\frac{\varphi_{110}}{\varphi_{101}}
\]
so we should consider the saddle points of the function 
\[
\Psi(z) := \log \left( 1 + z \ee^{-\frac{c}{2} (1 + \xi_*)}\right) + \frac{\varphi_{110}}{\varphi_{101}} \log z .
\]
We can compute the critical points of $\Psi(z)$ easily: 
\[
\dod{\Psi}{z} = \dfrac{\ee^{-\frac{c}{2} (1 + \xi_*)}}{1 + z \ee^{-\frac{c}{2} (1 + \xi_*)}} + \frac{\varphi_{110}}{\varphi_{101}} \dfrac{1}{z} =0 \implies z = -\varphi_{110} \dfrac{\ee^{\frac{c}{2} (1 + \xi_*)}}{\varphi_{101} + \varphi_{110}}.
\]
Using the identities \eqref{eq:phi-derivatives-expressions} we find that the critical point of $\Psi(z)$ is exactly at $z = s_*$. An example of the set $\re(\Psi(z) - \Psi(s_*)) = 0$ is shown in Figure \ref{fig:psi-trajectories}. 
\begin{figure}
    \centering
    \includegraphics[width=0.4\linewidth]{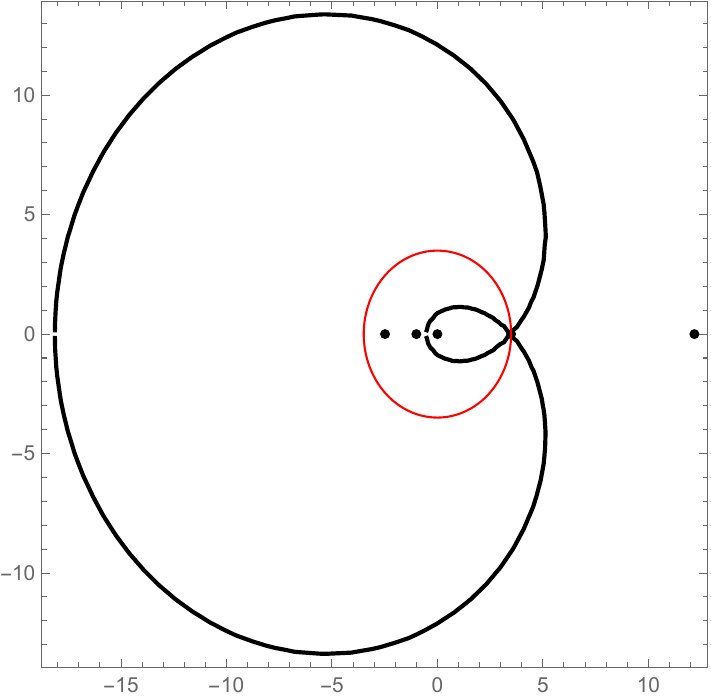}
    \put(-50,100){$s_*$}
    \put(-15,90){$\ee^\frac{c}{2}$}
    \put(-120,90){$\ee^{\frac{c}{2}(1 + \xi_*)}$}
    \put(-79,125){{\color{red}$<$}}
    \caption{\centering The set $\re(\Psi(z) - \Psi(s_*)) = 0$ when $c = 5$. Regions marked $+/-$ correspond to regions where $\re(\Psi(z) - \Psi(s_*))$ is positive/negative. The contour of integration is shown in red. }
    \label{fig:psi-trajectories}
\end{figure}
Arguments similar to those from Section \ref{subsec:case-1-pf} show that the contribution away from $z = s_*$ is exponentially small. \\

Working in the neighborhood $U_\delta$, we make the following, simplifying observation. Recall that in this section we assume $\tau_1 < \tau_2$. Let $\sigma, \sigma' >0$ be such that $\sigma + \sigma'+ \tau_1 - \tau_2  < 0$. Let $\gamma_z^{\swap}, \gamma_w^{\swap} $ be contours which agree with $\gamma_z,\gamma_w$ in $\C \setminus U_\delta$ and, within $U_\delta$, are line segments given by 
\[
\gamma_z^{\swap} \cap U_\delta  :=  \{z \in U_\delta \ : \ \re(z) = s_* + N^{-\frac{1}{3}}k_6 (\sigma+\tau_1) \} \qandq  \gamma_w^{\swap} \cap U_\delta := \{w \in U_\delta \ : \ \re(w) = s_* - N^{-\frac{1}{3}} k_6 (\sigma'-\tau_2) \}.
\]
Since $\sigma + \tau_1 < \tau_2 - \sigma'$, the contours pass through one another during the deformation and we pick up a residue term. The observation is that this residue term is, up to exponentially small terms, the negative of the single integral \eqref{eq:corr-kernel}. That is, for some $k > 0$, 
\begin{multline}
\dfrac{1}{(2\pi \ii)^2} \oint_{\gamma_w} \oint_{\gamma_z} \widetilde{\mc R}_N(w, z) \dfrac{F(z; x_1, y_1)}{F(w; x_2, y_2)} \dfrac{1}{z} \dfrac{1}{w - z} \dd z \dd w \\= \dfrac{1}{(2\pi \ii)^2} \oint_{\gamma_w^{\swap}} \oint_{\gamma_z^{\swap}} \widetilde{\mc R}_N(w, z) \dfrac{F(z; x_1, y_1)}{F(w; x_2, y_2)} \dfrac{1}{z} \dfrac{1}{w - z} \dd z \dd w + \dfrac{1}{2\pi \ii} \int_{\gamma} \widetilde{\mc R}_N(z, z) \dfrac{F(z; x_1, y_1)}{F(z; x_2, y_2)} \dfrac{1}{z} \dd z + \Oo(\ee^{-kN})\\
= \dfrac{1}{(2\pi \ii)^2} \oint_{\gamma_w^{\swap}} \oint_{\gamma_z^{\swap}} \widetilde{\mc R}_N(w, z) \dfrac{F(z; x_1, y_1)}{F(w; x_2, y_2)} \dfrac{1}{z} \dfrac{1}{w - z} \dd z \dd w + \dfrac{1}{2\pi \ii} \int_{\gamma}  \prod_{j = x_2}^{x_1 - 1} (1 + q^{-(j+1)}z)  \dfrac{1}{z^{y_1 - y_2 + 1}} \dd z + \Oo(\ee^{-kN}),
\end{multline}
where the last equality follows from the fact that, for $z \in \{z \ : \ |z| < \ee^{\frac{c}{2}}\}$, we have $\widetilde{\mc R}_N(z, z) = {\mc R}_N(z, z) = 1$ which is clear from the definitions, e.g., \eqref{eq:mc-r-rh}. With this, we have now reduced our analysis to the asymptotic analysis of the double integral with these contours, but this is identical to the calculation in \eqref{eq:double-int-limit}. By dominated convergence theorem, we again have \eqref{eq:double-int-limit} where $C_U, C_V$ have the same definition but with $\sigma, \sigma'$ chosen as in the beginning of this paragraph. Now, since $\re(U + \tau_1 - V -\tau_2) <0$, we use the identity 
\[
-\int_{-\infty}^0 \ee^{-t(U + \tau_1 - V -\tau_2)} \dd t = \dfrac{1}{U + \tau_1 - V - \tau_2},
\]
and the same classical formulas for the Airy functions \eqref{eq:airy-formula} to arrive at \eqref{eq:kernel-limit-non-inflection} in the case $x_1 > x_2$. This finishes the proof of Theorem \ref{thm:airy-convex}. From this and an argument identical to Section \ref{sububsec:inflection-pf}, we deduce Theorem \ref{thm:airy-inflection}. 

\section{AI Disclosure}
The authors made use of Anthropic's Claude - Opus 5.5 in the preparation of this manuscript. Precisely, a preprint of this manuscript, prepared without the use of any AI tools, appeared on arXiv \cite{qvol}*{Version 1} in which inequality \eqref{eq:the-inequality} was not proven. Claude was prompted to prove this and the result is the argument in Lemma \ref{lemma:sign}. In its proof, Claude made use of expression \eqref{eq:psi-def} as opposed to the integral formula \eqref{eq:psi-integral} which the authors had found and used previously. Claude was also used to simplify various expressions and arguments in the manuscript. The authors wrote the text, reviewed and revised the manuscript and take full responsibility for its content.

\bibliographystyle{amsrefs}
\bibliography{bibliography}
\end{document}